\documentclass[draft,10pt]{article}%
\usepackage{amscd}
\usepackage{amsmath}
\usepackage{amsfonts}
\usepackage{amssymb}
\usepackage{graphicx}
\usepackage[margin=0.5in]{geometry}%
\providecommand{\U}[1]{\protect\rule{.1in}{.1in}}
\providecommand{\U}[1]{\protect\rule{.1in}{.1in}}
\providecommand{\U}[1]{\protect\rule{.1in}{.1in}}

\numberwithin{equation}{section}
\newtheorem{theorem}{Theorem}[section]

\newtheorem{lemma}[theorem]{Lemma}
\newtheorem{proposition}[theorem]{Proposition}
\newtheorem{condition}[theorem]{Condition}
\newtheorem{rem}[theorem]{Remark}
\newtheorem{remark}[theorem]{Remark}

\newenvironment{proof}[1][Proof]{\noindent\textbf{#1.} }{\ \rule{0.5em}{0.5em}}

\def\p2{\mathcal A_{\Phi,2\pi}(B)}
\def\0p2{\mathcal A_{\Phi,2\pi}(0)}
\def\sp2{\mathcal A_{\Phi,2\pi,\hbox{\rm SR}}(B)}
\def\beq{\begin{equation}}
\def\ene{\end{equation}}

\def\qed{\ifhmode\unskip\nobreak\fi\ifmmode\ifinner
\else\hskip5pt\fi\fi\hbox{\hskip5pt\vrule width4pt height6pt
depth1.5pt\hskip1pt}}

\def\+out{x^{\rm out}}
\begin{document}

\title{Time delay for the Dirac equation.\thanks{AMS Subject Classifications: 35Q40;
35P25; 35Q41; 81U99} \thanks{Research partially supported by project PAPIIT-DGAPA
UNAM IN102215.}}
\author{Ivan Naumkin\thanks{ Electronic Mail: ivannaumkinkaikin@gmail.com} and Ricardo
Weder\thanks{Fellow, Sistema Nacional de Investigadores. Electronic mail:
weder@unam.mx }\\Departamento de F\'{\i}sica Matem\'{a}tica,\\Instituto de Investigaciones en Matem\'aticas Aplicadas y en Sistemas. \\Universidad Nacional Aut\'onoma de M\'exico.\\Apartado Postal 20-126, M\'exico DF 01000, M\'exico.}
\date{}
\maketitle




\centerline{{\bf Abstract}}

We consider time delay for the Dirac equation. A new method to calculate the
asymptotics of the expectation values of the operator $\int\limits_{0}%
^{\infty}e^{iH_{0}t}\zeta\left(  \frac{\left\vert x\right\vert }{R}\right)
e^{-iH_{0}t}dt,$ as $R\rightarrow\infty,$ is presented. Here $H_{0}$ is the
free Dirac operator and $\zeta\left(  t\right)  $ is such that $\zeta\left(
t\right)  =1$ for $0\leq t\leq1$ and $\zeta\left(  t\right)  =0$ for $t>1.$
This approach allows us to obtain the time delay operator $\delta
\mathcal{T}\left(  f\right)  $ for initial states $f$ in $\mathcal{H}%
_{2}^{3/2+\varepsilon}\left(  \mathbb{R}^{3};\mathbb{C}^{4}\right)  ,$
$\varepsilon>0,$ the Sobolev space of order $3/2+\varepsilon$ and weight $2.$
The relation between the time delay operator $\delta\mathcal{T}\left(
f\right)  $ and the Eisenbud-Wigner time delay operator is given. Also, the
relation between the averaged time delay and the spectral shift function is presented.

\section{Introduction.}

In the present paper we consider time delay for the quantum scattering pair
$\{H_{0},H\}$, where the free Dirac operator $H_{0}$ is given by
\begin{equation}
H_{0}=-i\alpha\cdot\nabla+m\alpha_{4}, \label{basicnotions1}%
\end{equation}
with $m$ - the mass of the particle, $\alpha=(\alpha_{1},\alpha_{2},\alpha
_{3})$ and $\alpha_{j},$ $j=1,2,3,4,$ are $4\times4$ Hermitian matrices that
satisfy the relation:%
\[
\alpha_{j}\alpha_{k}+\alpha_{k}\alpha_{j}=2\delta_{jk},\text{ }j,k=1,2,3,4,
\]
where $\delta_{jk}$ denotes the Kronecker symbol. The standard choice of
$\alpha_{j}$ is (\cite{25}):
\[
\alpha_{j}=%
\begin{pmatrix}
0 & \sigma_{j}\\
\sigma_{j} & 0
\end{pmatrix}
,\text{ \ }1\leq j\leq3,\text{ \ \ \ \ }\alpha_{4}=%
\begin{pmatrix}
I_{2} & 0\\
0 & -I_{2}%
\end{pmatrix}
=\beta,
\]
($I_{n}$ is the $n\times n$ unit matrix) and
\[
\sigma_{1}=%
\begin{pmatrix}
0 & 1\\
1 & 0
\end{pmatrix}
,\sigma_{2}=%
\begin{pmatrix}
0 & -i\\
i & 0
\end{pmatrix}
,\sigma_{3}=%
\begin{pmatrix}
1 & 0\\
0 & -1
\end{pmatrix}
\]
are the Pauli matrices. The operator $H_{0}$ is a self-adjoint operator on
$L^{2}\left(  \mathbb{R}^{3};\mathbb{C}^{4}\right)  $ (see Section \ref{BN}).
The perturbed Dirac operator is defined by
\begin{equation}
H=H_{0}+\mathbf{V}. \label{basicnotions6}%
\end{equation}
Here the potential $\mathbf{V}\left(  x\right)  $ is an Hermitian $4\times4$
matrix valued function defined for $x\in\mathbb{R}^{3}$ such that $H$ is a
self-adjoint operator on $L^{2}\left(  \mathbb{R}^{3};\mathbb{C}^{4}\right)  $
(see Section \ref{BN})$.$ For a detailed study of the Dirac equation we refer
to \cite{26}, \cite{25}, and the references quoted there.

We define time delay for the pair $\{H_{0},H\}$ as follows. Let $\zeta\left(
t\right)  $ be such that $\zeta\left(  t\right)  =1$ for $0\leq t\leq1$ and
$\zeta\left(  t\right)  =0$ for $t>1.$ For $R>0$ and a normalized $f\in
L^{2}\left(  \mathbb{R}^{3};\mathbb{C}^{4}\right)  $ we define the quantities%
\[
\mathcal{T}_{0,R}\left(  f\right)  :=\int\limits_{-\infty}^{\infty}\left\Vert
\zeta\left(  \frac{\left\vert x\right\vert }{R}\right)  e^{-iH_{0}%
t}f\right\Vert ^{2}dt=\int\limits_{-\infty}^{\infty}\int_{\left\vert
x\right\vert \leq R}\left\vert e^{-iH_{0}t}f\right\vert _{\mathbb{C}^{4}}%
^{2}dt
\]
and
\[
\mathcal{T}_{R}\left(  f\right)  :=\int\limits_{-\infty}^{\infty}\left\Vert
\zeta\left(  \frac{\left\vert x\right\vert }{R}\right)  e^{-iHt}f\right\Vert
^{2}dt=\int\limits_{-\infty}^{\infty}\int_{\left\vert x\right\vert \leq
R}\left\vert e^{-iHt}f\right\vert _{\mathbb{C}^{4}}^{2}dt.
\]
As $\left\Vert \zeta\left(  \frac{\left\vert x\right\vert }{R}\right)
e^{-iH_{0}t}f\right\Vert ^{2}$ is the probability that the state $e^{-iH_{0}%
t}f$ is localized in the ball $B_{R}:=\{\left.  x\in\mathbb{R}^{3}\right\vert
\left\vert x\right\vert \leq R\}$ at time $t$, $\mathcal{T}_{0,R}\left(
f\right)  $ represents the total time spent in $B_{R}$ by a normalized state
$f$ under the free evolution group $e^{-iH_{0}t},$ and similarly,
$\mathcal{T}_{R}\left(  f\right)  $ measures the total time that the state
represented by $f$ stays in $B_{R}$ under $e^{-iHt}.$

Since the scattering theory is based on the comparison of the
perturbed\ dynamics with the free one, it is natural to define time delay as
the difference of the time that the scattered particles stay in the scattering
region and the time that the free particles, subject to the same initial
conditions, spend in the scattering region. Let $f$ be the initial condition
at $t=0$ that defines the dynamics $e^{-iH_{0}t}f$ of the free particle. Then,
the wave operator $W_{-},$ if it exists, defines the initial condition
$W_{-}f$ in $t=0$ of the scattered state $e^{-iHt}W_{-}f$, with the property
that asymptotically for $t\rightarrow-\infty$ it has the same dynamics as the
free state$,$ i.e., $\left\Vert e^{-iH_{0}t}f-e^{-iHt}W_{-}f\right\Vert
\rightarrow0,$ when $t\rightarrow-\infty.$ We define the time delay in $B_{R}$
as the difference of the time spent in $B_{R}$ by the state $W_{-}f$ and the
time that the free state $f$ stays in $B_{R}.$ That is, the \textit{local time
delay }$\delta_{R}\mathcal{T}\left(  f\right)  $ is given by
\[
\delta_{R}\mathcal{T}\left(  f\right)  :=\mathcal{T}_{R}\left(  W_{-}f\right)
-\mathcal{T}_{0,R}\left(  f\right)  .
\]
As the effective scattering region is all of $\mathbb{R}^{3},$ we have to
consider the limit of\ $\delta_{R}\mathcal{T}\left(  f\right)  ,$ when
$R\rightarrow\infty.$ If the limit exists, this leads us to the definition of
the \textit{global time delay }or simply \textit{time delay}
\begin{equation}
\delta\mathcal{T}\left(  f\right)  :=\lim_{R\rightarrow\infty}\delta
_{R}\mathcal{T}\left(  f\right)  . \label{t95}%
\end{equation}
The main problem in this definition of time delay consists in exhibiting a set
of initial states $f\in L^{2}\left(  \mathbb{R}^{3};\mathbb{C}^{4}\right)  $
and a class of perturbations $H-H_{0}$ for which the limit in (\ref{t95}) exists.

There are a lot of papers concerning time delay for the Schr\"{o}dinger
equation (see \cite{9}, \cite{40}, \cite{10}, \cite{7}, \cite{11}, \cite{13},
\cite{14}, \cite{15}, \cite{16}, \cite{17}, \cite{18}, \cite{19}, \cite{20},
\cite{21}, \cite{22}, \cite{23}, and the references therein). The works
\cite{14}, \cite{16}, \cite{17}, \cite{41} and \cite{21} study more general,
abstract dynamics, however they apply the obtained results to the case of the
Schr\"{o}dinger equation only. Many physical aspects, as well as applications
of time delay are presented in \cite{12}. Time delay for dynamics given by a
regular enough pseudodifferential operator of hypoelliptic-type, such as the
Schr\"{o}dinger operator or the square-root Klein-Gordon operator
(pseudo-relativistic Schr\"{o}dinger operator), were treated in \cite{24}.
However, the Dirac operator was not considered in \cite{24}. As far as we know
there are no papers concerning time delay for the Dirac equation.

In order to present our results we make some definitions. For $1\leq p<\infty$
and $\alpha\in\mathbb{R}$ we denote by $L^{p}\left(  \mathbb{R}^{3}%
;\mathbb{C}^{4}\right)  $ and $\mathcal{H}^{\alpha}\left(  \mathbb{R}%
^{3};\mathbb{C}^{4}\right)  $ the Lebesgue and Sobolev spaces of
$\mathbb{C}^{4}$-vector valued functions, respectively (see, for example,
\cite{adams}). We often will write $L^{p}$ and $\mathcal{H}^{\alpha}$ instead
of $L^{p}\left(  \mathbb{R}^{3};\mathbb{C}^{4}\right)  $ and $\mathcal{H}%
^{\alpha}\left(  \mathbb{R}^{3};\mathbb{C}^{4}\right)  .$ Also, we introduce
the weighted $L^{2}$ spaces for $s\in\mathbb{R},$ $L_{s}^{2}:=\{f:\left\langle
x\right\rangle ^{s}f\left(  x\right)  \in L^{2}\},$ $\left\Vert f\right\Vert
_{L_{s}^{2}}:=\left\Vert \left\langle x\right\rangle ^{s}f\left(  x\right)
\right\Vert _{L^{2}},$ where $\left\langle x\right\rangle =\left(
1+\left\vert x\right\vert ^{2}\right)  ^{1/2}.$ Moreover, for any $\alpha
,s\in\mathbb{R}$ we define $\mathcal{H}_{s}^{\alpha}:=\{f:\left\langle
x\right\rangle ^{s}f\left(  x\right)  \in\mathcal{H}^{\alpha}\},$ $\left\Vert
f\right\Vert _{\mathcal{H}_{s}^{\alpha}}:=\left\Vert \left\langle
x\right\rangle ^{s}f\left(  x\right)  \right\Vert _{\mathcal{H}^{\alpha}},$
where $\left\Vert f\left(  x\right)  \right\Vert _{\mathcal{H}^{\alpha}%
}=\left(  \int_{\mathbb{R}^{3}}\left\langle \xi\right\rangle ^{2\alpha
}\left\vert \hat{f}\left(  \xi\right)  \right\vert ^{2}d\xi\right)  ^{1/2}.$
The Fourier transform $\mathcal{F}$ is given by
\[
\hat{f}\left(  \xi\right)  =\left(  \mathcal{F}f\right)  \left(  \xi\right)
:=\left(  2\pi\right)  ^{-3/2}%
{\displaystyle\int\limits_{\mathbb{R}^{3}}}
e^{-i\xi\cdot x}f\left(  x\right)  dx.
\]
We denote by $\left(  \cdot,\cdot\right)  $ the $L^{2}$ scalar product. The
scalar product in $\mathbb{C}^{4}$ is denoted by $\left\langle \cdot
,\cdot\right\rangle .$ The projections $\mathbf{P}_{\pm}$ on the positive and
negative energies of $H_{0}$ are defined by (see Section \ref{BN})%
\[
\mathbf{P}_{\pm}:=\frac{1}{2}\left(  I_{4}\pm\frac{H_{0}}{\left\vert
H_{0}\right\vert }\right)  .
\]
Also, we introduce
\begin{equation}
\mathbf{A}_{0}:=\frac{1}{2i}\left(  \frac{1}{\left\vert \nabla\right\vert
^{2}}\nabla\cdot x+x\cdot\nabla\frac{1}{\left\vert \nabla\right\vert ^{2}%
}\right)  =\frac{1}{2}\left(  \frac{1}{\left\vert \nabla\right\vert ^{2}%
}\mathbf{A+A}\frac{1}{\left\vert \nabla\right\vert ^{2}}\right)  ,
\label{t164}%
\end{equation}
where $\nabla$ is the gradient and the operator $\mathbf{A}$, known as
\textquotedblleft the generator of dilations\textquotedblright, is defined as
\[
\mathbf{A:=}\frac{1}{2i}\sum_{j=1}^{3}\left(  x\cdot\nabla+\nabla\cdot
x\right)  .
\]
The operators $\Gamma_{0}\left(  E\right)  ,$ given by (\ref{basicnotions11}),
define the spectral representation $\mathcal{F}_{0}$ of the operator $H_{0}$
(\ref{basicnotions51}). We denote by $\mathcal{\hat{H}}$ the space of the
spectral realization of $H_{0}$ under $\mathcal{F}_{0}$ (see
(\ref{basicnotions50})) and $\left(  \cdot,\cdot\right)  _{\mathcal{\hat{H}}}$
denotes the scalar product in $\mathcal{\hat{H}}$. Also for any open set
$O\subset(-\infty,-m)\cup(m,+\infty),$ we define
\begin{equation}
\Phi\left(  O\right)  :=\{\left.  f\in\mathcal{H}_{2}^{3/2+\varepsilon}\left(
\mathbb{R}^{3};\mathbb{C}^{4}\right)  ,\text{ }\varepsilon>0\right\vert \text{
\ }f=\psi_{f}\left(  H_{0}\right)  f,\text{ \ for some }\psi_{f}\text{ such
that }\operatorname*{supp}\psi_{f}\subset O\}. \label{t190}%
\end{equation}
Finally, $\rho\left(  H\right)  $ is the resolvent set of $H$ and $\sigma
_{p}\left(  H\right)  $ is the closure of the set of the eigenvalues of the
operator $H.$

We assume that the wave operators $W_{\pm}$ exist and are complete, and hence,
the scattering operator $\mathbf{S}$ is unitary. In Section \ref{BN} we give
sufficient conditions on the potential $\mathbf{V}$ under which this
assumption is true.

We are now in position to present our results.

\begin{theorem}
\label{T1}Suppose that the state $f\in\mathcal{H}_{2}^{3/2+\varepsilon}\left(
\mathbb{R}^{3};\mathbb{C}^{4}\right)  $, $\varepsilon>0,$ is such that
$\mathbf{S}f\in\mathcal{H}_{2}^{3/2+\varepsilon}\left(  \mathbb{R}%
^{3};\mathbb{C}^{4}\right)  ,$ the function $t\rightarrow\left\Vert \left(
W_{-}-e^{itH}e^{-itH_{0}}\right)  f\right\Vert _{L^{2}\left(  \mathbb{R}%
^{3}\right)  }$ belongs to $L^{1}\left(  (-\infty,0]\right)  $ and
$t\rightarrow\left\Vert \left(  W_{+}-e^{itH}e^{-itH_{0}}\right)
\mathbf{S}f\right\Vert _{L^{2}\left(  \mathbb{R}^{3}\right)  }$ belongs to
$L^{1}\left(  [0,\infty)\right)  .$ Then the limit in (\ref{t95}) exists and%
\begin{equation}
\delta\mathcal{T}\left(  f\right)  =\left(  f,\mathbf{T}f\right)
,\label{t102}%
\end{equation}
where
\[
\mathbf{T:=}H_{0}\mathbf{S}^{\ast}\left(  \mathbf{P}_{-}\left[  \mathbf{A}%
_{0},\mathbf{S}\right]  \mathbf{P}_{-}-\mathbf{P}_{+}\left[  \mathbf{A}%
_{0},\mathbf{S}\right]  \mathbf{P}_{+}\right)  .
\]
Moreover, let the scattering matrix $S\left(  E\right)  $ be continuously
differentiable with respect to $E$ on some open set $O\subset(-\infty
,-m)\cup(m,+\infty)\setminus\sigma_{p}\left(  H\right)  .$ Then, for any
$f\in\Phi\left(  O\right)  ,$ $\mathbf{T}$\ is the Eisenbud-Wigner time delay
operator, that is%
\[
\delta\mathcal{T}\left(  f\right)  =\left(  \Gamma_{0}\left(  E\right)
f,T\left(  E\right)  \Gamma_{0}\left(  E\right)  f\right)  _{\mathcal{\hat{H}%
}},
\]
with%
\begin{equation}
T\left(  E\right)  :=-iS\left(  E\right)  ^{\ast}\frac{d}{dE}S\left(
E\right)  .\label{t214}%
\end{equation}

\end{theorem}

\begin{remark}\rm
Observe that Theorem \ref{T1} remains valid in the case when $H-H_{0}$ is not a
multiplication operator.
\end{remark}

\begin{rem}\rm
\label{R1}We note that in the case of the Schr\"{o}dinger equation a similar
result was proved in \cite{9}. For pseudodifferential operators of
hypoelliptic-type, a result as in Theorem \ref{T1} was obtained in \cite{24}.
In these papers, the Fourier transforms of the functions $f$ and $Sf$ are
assumed to be compactly supported. We do not need this condition and we prove
Theorem \ref{T1} for $f$ and $Sf$ belonging to the weighted Sobolev space
$H_{2}^{3/2+\varepsilon}\left(  \mathbb{R}^{3};\mathbb{C}^{4}\right)  .$
\end{rem}

We suppose now that the potential $\mathbf{V}$ satisfies the following:

\begin{condition}\rm
\label{C1}The potential $\mathbf{V}$ has the form
\[
\mathbf{V}\left(  x\right)  \mathbf{=}\left\langle x\right\rangle ^{-\rho
}\left(  \mathbf{V}_{1}\left(  x\right)  +\mathbf{V}_{2}\left(  x\right)
\right)  ,\text{ \ }\rho>2,
\]
where each element of the matrix $\mathbf{V}_{1}$ belongs to $L^{\infty
}\left(  \mathbb{R}^{3}\right)  $ and the entries of $\mathbf{V}_{2}$ are
$L^{p}\left(  \mathbb{R}^{3}\right)  $ functions, for some $p>3.$ Moreover,
the elements of the matrix $x\cdot\nabla\mathbf{V}_{1}$ are in $L^{\infty
}\left(  \mathbb{R}^{3}\right)  +L^{q_{1}}\left(  \mathbb{R}^{3}\right)  $ and
the entries of $\left\langle x\right\rangle \mathbf{V}_{2}$ belong to
$L^{\infty}\left(  \mathbb{R}^{3}\right)  +L^{q_{2}}\left(  \mathbb{R}%
^{3}\right)  ,$ with $q_{j}\geq2,$ $j=1,2.$
\end{condition}

For $\tau>0$ we define the following dense subset of $L^{2}\left(
\mathbb{R}^{3};\mathbb{C}^{4}\right)  :$%

\[
\mathcal{D}_{\tau}:=\{\left.  f\in L_{\tau}^{2}\left(  \mathbb{R}%
^{3};\mathbb{C}^{4}\right)  \right\vert \text{ \ }f=\psi_{f}\left(
H_{0}\right)  f,\text{ \ for some }\psi_{f}\in C_{0}^{\infty}\left(
(-\infty,-m)\cup(m,+\infty)\setminus\sigma_{p}\left(  H\right)  \right)  \}.
\]

We have

\begin{theorem}
\label{T2}Suppose that $\mathbf{V}$ satisfies Condition \ref{C1} and let
$f\in\mathcal{D}_{\tau},$ $\tau>2.$ Then, the global time delay $\delta
\mathcal{T}\left(  f\right)  $ exists and relation (\ref{t102}) is true, with
$\mathbf{T}$ being the Eisenbud-Wigner time delay operator.
\end{theorem}

Let us now suppose that the potential $\mathbf{V}$ satisfy the following:

\begin{condition}\rm
\label{C2}The potential $\mathbf{V}$ satisfies the estimate
\[
\left\vert \mathbf{V}\left(  x\right)  \right\vert \leq C\left(  1+\left\vert
x\right\vert \right)  ^{-4-\varepsilon},\text{ \ }\varepsilon>0.
\]

\end{condition}

The \textit{spectral shift function} (SSF) is a real valued function
$\xi\left(  E;H,H_{0}\right)  $ such that the relation%
\[
\operatorname*{Tr}\left(  f\left(  H\right)  -f\left(  H_{0}\right)  \right)
=%
{\displaystyle\int\limits_{-\infty}^{\infty}}
\xi\left(  E;H,H_{0}\right)  f^{\prime}\left(  E\right)  dE,
\]
known as the trace formula, holds at least for all $f\in C_{0}^{\infty}\left(
\mathbb{R}\right)  $ (see \cite{39}, \cite{36} and the references there in).
Here $\operatorname*{Tr}A$ denotes the trace of an operator $A.$ The
\textit{average time delay} at energy $E$ is defined by $\operatorname*{Tr}%
T\left(  E\right)  ,$ with $T\left(  E\right)  $ given by (\ref{t214}).
Finally, we present a formula that relates the average time delay with the
SSF. We have the following result (see \cite{Buslaev}, \cite{21}, \cite{39},
\cite{36}, and the references therein in the case of the Schr\"{o}dinger
operator).\textit{ }

\begin{theorem}
\label{T4}Assume that $\mathbf{V}$ satisfies Condition \ref{C2}. Then, the
following equality is valid
\begin{equation}
\operatorname*{Tr}T\left(  E\right)  =-2\pi\xi^{\prime}\left(  E;H,H_{0}%
\right)  , \label{t186}%
\end{equation}
for $E\in(-\infty,-m)\cup(m,+\infty).$
\end{theorem}

\begin{rem}\rm
For an operator $A$ of trace class we denote by $\operatorname*{Det}\left(
I+A\right)  $ the determinant of $I+A$\ (\cite{39}, \cite{36}). If Condition
\ref{C2} holds, Theorem 4.5 of \cite{yafaev2005} implies that the operator
$S\left(  E\right)  -I$ is of trace class. The scattering phase $\theta\left(
E\right)  $ is defined by the relation $\operatorname*{Det}S\left(  E\right)
=e^{-2i\theta\left(  E\right)  }.$ It also follows from Theorem 4.5 of
\cite{yafaev2005} that the SSF $\xi\left(  E;H,H_{0}\right)  $ exists and it
is related to the scattering phase $\theta\left(  E\right)  $ by $\xi\left(
E;H,H_{0}\right)  =\left(  1/\pi\right)  \theta\left(  E\right)  $. Thus,
formula (\ref{t186}) shows that up to numerical coefficients, the average time
delay, the SSF and the scattering phase, that have different physical
meanings, coincide. On the other hand, we observe that Theorem \ref{T2} and
Theorem \ref{T4} establish, via the Eisenbud-Wigner operator, a connection
between $\delta\mathcal{T}\left(  f\right)  $ and the SSF $\xi\left(
E;H,H_{0}\right)  $.
\end{rem}

We now briefly explain our strategy. As in the case of the Schr\"{o}dinger
equation (\cite{9}), by Proposition \ref{t89} the study of the limit
(\ref{t95}) reduces to finding an asymptotic expansion for the quantity
\begin{equation}
I\left(  R\right)  :=\int\limits_{0}^{\infty}\left\langle e^{-iH_{0}t}%
f,\zeta\left(  \frac{\left\vert x\right\vert }{R}\right)  e^{-iH_{0}%
t}g\right\rangle dt, \label{t96}%
\end{equation}
as $R\rightarrow\infty.$ In order to find the asymptotics of $I\left(
R\right)  $ we proceed as follows. Separating $H_{0}$ into positive and
negative energies we obtain the decomposition (\ref{time9}) below. Since
$\sqrt{\left\vert p\right\vert ^{2}+m^{2}}+\sqrt{\left\vert q\right\vert
^{2}+m^{2}}$ is different from $0$ for all $p$ and $q,$ the terms
$I_{2}\left(  R\right)  $ and $I_{4}\left(  R\right)  $ in (\ref{time9}),
containing $e^{-i\left(  \sqrt{\left\vert p\right\vert ^{2}+m^{2}}%
+\sqrt{\left\vert q\right\vert ^{2}+m^{2}}\right)  t}$ or $e^{i\left(
\sqrt{\left\vert p\right\vert ^{2}+m^{2}}+\sqrt{\left\vert q\right\vert
^{2}+m^{2}}\right)  t},$ result to be $o\left(  1\right)  $ functions, as
$R\rightarrow\infty.$ On the other hand, as $\sqrt{\left\vert p\right\vert
^{2}+m^{2}}=\sqrt{\left\vert q\right\vert ^{2}+m^{2}}$ for $p=q,$ the terms
$I_{1}\left(  R\right)  $ and $I_{3}\left(  R\right)  ,$ that contain
$e^{-i\left(  \sqrt{\left\vert p\right\vert ^{2}+m^{2}}-\sqrt{\left\vert
q\right\vert ^{2}+m^{2}}\right)  t}$ or $e^{i\left(  \sqrt{\left\vert
p\right\vert ^{2}+m^{2}}-\sqrt{\left\vert q\right\vert ^{2}+m^{2}}\right)
t},$ give us the principal part in the asymptotics of $I\left(  R\right)  ,$
as $R\rightarrow\infty.$ The asymptotic expansion of $I\left(  R\right)  \ $is
given by Theorem \ref{t58}. To prove this theorem we first separate the part
in $I_{1}\left(  R\right)  $ and $I_{3}\left(  R\right)  $ that diverges as
$R,$ when $R\rightarrow\infty,$ and the constant part. These are the results
of Lemmas \ref{t54} and \ref{t57}, respectively. After that, we need to show
that the remaining part is $o\left(  1\right)  $ function, as $R\rightarrow
\infty.$ We prove this result in Lemmas \ref{t45} and \ref{t49}. Finally, in
Lemma \ref{t48} we show that the terms $I_{2}\left(  R\right)  $ and
$I_{4}\left(  R\right)  $ also are $o\left(  1\right)  $ functions, as
$R\rightarrow\infty$. This completes the scheme of the proof of Theorem
\ref{t58}. We observe that the difficulty in obtaining the asymptotics of
$I\left(  R\right)  $ consists in that the integral in $t$ in (\ref{t96}) is
conditionally convergent and this convergence depends on $R.$ The dependence
on $R$ is rather delicate and therefore we need some sharp estimates in
weighted Sobolev spaces in order to obtain our result. Also, we note that we
do not use a formula analogous to the Alsholm-Kato formula (see (2.1) of
\cite{AsholmKato}) that was used in \cite{9} or \cite{24}. Besides, our
approach allows us to obtain the asymptotics of $I\left(  R\right)  $ for
functions $f,g\ $in weighted Sobolev space $\mathcal{H}_{2}^{3/2+\varepsilon}%
$, $\varepsilon>0,$ and we do not need that the Fourier transforms of $f$ and
$g$ have compact support. This enables us to prove Theorem \ref{T1} for $f$
and $\mathbf{S}f$ belonging to $\mathcal{H}_{2}^{3/2+\varepsilon}\left(
\mathbb{R}^{3};\mathbb{C}^{4}\right)  $ (see Remark \ref{R1}). If we assume
that $f$ and $g$ in Theorem \ref{t58} are such that their Fourier transforms
are compactly supported, the proof of Theorem \ref{t58} results to be
technically easier. The first assertion of Theorem \ref{T1} is proved by using
Proposition \ref{t89} and Theorem \ref{t58}, and then, we give the relation of
the time delay $\mathbf{T}$ and the Eisenbud-Wigner time delay operator
(\cite{43}, \cite{42}), which is the result of the second assertion of Theorem
\ref{T1} (see Subsection \ref{proofT1}). We observe that our method is direct
and can be applied to another equations in quantum scattering theory, such as,
for example, the Schr\"{o}dinger operator, the Klein-Gordon equation or the
Pauli operator. The proof of Theorem \ref{T2} consists in showing that under
Condition \ref{C1} on the potential $\mathbf{V}$ the assumptions of Theorem
\ref{T1} are valid. We do this by adapting the results of \cite{40} and
\cite{15} for the Schr\"{o}dinger operator to the case of the Dirac operator.
Finally, we prove Theorem \ref{T4} by using the Birman-Krein's formula,
obtained for the Dirac equation in \cite{yafaev2005}.

The paper is organized as follows. In Section 2 we give some known results
about scattering theory for the Dirac operator. In Section 3 we obtain the
asymptotic expansion for $I\left(  R\right)  ,$ as $R\rightarrow\infty.$
Section 4 is dedicated to the proofs of our theorems. In Subsection
\ref{proofT1} we use the asymptotics of $I\left(  R\right)  $ in order to
prove Theorem \ref{T1}. Subsection \ref{proofT2} is dedicated to Theorem
\ref{T2}. Finally, the proof of Theorem \ref{T4} is given in Subsection
\ref{proofT4}.

\section{Basic notions.\label{BN}}

The free Dirac operator $H_{0}$ (\ref{basicnotions1}) is a self-adjoint
operator on $L^{2}\left(  \mathbb{R}^{3};\mathbb{C}^{4}\right)  $ with domain
$D\left(  H_{0}\right)  =\mathcal{H}^{1}\left(  \mathbb{R}^{3};\mathbb{C}%
^{4}\right)  $ (\cite{25}). We can diagonalize $H_{0}$ by the Fourier
transform $\mathcal{F}.$ Actually, $\mathcal{F}H_{0}\mathcal{F}^{\ast}$ acts
as multiplication by the matrix $h_{0}\left(  \xi\right)  =\alpha\cdot
\xi+m\beta.$ This matrix has two eigenvalues $E=\pm\sqrt{\xi^{2}+m^{2}}$ and
each eigenspace $X^{\pm}\left(  \xi\right)  $ is a two-dimensional subspace of
$\mathbb{C}^{4}.$ The orthogonal projections onto these eigenspaces are given
by (see \cite{25}, page 9)%
\begin{equation}
P_{\pm}\left(  \xi\right)  :=\frac{1}{2}\left(  I_{4}\pm\left(  \xi^{2}%
+m^{2}\right)  ^{-1/2}\left(  \alpha\cdot\xi+m\beta\right)  \right)  .
\label{basicnotions49}%
\end{equation}
Note that
\[
P_{\pm}\left(  \xi\right)  \mathcal{F}=\left(  \mathcal{F}\mathbf{P}_{\pm
}\right)  \left(  \xi\right)  ,
\]
where%
\[
\mathbf{P}_{\pm}:=\frac{1}{2}\left(  I_{4}\pm\frac{H_{0}}{\left\vert
H_{0}\right\vert }\right)  .
\]
The spectrum of $H_{0}$ is purely absolutely continuous and it is given by
$\sigma\left(  H_{0}\right)  =\sigma_{ac}\left(  H_{0}\right)  =(-\infty
,-m]\cup\lbrack m,\infty).$

Let us now consider the perturbed Dirac operator $H$, given by
(\ref{basicnotions6}). Suppose that the Hermitian $4\times4$ matrix valued
potential $\mathbf{V,}$ defined for $x\in\mathbb{R}^{3}$ satisfies the following

\begin{condition}\rm
\label{basicnotions26}For some $s_{0}>1/2,$ $\left\langle x\right\rangle
^{2s_{0}}\mathbf{V}$ is a compact operator from $\mathcal{H}^{1}$ to $L^{2}.$
\end{condition}

The assumptions on a potential $\mathbf{V,}$ assuring Condition
\ref{basicnotions26} are well known (see, for example, \cite{27}). In
particular, Condition \ref{basicnotions26} for $\mathbf{V}$ holds, if for some
$\varepsilon>0,$ $\sup_{x\in\mathbb{R}^{3}}\int_{\left\vert x-y\right\vert
\leq1}\left\vert \left\langle y\right\rangle ^{2s_{0}}\mathbf{V}\left(
y\right)  \right\vert ^{3+\varepsilon}dy<\infty$ and $\int_{\left\vert
x-y\right\vert \leq1}\left\vert \left\langle y\right\rangle ^{2s_{0}%
}\mathbf{V}\left(  y\right)  \right\vert ^{3+\varepsilon}dy\rightarrow0,$ as
$\left\vert x\right\vert \rightarrow\infty$ (see Theorem 9.6, Chapter 6, of
\cite{27}). Of course, the last two relations are true if $\mathbf{V}$
satisfies Condition \ref{C1}.

Since $\mathbf{V}$ is an Hermitian $4\times4$ matrix valued potential
$\mathbf{V,}$ Condition \ref{basicnotions26} implies assumptions (A$_{1}%
$)-(A$_{3}$) of \cite{28}. Thus, under Condition \ref{basicnotions26} $H$ is a
self-adjoint operator on $D\left(  H\right)  =\mathcal{H}^{1}$ and the
essential spectrum $\sigma_{ess}\left(  H\right)  =\sigma\left(  H_{0}\right)
$. The wave operators (WO), defined as the following strong limit%
\[
W_{\pm}\left(  H,H_{0}\right)  :=s-\lim_{t\rightarrow\pm\infty}e^{iHt}%
e^{-iH_{0}t},
\]
exist and are complete, i.e., $\operatorname*{Range}W_{\pm}=\mathcal{H}_{ac}$
(the subspace of absolutely continuity of $H$) and the singular continuous
spectrum of $H$ is absent.

\begin{remark}\rm
\label{Rem1}We recall that the study about the absence of eigenvalues of $H$
embedded in the absolutely continuous spectrum was made in \cite{29},
\cite{32}, \cite{30}, \cite{31}, and the references quoted there. In
particular, there are no eigenvalues in the absolutely continuous spectrum if
\[
\left\vert \mathbf{V}\left(  x\right)  \right\vert \leq C\left(  1+\left\vert
x\right\vert \right)  ^{-1-\varepsilon},\text{ \ }\varepsilon>0.
\]

\end{remark}

From the existence of the WO it follows that $HW_{\pm}=W_{\pm}H_{0}$
(intertwining relations). The scattering operator, defined by%
\[
\mathbf{S=S}\left(  H,H_{0}\right)  :=W_{+}^{\ast}W_{-},
\]
commutes with $H_{0}$ and it is unitary.

Let $H_{0S}:=-\triangle$ be the free Schr\"{o}dinger operator in $L^{2}\left(
\mathbb{R}^{3};\mathbb{C}^{4}\right)  .$ The limiting absorption principle
(LAP) is the following statement. For $z$ in the resolvent set of $H_{0S}$ let
$R_{0S}\left(  z\right)  :=\left(  H_{0S}-z\right)  ^{-1}$ be the resolvent.
The limits $R_{0S}\left(  \lambda\pm i0\right)  =\lim_{\varepsilon
\rightarrow+0}R_{0S}\left(  \lambda\pm i\varepsilon\right)  ,$ ($\varepsilon
\rightarrow+0$ means $\varepsilon\rightarrow0$ with $\varepsilon>0$) exist in
the uniform operator topology in $\mathcal{B}\left(  L_{s}^{2},\mathcal{H}%
_{-s}^{\alpha}\right)  ,$ $s>1/2,$ $\left\vert \alpha\right\vert \leq2$
(\cite{33},\cite{34},\cite{35},\cite{36}) and, moreover, $\left\Vert
R_{0S}\left(  \lambda\pm i0\right)  f\right\Vert _{\mathcal{H}^{\alpha,-s}%
}\leq C_{s,\delta}\lambda^{-\left(  1-\left\vert \alpha\right\vert \right)
/2}\left\Vert f\right\Vert _{L_{s}^{2}},$ for $\lambda\in\lbrack\delta
,\infty),$ $\delta>0$. Here for any pair of Banach spaces $X,Y,$
$\mathcal{B}\left(  X,Y\right)  $ denotes the Banach space of all bounded
operators from $X$ into $Y.$ The functions $R_{0S}^{\pm}\left(  \lambda
\right)  ,$ given by $R_{0S}\left(  \lambda\right)  $ if $\operatorname{Im}%
\lambda\neq0,$ and $R_{0S}\left(  \lambda\pm i0\right)  ,$ if $\lambda
\in(0,\infty),$ are defined for $\lambda\in\mathbb{C}^{\pm}\cup\left(
0,\infty\right)  $ ($\mathbb{C}^{\pm}$ denotes, respectively, the upper,
lower, open complex half-plane) with values in $\mathcal{B}\left(  L_{s}%
^{2},\mathcal{H}_{-s}^{\alpha}\right)  $ and they are analytic for
$\operatorname{Im}\lambda\neq0$ and locally H\"{o}lder continuous for
$\lambda\in(0,\infty)$ with exponent $\vartheta$ satisfying the estimates
$0<\vartheta\leq s-1/2$ and $\vartheta<1.$

For $z$ in the resolvent set of $H_{0}$ let $R_{0}\left(  z\right)  :=\left(
H_{0}-z\right)  ^{-1}$ be the resolvent. From the LAP for $H_{0S}$ it follows
that the limits (see Lemma 3.1 of \cite{28})%
\begin{equation}
\left.  R_{0}\left(  E\pm i0\right)  =\lim_{\varepsilon\rightarrow+0}%
R_{0}\left(  E\pm i\varepsilon\right)  =\left\{
\begin{array}
[c]{c}%
\left(  H_{0}+E\right)  R_{0S}\left(  \left(  E^{2}-m^{2}\right)  \pm
i0\right)  \text{ for }E>m\\
\left(  H_{0}+E\right)  R_{0S}\left(  \left(  E^{2}-m^{2}\right)  \mp
i0\right)  \text{ for }E<-m,
\end{array}
\right.  \right.  \label{basicnotions27}%
\end{equation}
exist for $E\in(-\infty,-m)\cup(m,\infty)$ in the uniform operator topology in
$\mathcal{B}\left(  L_{s}^{2},\mathcal{H}_{-s}^{\alpha}\right)  ,$ $s>1/2,$
$\alpha\leq1,$ and $\left\Vert R_{0}\left(  E\pm i0\right)  f\right\Vert
_{\mathcal{H}^{\alpha,-s}}$ $\leq C_{s,\delta}\left\vert E\right\vert
^{\left\vert \alpha\right\vert }\left\Vert f\right\Vert _{L_{s}^{2}},$ for
$\left\vert E\right\vert \in\lbrack m+\delta,\infty),$ $\delta>0$.
Furthermore, the functions, $R_{0}^{\pm}\left(  E\right)  ,$ given by
$R_{0}\left(  E\right)  ,$ if $\operatorname{Im}E\neq0,$ and by $R_{0}\left(
E\pm i0\right)  ,$ if $E\in(-\infty,-m)\cup(m,\infty),$ are defined for
$E\in\mathbb{C}^{\pm}\cup\left(  -\infty,-m\right)  \cup\left(  m,+\infty
\right)  $ with values in $\mathcal{B}\left(  L_{s}^{2},\mathcal{H}%
_{-s}^{\alpha}\right)  ,$ and moreover, they are analytic for
$\operatorname{Im}E\neq0$ and locally H\"{o}lder continuous for $E\in
(-\infty,-m)\cup(m,\infty)$ with exponent $\vartheta$ such that $0<\vartheta
\leq s-1/2$ and $\vartheta<1.$

Next we consider the resolvent $R\left(  z\right)  :=\left(  H-z\right)
^{-1}$ for $z$ in the resolvent set of $H.$ The following limits exist for
$E\in\{(-\infty,-m)\cup(m,\infty)\}\backslash\sigma_{p}\left(  H\right)  $ in
the uniform operator topology in $\mathcal{B}\left(  L_{s}^{2},\mathcal{H}%
_{-s}^{\alpha}\right)  ,$ $s\in\left(  1/2,s_{0}\right]  ,$ $\left\vert
\alpha\right\vert \leq1,$ where $s_{0}$ is defined by Condition
\ref{basicnotions26} (see Theorem 3.9 of \cite{28})%
\begin{equation}
\left.  R\left(  E\pm i0\right)  =\lim_{\varepsilon\rightarrow+0}R\left(  E\pm
i\varepsilon\right)  =R_{0}\left(  E\pm i0\right)  \left(  1+\mathbf{V}%
R_{0}\left(  E\pm i0\right)  \right)  ^{-1}.\right.  \label{basicnotions13}%
\end{equation}
From this relation and the properties of $R_{0}^{\pm}\left(  E\right)  $ it
follows that the functions, $R^{\pm}\left(  E\right)  :=\{R\left(  E\right)  $
if $\operatorname{Im}E\neq0,$ and $R\left(  E\pm i0\right)  ,$ $E\in
\{(-\infty,-m)\cup(m,\infty)\}\backslash\sigma_{p}\left(  H\right)  \},$
defined for $E\in\mathbb{C}^{\pm}\cup\{(-\infty,-m)\cup(m,\infty
)\}\backslash\sigma_{p}\left(  H\right)  ,$ with values in $\mathcal{B}\left(
L_{s}^{2},\mathcal{H}_{-s}^{\alpha}\right)  $ are analytic for
$\operatorname{Im}E\neq0$ and locally H\"{o}lder continuous for $E\in
\{(-\infty,-m)\cup(m,\infty)\}\backslash\sigma_{p}\left(  H\right)  $ with
exponent $\vartheta$ such that $0<\vartheta\leq s-1/2,$ $s<s_{0}$ and
$\vartheta<1.$

We now give a spectral representation of $H_{0}$. Let us define%

\begin{equation}
\left(  \Gamma_{0}\left(  E\right)  f\right)  \left(  \omega\right)  :=\left(
2\pi\right)  ^{-\frac{3}{2}}\upsilon\left(  E\right)  P_{\omega}\left(
E\right)  \int_{\mathbb{R}^{3}}e^{-i\nu\left(  E\right)  \omega\cdot
x}f\left(  x\right)  dx, \label{basicnotions11}%
\end{equation}
where
\[
\upsilon\left(  E\right)  =\left(  E^{2}\left(  E^{2}-m^{2}\right)  \right)
^{\frac{1}{4}},\text{ and \ }\nu\left(  E\right)  =\sqrt{E^{2}-m^{2}%
},\text{\ }%
\]
and%
\[
P_{\omega}\left(  E\right)  :=\left\{
\begin{array}
[c]{c}%
P_{+}\left(  \nu\left(  E\right)  \omega\right)  ,\text{ \ }E>m,\\
P_{-}\left(  \nu\left(  E\right)  \omega\right)  ,\text{ \ }E<-m,
\end{array}
\right.
\]
($P_{\pm}\left(  \xi\right)  $ are given by (\ref{basicnotions49})). The
adjoint operator $\Gamma_{0}^{\ast}\left(  E\right)  :L^{2}\left(
\mathbb{S}^{2};\mathbb{C}^{4}\right)  \rightarrow L_{-s}^{2},$ $s>1/2,$ is
given by
\[
\left(  \Gamma_{0}^{\ast}\left(  E\right)  f\right)  \left(  \omega\right)
:=\left(  2\pi\right)  ^{-\frac{3}{2}}\upsilon\left(  E\right)  \int
_{\mathbb{S}^{2}}e^{i\nu\left(  E\right)  x\cdot\omega}P_{\omega}\left(
E\right)  f\left(  \omega\right)  d\omega.
\]
Note that $\Gamma_{0}\left(  E\right)  $ is unitary equivalent to the trace
operator $T_{0}\left(  E\right)  ,$ defined by using the Foldy-Wouthuysen
transform in \cite{28} (see \cite{NW}). Then, from the properties of
$T_{0}\left(  E\right)  $~(\cite{37}, \cite{38}, \cite{36}) we conclude that
$\Gamma_{0}\left(  E\right)  $ is bounded from $L_{s}^{2},$ $s>1/2,$ into
$L^{2}\left(  \mathbb{S}^{2};\mathbb{C}^{4}\right)  $ and the operator valued
function $\Gamma_{0}\left(  E\right)  $ is locally H\"{o}lder continuous on
$(-\infty,-m)\cup(m,\infty)$\ with exponent $\vartheta$ satisfying
$0<\vartheta\leq s-1/2$ and $\vartheta<1$.

Since the operators $\Gamma_{0}\left(  E\right)  $ and $T_{0}\left(  E\right)
$ are unitary equivalent, it follows from Section 3 of \cite{28} that the
operator
\begin{equation}
\left(  \mathcal{F}_{0}f\right)  \left(  E,\omega\right)  :=\left(  \Gamma
_{0}\left(  E\right)  f\right)  \left(  \omega\right)  )
\label{basicnotions51}%
\end{equation}
extends to unitary operator from $L^{2}$ onto
\begin{equation}
\mathcal{\hat{H}}:=\int_{\left(  -\infty,-m\right)  \cup\left(  m,+\infty
\right)  }^{\oplus}\mathcal{H}\left(  E\right)  dE, \label{basicnotions50}%
\end{equation}
where%
\[
\mathcal{H}\left(  E\right)  :=\int_{\mathbb{S}^{2}}^{\oplus}X^{\pm}\left(
\nu\left(  E\right)  \omega\right)  d\omega,\text{ \ }\pm E>m.
\]
Moreover, $\mathcal{F}_{0}$ gives a spectral representation of $H_{0}$%
\[
\mathcal{F}_{0}H_{0}\mathcal{F}_{0}^{\ast}=E,
\]
the operator of multiplication by $E$ in $\mathcal{\hat{H}}$. For these
results see \cite{NW}.

Since the scattering operator $\mathbf{S}$ commutes with $H_{0},$ the operator
$\mathcal{F}_{0}\mathbf{S}\mathcal{F}_{0}^{\ast}$ acts as a multiplication by
the operator valued function $S\left(  E\right)  :\mathcal{H}\left(  E\right)
\rightarrow\mathcal{H}\left(  E\right)  .$ We obtain from Theorem 4.2 of
\cite{28} (see also \cite{39},\cite{35},\cite{36}) the following stationary
formula for $S\left(  E\right)  $ (see \cite{NW}),
\begin{equation}
S\left(  E\right)  =I-2\pi i\Gamma_{0}\left(  E\right)  \left(  \mathbf{V}%
-\mathbf{V}R\left(  E+i0\right)  \mathbf{V}\right)  \Gamma_{0}\left(
E\right)  ^{\ast}, \label{basicnotions14}%
\end{equation}
for $E\in\{(-\infty,-m)\cup(m,\infty)\}\backslash\sigma_{p}\left(  H\right)
.$ Here $I$ is the identity operator on $\mathcal{H}\left(  E\right)  .$
$S\left(  E\right)  $ is called the scattering matrix.

\section{Asymptotics for $I\left(  R\right)  .$}

Let us first show that $I\left(  R\right)  $ is well defined. The following
result holds.

\begin{lemma}
\label{t42}For any fixed $0<R<\infty,$ and $f,g\in\mathcal{H}_{3/2+\varepsilon
}^{3/2+\varepsilon}$, $\varepsilon>0,$ we have
\[
I\left(  R\right)  <\infty.
\]

\end{lemma}

\begin{proof}
Using that $\mathbf{P}_{+}+\mathbf{P}_{-}=I$, we get
\begin{equation}
\left.
\begin{array}
[c]{c}%
I\left(  R\right)  =%
{\displaystyle\int\limits_{0}^{\infty}}
\left\langle \mathcal{F}\left(  \left(  \mathbf{P}_{+}+\mathbf{P}_{-}\right)
\left(  e^{-iH_{0}t}f\right)  \right)  ,\mathcal{F}\left(  \zeta\left(
\frac{\left\vert x\right\vert }{R}\right)  \left(  \mathbf{P}_{+}%
+\mathbf{P}_{-}\right)  e^{-iH_{0}t}g\right)  \right\rangle dt\\
=\left(  2\pi\right)  ^{-3}%
{\displaystyle\int\limits_{0}^{\infty}}
\left\langle \mathcal{F}\left(  \left(  \mathbf{P}_{+}+\mathbf{P}_{-}\right)
e^{-iH_{0}t}f\right)  ,\tilde{\zeta}_{R}\ast\mathcal{F}\left(  \left(
\mathbf{P}_{+}+\mathbf{P}_{-}\right)  e^{-iH_{0}t}g\right)  \right\rangle dt\\
=\left(  2\pi\right)  ^{-3}\left(  I_{1}+I_{2}+I_{3}+I_{4}\right)  \left(
R\right)  ,
\end{array}
\right.  \label{time9}%
\end{equation}
where $\ast$ denotes the convolution, $\tilde{\zeta}_{R}=\left(  2\pi\right)
^{3/2}\mathcal{F}\left(  \zeta\left(  \frac{\left\vert x\right\vert }%
{R}\right)  \right)  $ and
\[
I_{1}\left(  R\right)  :=\int\limits_{0}^{\infty}\int\limits_{\mathbb{R}^{3}%
}\int\limits_{\mathbb{R}^{3}}\left\langle e^{-i\sqrt{\left\vert p\right\vert
^{2}+m^{2}}t}f_{+}\left(  p\right)  ,\tilde{\zeta}_{R}\left(  p-q\right)
e^{-i\sqrt{\left\vert q\right\vert ^{2}+m^{2}}t}g_{+}\left(  q\right)
\right\rangle dqdpdt,
\]%
\[
I_{2}\left(  R\right)  :=\int\limits_{0}^{\infty}\int\limits_{\mathbb{R}^{3}%
}\int\limits_{\mathbb{R}^{3}}\left\langle e^{-i\sqrt{\left\vert p\right\vert
^{2}+m^{2}}t}f_{+}\left(  p\right)  ,\tilde{\zeta}_{R}\left(  p-q\right)
e^{i\sqrt{\left\vert q\right\vert ^{2}+m^{2}}t}g_{-}\left(  q\right)
\right\rangle dqdpdt,
\]%
\[
I_{3}\left(  R\right)  :=\int\limits_{0}^{\infty}\int\limits_{\mathbb{R}^{3}%
}\int\limits_{\mathbb{R}^{3}}\left\langle e^{i\sqrt{\left\vert p\right\vert
^{2}+m^{2}}t}f_{-}\left(  p\right)  ,\tilde{\zeta}_{R}\left(  p-q\right)
e^{-i\sqrt{\left\vert q\right\vert ^{2}+m^{2}}t}g_{+}\left(  q\right)
\right\rangle dqdpdt,
\]%
\[
I_{4}\left(  R\right)  :=\int\limits_{0}^{\infty}\int\limits_{\mathbb{R}^{3}%
}\int\limits_{\mathbb{R}^{3}}\left\langle e^{i\sqrt{\left\vert p\right\vert
^{2}+m^{2}}t}f_{-}\left(  p\right)  ,\tilde{\zeta}_{R}\left(  p-q\right)
e^{i\sqrt{\left\vert q\right\vert ^{2}+m^{2}}t}g_{-}\left(  q\right)
\right\rangle dqdpdt,
\]
with $f_{\pm}\left(  p\right)  :=P_{\pm}\left(  p\right)  \hat{f}\left(
p\right)  $ and $g_{\pm}\left(  p\right)  :=P_{\pm}\left(  p\right)  \hat
{g}\left(  p\right)  .$ We prove that $\left\vert I_{1}\left(  R\right)
\right\vert <\infty.$ The proof of $\left\vert I_{j}\left(  R\right)
\right\vert <\infty,$ for $j=2,3,4,$ is similar.

Let us define%
\[
h\left(  f_{+},g_{+};t\right)  :=\int\limits_{\mathbb{R}^{3}}\int
\limits_{\mathbb{R}^{3}}\left\langle e^{-i\sqrt{\left\vert p\right\vert
^{2}+m^{2}}t}f_{+}\left(  p\right)  ,\tilde{\zeta}_{R}\left(  p-q\right)
e^{-i\sqrt{\left\vert q\right\vert ^{2}+m^{2}}t}g_{+}\left(  q\right)
\right\rangle dqdp.
\]
It is enough to show that%
\begin{equation}
\left.  \left\vert \int\limits_{0}^{\infty}h\left(  f_{+},g_{+};t\right)
dt\right\vert \leq C\left\Vert f_{+}\right\Vert _{\mathcal{H}_{3/2+\varepsilon
}^{3/2+\varepsilon}\left(  \mathbb{R}^{3}\right)  }\left\Vert g_{+}\right\Vert
_{\mathcal{H}_{3/2+\varepsilon}^{3/2+\varepsilon}\left(  \mathbb{R}%
^{3}\right)  }.\right.  \label{time11}%
\end{equation}
Since
\[
\left\vert \tilde{\zeta}_{R}\left(  p\right)  \right\vert \leq C,\text{
\ }p\in\mathbb{R}^{3},
\]
then%
\[
\left\vert \int\limits_{0}^{1}h\left(  f_{+},g_{+};t\right)  dt\right\vert
\leq C\left\Vert f_{+}\right\Vert _{L^{1}\left(  \mathbb{R}^{3}\right)
}\left\Vert g_{+}\right\Vert _{L^{1}\left(  \mathbb{R}^{3}\right)  }\leq
C\left\Vert f_{+}\right\Vert _{L_{3/2+\varepsilon}^{2}\left(  \mathbb{R}%
^{3}\right)  }\left\Vert g_{+}\right\Vert _{L_{3/2+\varepsilon}^{2}\left(
\mathbb{R}^{3}\right)  },
\]
and thus, in order to get (\ref{time11}), we need the estimate%
\begin{equation}
\left.  \left\vert \int\limits_{1}^{\infty}h\left(  f_{+},g_{+};t\right)
dt\right\vert \leq C\left\Vert f_{+}\right\Vert _{\mathcal{H}_{3/2+\varepsilon
}^{3/2+\varepsilon}\left(  \mathbb{R}^{3}\right)  }\left\Vert g_{+}\right\Vert
_{\mathcal{H}_{3/2+\varepsilon}^{3/2+\varepsilon}\left(  \mathbb{R}%
^{3}\right)  }.\right.  \label{t1}%
\end{equation}
Suppose that $f,g\in\mathcal{S}$ (here $\mathcal{S}$ denote the Schwartz
class). Observe that%
\begin{equation}
\left.
\begin{array}
[c]{c}%
h\left(  f_{+},g_{+};t\right)  =%
{\displaystyle\int\limits_{\mathbb{S}^{2}}}
{\displaystyle\int\limits_{\mathbb{S}^{2}}}
{\displaystyle\int\limits_{0}^{\infty}}
{\displaystyle\int\limits_{0}^{\infty}}
\left\langle e^{-i\sqrt{r^{2}+m^{2}}t}f_{+}\left(  r\omega\right)
,\tilde{\zeta}_{R}\left(  r\omega-r_{1}\omega^{\prime}\right)  e^{-i\sqrt
{r_{1}^{2}+m^{2}}t}g_{+}\left(  r_{1}\omega^{\prime}\right)  \right\rangle
r^{2}r_{1}^{2}drdr_{1}d\omega^{\prime}d\omega\\
=\dfrac{1}{t^{2}}%
{\displaystyle\int\limits_{\mathbb{S}^{2}}}
{\displaystyle\int\limits_{\mathbb{S}^{2}}}
{\displaystyle\int\limits_{0}^{\infty}}
{\displaystyle\int\limits_{0}^{\infty}}
\left\langle \left(  \partial_{r}e^{-i\sqrt{r^{2}+m^{2}}t}\right)
\dfrac{\sqrt{r^{2}+m^{2}}}{r}f_{+}\left(  r\omega\right)  ,\right. \\
\left.  \times\left(  \partial_{r_{1}}e^{-i\sqrt{r_{1}^{2}+m^{2}}t}\right)
\tilde{\zeta}_{R}\left(  r\omega-r_{1}\omega^{\prime}\right)  \dfrac
{\sqrt{r_{1}^{2}+m^{2}}}{r_{1}}g_{+}\left(  r_{1}\omega^{\prime}\right)
\right\rangle r^{2}r_{1}^{2}drdr_{1}d\omega^{\prime}d\omega.
\end{array}
\right.  \label{time10}%
\end{equation}
Integrating by parts in both $r$ and $r_{1},$ we get
\[
\left.
\begin{array}
[c]{c}%
h\left(  f_{+},g_{+};t\right)  =\dfrac{1}{t^{2}}%
{\displaystyle\int\limits_{\mathbb{S}^{2}}}
{\displaystyle\int\limits_{\mathbb{S}^{2}}}
{\displaystyle\int\limits_{0}^{\infty}}
{\displaystyle\int\limits_{0}^{\infty}}
e^{-i\left(  \sqrt{r^{2}+m^{2}}-\sqrt{r_{1}^{2}+m^{2}}\right)  t}\partial
_{r}\partial_{r_{1}}\left\langle r\sqrt{r^{2}+m^{2}}f_{+}\left(
r\omega\right)  ,\right. \\
\left.  \times\tilde{\zeta}_{R}\left(  r\omega-r_{1}\omega^{\prime}\right)
r_{1}\sqrt{r_{1}^{2}+m^{2}}g_{+}\left(  r_{1}\omega^{\prime}\right)
\right\rangle drdr_{1}d\omega^{\prime}d\omega.
\end{array}
\right.
\]
Noting that
\[
\left\vert \nabla^{j}\tilde{\zeta}_{R}\left(  p\right)  \right\vert \leq
C_{j},\text{ }p\in\mathbb{R}^{3},\text{\ for all }j,
\]
we obtain%
\[
\left.  \left\vert \int\limits_{1}^{\infty}h\left(  f_{+},g_{+};t\right)
dt\right\vert \leq C\left(  \left\Vert \hat{f}\right\Vert _{L^{\infty}\left(
\left\vert p\right\vert \leq1\right)  }+\left\Vert \hat{f}\right\Vert
_{\mathcal{H}_{3/2+\varepsilon}^{1}\left(  \mathbb{R}^{3}\right)  }\right)
\left(  \left\Vert \hat{g}\right\Vert _{L^{\infty}\left(  \left\vert
p\right\vert \leq1\right)  }+\left\Vert \hat{g}\right\Vert _{\mathcal{H}%
_{3/2+\varepsilon}^{1}\left(  \mathbb{R}^{3}\right)  }\right)  \right.
\]
and then, using the Sobolev embedding theorem, we get (\ref{t1}) for
$f,g\in\mathcal{S}$. Hence, by continuity, we extend the estimate (\ref{t1})
for all $f,g\in\mathcal{H}_{3/2+\varepsilon}^{3/2+\varepsilon}$ and therefore,
we arrive to (\ref{time11}).
\end{proof}

We now study the asymptotics of $I\left(  R\right)  $, as $R\rightarrow
\infty.$ We have the following

\begin{theorem}
\label{t58}Let $f,g\in\mathcal{H}_{2}^{3/2+\varepsilon}$, $\varepsilon>0$.
Then, as $R\rightarrow\infty,$%
\begin{equation}
\left.
\begin{array}
[c]{c}%
I\left(  R\right)  =R%
{\displaystyle\int\limits_{\mathbb{R}^{3}}}
\left\langle \dfrac{\sqrt{\left\vert p\right\vert ^{2}+m^{2}}}{\left\vert
p\right\vert }f_{+}\left(  p\right)  ,g_{+}\left(  p\right)  \right\rangle
dp+R%
{\displaystyle\int\limits_{\mathbb{R}^{3}}}
\left\langle \dfrac{\sqrt{\left\vert p\right\vert ^{2}+m^{2}}}{\left\vert
p\right\vert }f_{-}\left(  p\right)  ,g_{-}\left(  p\right)  \right\rangle
dp\\
+i%
{\displaystyle\int\limits_{\mathbb{R}^{3}}}
\left\langle f_{+}\left(  p\right)  ,\dfrac{\sqrt{\left\vert p\right\vert
^{2}+m^{2}}}{2\left\vert p\right\vert ^{2}}g_{+}\left(  p\right)
+\dfrac{\sqrt{\left\vert p\right\vert ^{2}+m^{2}}}{2\left\vert p\right\vert
^{2}}p\cdot\triangledown\left(  g_{+}\left(  p\right)  \right)  +\dfrac
{1}{2\left\vert p\right\vert ^{2}}p\cdot\triangledown\left(  \sqrt{\left\vert
p\right\vert ^{2}+m^{2}}g_{+}\left(  p\right)  \right)  \right\rangle dp\\
-i%
{\displaystyle\int\limits_{\mathbb{R}^{3}}}
\left\langle f_{-}\left(  p\right)  ,\dfrac{\sqrt{\left\vert p\right\vert
^{2}+m^{2}}}{2\left\vert p\right\vert ^{2}}g_{-}\left(  p\right)
+\dfrac{\sqrt{\left\vert p\right\vert ^{2}+m^{2}}}{2\left\vert p\right\vert
^{2}}p\cdot\triangledown\left(  g_{-}\left(  p\right)  \right)  +\dfrac
{1}{2\left\vert p\right\vert ^{2}}p\cdot\triangledown\left(  \sqrt{\left\vert
p\right\vert ^{2}+m^{2}}g_{-}\left(  p\right)  \right)  \right\rangle
dp+o\left(  1\right)  ,
\end{array}
\right.  \label{t55}%
\end{equation}
where $f_{\pm}\left(  p\right)  =P_{\pm}\left(  p\right)  \hat{f}\left(
p\right)  $ and $g_{\pm}\left(  p\right)  =P_{\pm}\left(  p\right)  \hat
{g}\left(  p\right)  .$
\end{theorem}

\begin{proof}
We decompose $I\left(  R\right)  $ as in relation (\ref{time9}). Let us
consider the term $I_{1}.$ Let $\varphi\left(  s\right)  \in C_{0}^{\infty
}\left(  \mathbb{R}\right)  ,$ be such that $\varphi\left(  s\right)  =1$ in
some neighborhood of $s=0.$ Observe that%
\[
\left.
\begin{array}
[c]{c}%
I_{1}\left(  R\right)  =\lim\limits_{\varepsilon\rightarrow0}%
{\displaystyle\int\limits_{0}^{\infty}}
{\displaystyle\int\limits_{\mathbb{R}^{3}}}
{\displaystyle\int\limits_{\mathbb{R}^{3}}}
e^{-i\left(  \left\vert p\right\vert -\left\vert q\right\vert \right)
t}\varphi\left(  \varepsilon\dfrac{\sqrt{\left\vert p\right\vert ^{2}+m^{2}%
}+\sqrt{\left\vert q\right\vert ^{2}+m^{2}}}{\left\vert p\right\vert
+\left\vert q\right\vert }t\right) \\
\times\left\langle \dfrac{\sqrt{\left\vert p\right\vert ^{2}+m^{2}}%
+\sqrt{\left\vert q\right\vert ^{2}+m^{2}}}{\left\vert p\right\vert
+\left\vert q\right\vert }f_{+}\left(  p\right)  ,\tilde{\zeta}_{R}\left(
p-q\right)  g_{+}\left(  q\right)  \right\rangle dqdpdt.
\end{array}
\right.
\]

Proceeding as in (\ref{time10}), we show that
\[
\int\limits_{\mathbb{R}^{3}}\int\limits_{\mathbb{R}^{3}}e^{-i\left(
\left\vert p\right\vert -\left\vert q\right\vert \right)  t}\varphi\left(
\varepsilon\frac{\sqrt{\left\vert p\right\vert ^{2}+m^{2}}+\sqrt{\left\vert
q\right\vert ^{2}+m^{2}}}{\left\vert p\right\vert +\left\vert q\right\vert
}t\right)  \left\langle \frac{\sqrt{\left\vert p\right\vert ^{2}+m^{2}}%
+\sqrt{\left\vert q\right\vert ^{2}+m^{2}}}{\left\vert p\right\vert
+\left\vert q\right\vert }f_{+}\left(  p\right)  ,\tilde{\zeta}_{R}\left(
p-q\right)  g_{+}\left(  q\right)  \right\rangle dqdp
\]
belongs to $L^{1},$ as a function of $t,$ uniformly on $\varepsilon\leq1.$
Then, it follows from the dominated convergence theorem that%
\[
I_{1}\left(  R\right)  =\int\limits_{0}^{\infty}\int\limits_{\mathbb{R}^{3}%
}\int\limits_{\mathbb{R}^{3}}\left\langle e^{-i\left(  \left\vert p\right\vert
-\left\vert q\right\vert \right)  t}\frac{\sqrt{\left\vert p\right\vert
^{2}+m^{2}}+\sqrt{\left\vert q\right\vert ^{2}+m^{2}}}{\left\vert p\right\vert
+\left\vert q\right\vert }f_{+}\left(  p\right)  ,\tilde{\zeta}_{R}\left(
p-q\right)  g_{+}\left(  q\right)  \right\rangle dqdpdt.
\]
Let the function $F\left(  s\right)  ,$ $s\in\mathbb{R}$, be such that
$F\left(  s\right)  =1$ for $0\leq s<\infty$ and $F\left(  s\right)  =0$ for
$s<0$. Passing to the spherical coordinate system in the $q$ variable in the
expression for $I_{1}\left(  R\right)  $, we get
\[
\left.
\begin{array}
[c]{c}%
I_{1}\left(  R\right)  =\lim\limits_{\tau\rightarrow+0}%
{\displaystyle\int\limits_{\mathbb{R}^{3}}}
{\displaystyle\int\limits_{\mathbb{S}^{2}}}
{\displaystyle\int\limits_{0}^{\infty}}
\left(  \int\limits_{0}^{\infty}e^{-i\left(  \left\vert p\right\vert
-r\right)  t}e^{-\tau t}dt\right)  \left\langle \dfrac{\sqrt{\left\vert
p\right\vert ^{2}+m^{2}}+\sqrt{r^{2}+m^{2}}}{\left\vert p\right\vert +r}%
f_{+}\left(  p\right)  ,\tilde{\zeta}_{R}\left(  p-r\omega\right)
g_{+}\left(  r\omega\right)  \right\rangle r^{2}drdpd\omega\\
=\lim\limits_{\tau\rightarrow+0}%
{\displaystyle\int\limits_{\mathbb{R}^{3}}}
{\displaystyle\int\limits_{\mathbb{S}^{2}}}
{\displaystyle\int\limits_{-\infty}^{\infty}}
\left(
{\displaystyle\int\limits_{0}^{\infty}}
e^{-irt}e^{-\tau t}dt\right)  f\left(  \left(  \left\vert p\right\vert
-r\right)  \omega,p\right)  drdpd\omega,
\end{array}
\right.
\]
where%
\[
f\left(  q,p\right)  :=\left\langle \frac{\sqrt{\left\vert p\right\vert
^{2}+m^{2}}+\sqrt{\left\vert q\right\vert ^{2}+m^{2}}}{\left\vert p\right\vert
+\left\vert q\right\vert }f_{+}\left(  p\right)  ,\tilde{\zeta}_{R}\left(
p-q\right)  g_{+}\left(  q\right)  \right\rangle F\left(  \left\vert
q\right\vert \right)  \left\vert q\right\vert ^{2}.
\]
Moreover, as
\[
\left.
\begin{array}
[c]{c}%
\lim\limits_{\tau\rightarrow+0}%
{\displaystyle\int\limits_{\mathbb{R}^{3}}}
{\displaystyle\int\limits_{\mathbb{S}^{2}}}
{\displaystyle\int\limits_{-\infty}^{\infty}}
\left(
{\displaystyle\int\limits_{0}^{\infty}}
e^{-irt}e^{-\tau t}dt\right)  f\left(  \left(  \left\vert p\right\vert
-r\right)  \omega,p\right)  drd\omega dp=-\lim\limits_{\tau\rightarrow+0}i%
{\displaystyle\int\limits_{\mathbb{R}^{3}}}
{\displaystyle\int\limits_{\mathbb{S}^{2}}}
{\displaystyle\int\limits_{-\infty}^{\infty}}
\dfrac{f\left(  \left(  \left\vert p\right\vert -r\right)  \omega,p\right)
}{r-i\tau}drd\omega dp\\
=-\lim\limits_{\delta\rightarrow0}\lim\limits_{\tau\rightarrow+0}%
{\displaystyle\int\limits_{\mathbb{R}^{3}}}
{\displaystyle\int\limits_{\mathbb{S}^{2}}}
if\left(  \left\vert p\right\vert \omega,p\right)  \left(
{\displaystyle\int\limits_{-\delta}^{\delta}}
\dfrac{dr}{r-i\tau}\right)  d\omega dp-i\lim\limits_{\delta\rightarrow0}%
{\displaystyle\int\limits_{\mathbb{R}^{3}}}
{\displaystyle\int\limits_{\mathbb{S}^{2}}}
\left[
{\displaystyle\int\limits_{-\infty}^{-\delta}}
+%
{\displaystyle\int\limits_{\delta}^{\infty}}
\right]  \dfrac{f\left(  \left(  \left\vert p\right\vert -r\right)
\omega,p\right)  }{r}drd\omega dp\\
-\lim\limits_{\delta\rightarrow0}\lim\limits_{\tau\rightarrow+0}%
{\displaystyle\int\limits_{\mathbb{R}^{3}}}
{\displaystyle\int\limits_{\mathbb{S}^{2}}}
i%
{\displaystyle\int\limits_{-\delta}^{\delta}}
\dfrac{f\left(  \left(  \left\vert p\right\vert -r\right)  \omega,p\right)
-f\left(  \left\vert p\right\vert \omega,p\right)  }{r-i\tau}drd\omega dp,
\end{array}
\right.
\]
we conclude that%
\[
I_{1}\left(  R\right)  =I_{1,1}\left(  R\right)  +I_{1,2}\left(  R\right)  ,
\]
with%
\[
I_{1,1}\left(  R\right)  :=\pi\int\limits_{\mathbb{R}^{3}}\int
\limits_{\mathbb{S}^{2}}f\left(  \left\vert p\right\vert \omega,p\right)
d\omega dp
\]
and
\[
I_{1,2}\left(  R\right)  :=-i\lim\limits_{\delta\rightarrow0}\int
\limits_{\mathbb{R}^{3}}\int\limits_{\mathbb{S}^{2}}\left[  \int
\limits_{-\infty}^{-\delta}+\int\limits_{\delta}^{\infty}\right]
\frac{f\left(  \left(  \left\vert p\right\vert -r\right)  \omega,p\right)
}{r}drd\omega dp.
\]
Using Lemma \ref{t54} for $I_{1,1}\left(  R\right)  $ we obtain the first term
in the R.H.S. of the asymptotic expansion (\ref{t55}). Decomposing
$I_{1,2}\left(  R\right)  $ as in the sum (\ref{t56}) and applying Lemma
\ref{t45} to $I_{1,2}^{1}\left(  R\right)  ,$ Lemma \ref{t57} to $I_{1,2}%
^{2}\left(  R\right)  $ and Lemma \ref{t49} to $I_{1,2}^{3}\left(  R\right)  $
we get the third term in the R.H.S. of (\ref{t55}).

Now note that
\[
\left.
\begin{array}
[c]{c}%
I_{4}\left(  R\right)  =\lim\limits_{\tau\rightarrow+0}\int\limits_{\mathbb{R}%
^{3}}\int\limits_{\mathbb{S}^{2}}\int\limits_{0}^{\infty}\left(
\int\limits_{0}^{\infty}e^{i\left(  \left\vert p\right\vert -r\right)
t}e^{-\tau t}dt\right)  \left\langle \frac{\sqrt{\left\vert p\right\vert
^{2}+m^{2}}+\sqrt{r^{2}+m^{2}}}{\left\vert p\right\vert +r}f_{-}\left(
p\right)  ,\tilde{\zeta}_{R}\left(  p-r\omega\right)  g_{-}\left(
r\omega\right)  \right\rangle r^{2}drdpd\omega\\
=\lim\limits_{\tau\rightarrow+0}\int\limits_{\mathbb{R}^{3}}\int
\limits_{\mathbb{S}^{2}}\int\limits_{-\infty}^{\infty}\left(  \int
\limits_{0}^{\infty}e^{irt}e^{-\tau t}dt\right)  f_{1}\left(  \left(
\left\vert p\right\vert -r\right)  \omega,p\right)  drdpd\omega=\lim
\limits_{\tau\rightarrow+0}i\int\limits_{\mathbb{R}^{3}}\int
\limits_{\mathbb{S}^{2}}\int\limits_{-\infty}^{\infty}\frac{f_{1}\left(
\left(  \left\vert p\right\vert -r\right)  \omega,p\right)  }{r+i\tau
}drd\omega dp\\
=i\lim\limits_{\delta\rightarrow0}\int\limits_{\mathbb{R}^{3}}\int
\limits_{\mathbb{S}^{2}}\left[  \int\limits_{-\infty}^{-\delta}+\int
\limits_{\delta}^{\infty}\right]  \frac{f_{1}\left(  \left(  \left\vert
p\right\vert -r\right)  \omega,p\right)  }{r}drd\omega dp+\lim\limits_{\delta
\rightarrow0}\lim\limits_{\tau\rightarrow+0}\int\limits_{\mathbb{R}^{3}}%
\int\limits_{\mathbb{S}^{2}}i\int\limits_{-\delta}^{\delta}\frac{f_{1}\left(
\left(  \left\vert p\right\vert -r\right)  \omega,p\right)  -f_{1}\left(
\left\vert p\right\vert \omega,p\right)  }{r+i\tau}drd\omega dp\\
+\lim\limits_{\delta\rightarrow0}\lim\limits_{\tau\rightarrow+0}%
\int\limits_{\mathbb{R}^{3}}\int\limits_{\mathbb{S}^{2}}if_{1}\left(
\left\vert p\right\vert \omega,p\right)  \left(  \int\limits_{-\delta}%
^{\delta}\frac{dr}{r+i\tau}\right)  d\omega dp,
\end{array}
\right.
\]
where $f_{1}\left(  q,p\right)  :=\left\langle \frac{\sqrt{\left\vert
p\right\vert ^{2}+m^{2}}+\sqrt{\left\vert q\right\vert ^{2}+m^{2}}}{\left\vert
p\right\vert +\left\vert q\right\vert }f_{-}\left(  p\right)  ,\tilde{\zeta
}_{R}\left(  p-q\right)  g_{-}\left(  q\right)  \right\rangle F\left(
\left\vert q\right\vert \right)  \left\vert q\right\vert ^{2},$ and then
\[
I_{4}\left(  R\right)  =I_{4,1}\left(  R\right)  +I_{4,2}\left(  R\right)  ,
\]
where%
\[
I_{4}\left(  R\right)  :=\pi\int\limits_{\mathbb{R}^{3}}\int
\limits_{\mathbb{S}^{2}}f_{1}\left(  \left\vert p\right\vert \omega,p\right)
d\omega dp
\]
and
\[
I_{4,2}\left(  R\right)  :=i\lim\limits_{\delta\rightarrow0}\int
\limits_{\mathbb{R}^{3}}\int\limits_{\mathbb{S}^{2}}\left[  \int
\limits_{-\infty}^{-\delta}+\int\limits_{\delta}^{\infty}\right]  \frac
{f_{1}\left(  \left(  \left\vert p\right\vert -r\right)  \omega,p\right)  }%
{r}drd\omega dp.
\]
Thus, similarly to the case of $I_{1}\left(  R\right)  $ above, we obtain the
second and fourth terms in the R.H.S. of the asymptotic expansion (\ref{t55}).
Finally, applying Lemma \ref{t48} to $I_{2}\left(  R\right)  $ and
$I_{3}\left(  R\right)  ,$ we complete the proof.
\end{proof}

Let us now check the results that we use in the proof of Theorem \ref{t58}.
First, we calculate the asymptotics of $I_{1,1}\left(  R\right)  $ as
$R\rightarrow\infty.$ We have

\begin{lemma}
\label{t54}Suppose that $f,g\in\mathcal{H}_{2}^{3/2+\varepsilon}$,
$\varepsilon>0.$\ The following relation holds%
\begin{equation}
I_{1,1}\left(  R\right)  =8\pi^{3}R\int\limits_{\mathbb{R}^{3}}\left\langle
\frac{\sqrt{\left\vert p\right\vert ^{2}+m^{2}}}{\left\vert p\right\vert
}f_{+}\left(  p\right)  ,g_{+}\left(  p\right)  \right\rangle dp+o\left(
1\right)  ,\text{ \ as }R\rightarrow\infty. \label{time70}%
\end{equation}

\end{lemma}

\begin{proof}
Noting that $\tilde{\zeta}_{R}\left(  p\right)  =4\pi\left(  -R\frac{\cos
R\left\vert p\right\vert }{\left\vert p\right\vert ^{2}}+\frac{\sin
R\left\vert p\right\vert }{\left\vert p\right\vert ^{3}}\right)  $ (see
Theorem 56, page 235 of \cite{Bochner}) and passing to the spherical
coordinate system in $\omega$, where the $z-$axis is directed along the vector
$p,$ we have%
\[
\left.  I_{1,1}\left(  R\right)  =4\pi^{2}\int\limits_{\mathbb{R}^{3}}%
\int\limits_{0}^{2\pi}\int\limits_{0}^{\pi}\left\langle f_{0}\left(  p\right)
,\left(  -R\frac{\cos\left(  R\left\vert p\right\vert \sqrt{2-2\cos\theta
}\right)  }{\left(  2-2\cos\theta\right)  }+\frac{\sin\left(  R\left\vert
p\right\vert \sqrt{2-2\cos\theta}\right)  }{\left\vert p\right\vert \left(
2-2\cos\theta\right)  ^{\frac{3}{2}}}\right)  g_{+}\left(  \left\vert
p\right\vert \omega(\theta,\varphi)\right)  \right\rangle \sin\theta d\theta
d\varphi dp,\right.
\]
where $f_{0}\left(  p\right)  :=\frac{\sqrt{\left\vert p\right\vert ^{2}%
+m^{2}}}{\left\vert p\right\vert }f_{+}\left(  p\right)  $ and $\omega
(\theta,\varphi):=\left(  \cos\varphi\sin\theta,\sin\varphi\sin\theta
,\cos\theta\right)  .$ Observing that
\[
\left(  -R\frac{\cos\left(  R\left\vert p\right\vert \sqrt{2-2\cos\theta
}\right)  }{\left(  2-2\cos\theta\right)  }+\frac{\sin\left(  R\left\vert
p\right\vert \sqrt{2-2\cos\theta}\right)  }{\left\vert p\right\vert \left(
2-2\cos\theta\right)  ^{\frac{3}{2}}}\right)  \sin\theta=-\partial_{\theta
}\frac{\sin\left(  R\left\vert p\right\vert \sqrt{2-2\cos\theta}\right)
}{\left\vert p\right\vert \sqrt{2-2\cos\theta}}%
\]
and integrating by parts in $\theta$ we get%
\begin{equation}
\left.  I_{1,1}\left(  R\right)  =8\pi^{3}R\int\limits_{\mathbb{R}^{3}%
}\left\langle f_{0}\left(  p\right)  ,g_{+}\left(  p\right)  \right\rangle
dp+I_{1,1}^{1}\left(  R\right)  +I_{1,1}^{2}\left(  R\right)  ,\right.
\label{time71}%
\end{equation}
where%
\begin{equation}
I_{1,1}^{1}\left(  R\right)  :=-8\pi^{3}\int\limits_{\mathbb{R}^{3}}%
\sin\left(  2R\left\vert p\right\vert \right)  \left\langle f_{0}\left(
p\right)  ,\frac{g_{+}\left(  -p\right)  }{2\left\vert p\right\vert
}\right\rangle dp, \label{t111}%
\end{equation}
and%
\[
I_{1,1}^{2}\left(  R\right)  :=4\pi^{2}\int\limits_{\mathbb{R}^{3}}%
\int\limits_{0}^{2\pi}\int\limits_{0}^{\pi}\left\langle f_{0}\left(  p\right)
,\frac{\sin\left(  R\left\vert p\right\vert \sqrt{2-2\cos\theta}\right)
}{\left\vert p\right\vert \sqrt{2-2\cos\theta}}\partial_{\theta}g_{+}\left(
\left\vert p\right\vert \omega(\theta,\varphi)\right)  \right\rangle d\theta
d\varphi dp.
\]
Note that
\begin{equation}
\left.
{\displaystyle\int\limits_{\mathbb{R}^{3}}}
\left\vert \left\langle f_{0}\left(  p\right)  ,g_{+}\left(  -p\right)
\right\rangle \right\vert \dfrac{1}{2\left\vert p\right\vert }dp\leq C\left(
\left\Vert f_{+}\right\Vert _{L^{\infty}\left(  \left\vert p\right\vert
\leq1\right)  }+\left\Vert f_{+}\right\Vert _{L^{2}\left(  \mathbb{R}%
^{3}\right)  }\right)  \left(  \left\Vert g_{+}\right\Vert _{L^{\infty}\left(
\left\vert p\right\vert \leq1\right)  }+\left\Vert g_{+}\right\Vert
_{L^{2}\left(  \mathbb{R}^{3}\right)  }\right)  <\infty.\right.  \label{t112}%
\end{equation}
Taking $p=\left\vert p\right\vert \omega$ in the integral in (\ref{t111}), it
follows from (\ref{t112}) and Fubini's theorem that%
\[
\int\limits_{0}^{\infty}\left\langle f_{0}\left(  p\right)  ,\frac
{g_{+}\left(  -p\right)  }{2\left\vert p\right\vert }\right\rangle \left\vert
p\right\vert ^{2}d\left\vert p\right\vert \in L^{1}\left(  \mathbb{S}%
^{2}\right)  .
\]
In particular, we have%
\[
\int\limits_{0}^{\infty}\left\langle f_{0}\left(  p\right)  ,\frac
{g_{+}\left(  -p\right)  }{2\left\vert p\right\vert }\right\rangle \left\vert
p\right\vert ^{2}d\left\vert p\right\vert <\infty,
\]
for almost all $\omega\in\mathbb{S}^{2}.$ Then by the Riemann-Lebesgue lemma
we get%
\[
\lim\limits_{R\rightarrow\infty}\int\limits_{0}^{\infty}\sin\left(
2R\left\vert p\right\vert \right)  \left\langle f_{0}\left(  p\right)
,\frac{g_{+}\left(  -p\right)  }{2\left\vert p\right\vert }\right\rangle
\left\vert p\right\vert ^{2}d\left\vert p\right\vert =0,
\]
a.e. in $\omega\in\mathbb{S}^{2}.$ Thus, using (\ref{t112}) to apply the
dominated convergence theorem in (\ref{t111}) we obtain%
\begin{equation}
\lim\limits_{R\rightarrow\infty}I_{1,1}^{1}\left(  R\right)  =0. \label{t2}%
\end{equation}

Let us consider now $I_{1,1}^{2}\left(  R\right)  .$ Note that
\begin{equation}
\left.
\begin{array}
[c]{c}%
I_{1,1}^{2}\left(  R\right)  =4\pi^{2}%
{\displaystyle\int\limits_{\mathbb{R}^{3}}}
{\displaystyle\int\limits_{0}^{2\pi}}
{\displaystyle\int\limits_{0}^{\pi}}
\left\langle f_{0}\left(  p\right)  ,\dfrac{\sin\left(  R\left\vert
p\right\vert \sqrt{2-2\cos\theta}\right)  }{2\sqrt{2-2\cos\theta}}\left(
\cos\varphi,\sin\varphi,0\right)  \cdot\nabla g_{+}\left(  \left\vert
p\right\vert \omega(\theta,\varphi)\right)  \right\rangle \cos\theta d\theta
d\varphi dp\\
-4\pi^{2}%
{\displaystyle\int\limits_{\mathbb{R}^{3}}}
{\displaystyle\int\limits_{0}^{2\pi}}
{\displaystyle\int\limits_{0}^{\pi}}
\left\langle f_{0}\left(  p\right)  ,\dfrac{\sqrt{2+2\cos\theta}\sin\left(
R\left\vert p\right\vert \sqrt{2-2\cos\theta}\right)  }{2}\left(
0,0,1\right)  \cdot\nabla g_{+}\left(  \left\vert p\right\vert \omega
(\theta,\varphi)\right)  \right\rangle d\theta d\varphi dp.
\end{array}
\right.  \label{t97}%
\end{equation}
We have%
\[
\int\limits_{\mathbb{R}^{3}}\left\vert f_{0}\left(  p\right)  \right\vert
\int\limits_{0}^{2\pi}\int\limits_{0}^{\pi}\left\vert \nabla g_{+}\left(
\left\vert p\right\vert \omega(\theta,\varphi)\right)  \right\vert d\theta
d\varphi dp\leq C\left(  \left\Vert f_{+}\right\Vert _{L^{\infty}\left(
\left\vert p\right\vert \leq1\right)  }+\left\Vert f_{+}\right\Vert
_{L^{2}\left(  \mathbb{R}^{3}\right)  }\right)  \left\Vert g_{+}\right\Vert
_{\mathcal{H}^{1}\left(  \mathbb{R}^{3}\right)  }.
\]
Then, arguing as in the case of (\ref{t2}), we get
\begin{equation}
\lim\limits_{R\rightarrow\infty}4\pi^{2}\int\limits_{\mathbb{R}^{3}}%
\int\limits_{0}^{2\pi}\int\limits_{0}^{\pi}\left\langle f_{0}\left(  p\right)
,\frac{\sqrt{2+2\cos\theta}\sin\left(  R\left\vert p\right\vert \sqrt
{2-2\cos\theta}\right)  }{2}\left(  0,0,1\right)  \cdot\nabla g_{+}\left(
\left\vert p\right\vert \omega(\theta,\varphi)\right)  \right\rangle d\theta
d\varphi dp=0. \label{t98}%
\end{equation}
Observe now that%
\[
\left.
\begin{array}
[c]{c}%
4\pi^{2}%
{\displaystyle\int\limits_{\mathbb{R}^{3}}}
{\displaystyle\int\limits_{0}^{2\pi}}
{\displaystyle\int\limits_{0}^{\pi}}
\left\langle f_{0}\left(  p\right)  ,\dfrac{\sin\left(  R\left\vert
p\right\vert \sqrt{2-2\cos\theta}\right)  }{2\sqrt{2-2\cos\theta}}\left(
\cos\varphi,\sin\varphi,0\right)  \cdot\nabla g_{+}\left(  \left\vert
p\right\vert \omega(\theta,\varphi)\right)  \right\rangle \cos\theta d\theta
d\varphi dp\\
=4\pi^{2}%
{\displaystyle\int\limits_{\mathbb{R}^{3}}}
{\displaystyle\int\limits_{0}^{2\pi}}
{\displaystyle\int\limits_{0}^{\pi}}
\left\langle f_{0}\left(  p\right)  ,\dfrac{\sin\left(  R\left\vert
p\right\vert \sqrt{2-2\cos\theta}\right)  }{2\sqrt{2-2\cos\theta}}\left(
\cos\varphi,\sin\varphi,0\right)  \cdot\left(  \nabla g_{+}\left(  \left\vert
p\right\vert \omega(\theta,\varphi)\right)  -\nabla g_{+}\left(  p\right)
\right)  \right\rangle \cos\theta d\theta d\varphi dp.
\end{array}
\right.
\]
Thus, as%
\[
\left.
\begin{array}
[c]{c}%
{\displaystyle\int\limits_{\mathbb{R}^{3}}}
{\displaystyle\int\limits_{0}^{2\pi}}
{\displaystyle\int\limits_{0}^{\pi}}
\dfrac{\left\vert f_{0}\left(  p\right)  \right\vert }{\sqrt{2-2\cos\theta}%
}\left\vert \nabla g_{+}\left(  \left\vert p\right\vert \omega(\theta
,\varphi)\right)  -\nabla g_{+}\left(  p\right)  \right\vert d\theta d\varphi
dp\\
\leq C\left(
{\displaystyle\int\limits_{0}^{\pi}}
\frac{\left\vert \theta\right\vert ^{1/2}}{\sqrt{2-2\cos\theta}}%
d\theta\right)
{\displaystyle\int\limits_{\mathbb{R}^{3}}}
\left\vert f_{0}\left(  p\right)  \right\vert \left(
{\displaystyle\int\limits_{0}^{2\pi}}
{\displaystyle\int\limits_{0}^{\pi}}
\left\vert \partial_{\theta}\left(  \nabla g_{+}\left(  \left\vert
p\right\vert \omega(\theta,\varphi)\right)  \right)  \right\vert ^{2}d\theta
d\varphi\right)  ^{1/2}dp\leq C\left\Vert f_{+}\right\Vert _{L_{1}^{2}\left(
\mathbb{R}^{3}\right)  }\left\Vert g_{+}\right\Vert _{\mathcal{H}^{2}\left(
\mathbb{R}^{3}\right)  },
\end{array}
\right.
\]
arguing as in (\ref{t2}), we see that the limit of the first term in the
R.H.S. of (\ref{t97}), as $R\rightarrow\infty,$ is equals to $0.$ Therefore,
passing to the limit, as $R\rightarrow\infty,$ in (\ref{t97}) and using
(\ref{t98}) we conclude that
\begin{equation}
\lim\limits_{R\rightarrow\infty}I_{1,1}^{2}\left(  R\right)  =0. \label{t3}%
\end{equation}
Using relations (\ref{t2}) and (\ref{t3}) in (\ref{time71}) we obtain
(\ref{time70}).
\end{proof}

Next we study the asymptotics of $I_{1,2}\left(  R\right)  $ as $R\rightarrow
\infty.$ Passing to the spherical coordinate system, where the $z-$axis is
directed along the vector $p,$ we obtain
\[
\left.
\begin{array}
[c]{c}%
I_{1,2}\left(  R\right)  =-4\pi i\lim\limits_{\delta\rightarrow0}%
{\displaystyle\int\limits_{\mathbb{R}^{3}}}
{\displaystyle\int\limits_{0}^{2\pi}}
{\displaystyle\int\limits_{0}^{\pi}}
\left[
{\displaystyle\int\limits_{-\infty}^{\left\vert p\right\vert -\delta}}
+%
{\displaystyle\int\limits_{\left\vert p\right\vert +\delta}^{\infty}}
\right] \\
\times\dfrac{\left\langle f_{+}\left(  p\right)  ,\left(  -R\frac{\cos\left(
R\sqrt{\left\vert p\right\vert ^{2}-2r\left\vert p\right\vert \cos\theta
+r^{2}}\right)  }{\left\vert p\right\vert ^{2}-2r\left\vert p\right\vert
\cos\theta+r^{2}}+\frac{\sin\left(  R\sqrt{\left\vert p\right\vert
^{2}-2r\left\vert p\right\vert \cos\theta+r^{2}}\right)  }{\left(  \left\vert
p\right\vert ^{2}-2r\left\vert p\right\vert \cos\theta+r^{2}\right)  ^{3/2}%
}\right)  g_{0}\left(  r\omega(\theta,\varphi);\left\vert p\right\vert
\right)  \sin\theta\right\rangle }{\left(  \left\vert p\right\vert -r\right)
\left(  \left\vert p\right\vert +r\right)  }F\left(  r\right)  r^{2}drd\theta
d\varphi dp,
\end{array}
\right.
\]
where $g_{0}\left(  q;\left\vert p\right\vert \right)  :=\left(
\sqrt{\left\vert p\right\vert ^{2}+m^{2}}+\sqrt{\left\vert q\right\vert
^{2}+m^{2}}\right)  g_{+}\left(  q\right)  $ $\ $and $\omega(\theta
,\varphi)=\left(  \cos\varphi\sin\theta,\sin\varphi\sin\theta,\cos
\theta\right)  .$ Noting that
\[
\left(  -R\frac{\cos\left(  R\sqrt{\left\vert p\right\vert ^{2}-2r\left\vert
p\right\vert \cos\theta+r^{2}}\right)  }{\left\vert p\right\vert
^{2}-2r\left\vert p\right\vert \cos\theta+r^{2}}+\frac{\sin\left(
R\sqrt{\left\vert p\right\vert ^{2}-2r\left\vert p\right\vert \cos\theta
+r^{2}}\right)  }{\left(  \left\vert p\right\vert ^{2}-2r\left\vert
p\right\vert \cos\theta+r^{2}\right)  ^{3/2}}\right)  \sin\theta
=-\partial_{\theta}\frac{\sin\left(  R\sqrt{\left\vert p\right\vert
^{2}-2r\left\vert p\right\vert \cos\theta+r^{2}}\right)  }{r\left\vert
p\right\vert \sqrt{\left\vert p\right\vert ^{2}-2r\left\vert p\right\vert
\cos\theta+r^{2}}}%
\]
we have \
\begin{equation}
\left.
\begin{array}
[c]{c}%
{\displaystyle\int\limits_{0}^{\pi}}
\tilde{\zeta}_{R}\left(  p-r\omega(\theta,\varphi)\right)  g_{0}\left(
r\omega(\theta,\varphi);\left\vert p\right\vert \right)  \sin\theta d\theta\\
=%
{\displaystyle\int\limits_{0}^{\pi}}
\left(  -R\frac{\cos\left(  R\sqrt{\left\vert p\right\vert ^{2}-2r\left\vert
p\right\vert \cos\theta+r^{2}}\right)  }{\left\vert p\right\vert
^{2}-2r\left\vert p\right\vert \cos\theta+r^{2}}+\frac{\sin\left(
R\sqrt{\left\vert p\right\vert ^{2}-2r\left\vert p\right\vert \cos\theta
+r^{2}}\right)  }{\left(  \left\vert p\right\vert ^{2}-2r\left\vert
p\right\vert \cos\theta+r^{2}\right)  ^{3/2}}\right)  g_{0}\left(
r\omega(\theta,\varphi);\left\vert p\right\vert \right)  \sin\theta d\theta\\
=-%
{\displaystyle\int\limits_{0}^{\pi}}
\partial_{\theta}\frac{\sin\left(  R\sqrt{\left\vert p\right\vert
^{2}-2r\left\vert p\right\vert \cos\theta+r^{2}}\right)  }{r\left\vert
p\right\vert \sqrt{\left\vert p\right\vert ^{2}-2r\left\vert p\right\vert
\cos\theta+r^{2}}}g_{0}\left(  r\omega(\theta,\varphi);\left\vert p\right\vert
\right)  d\theta\\
=-\dfrac{\sin\left(  R\left(  \left\vert p\right\vert +r\right)  \right)
}{r\left\vert p\right\vert \left(  \left\vert p\right\vert +r\right)  }%
g_{0}\left(  -r\frac{p}{\left\vert p\right\vert };\left\vert p\right\vert
\right)  +\dfrac{\sin\left(  R\left(  \left\vert p\right\vert -r\right)
\right)  }{r\left\vert p\right\vert \left(  \left\vert p\right\vert -r\right)
}g_{0}\left(  r\frac{p}{\left\vert p\right\vert };\left\vert p\right\vert
\right)  +%
{\displaystyle\int\limits_{0}^{\pi}}
\frac{\sin\left(  R\sqrt{\left\vert p\right\vert ^{2}-2r\left\vert
p\right\vert \cos\theta+r^{2}}\right)  }{r\left\vert p\right\vert
\sqrt{\left\vert p\right\vert ^{2}-2r\left\vert p\right\vert \cos\theta+r^{2}%
}}\partial_{\theta}g_{0}\left(  r\omega(\theta,\varphi);\left\vert
p\right\vert \right)  d\theta,
\end{array}
\right.  \label{t28}%
\end{equation}
and therefore, we obtain the following decomposition
\begin{equation}
I_{1,2}\left(  R\right)  =I_{1,2}^{1}\left(  R\right)  +I_{1,2}^{2}\left(
R\right)  +I_{1,2}^{3}\left(  R\right)  , \label{t56}%
\end{equation}
with
\[
I_{1,2}^{1}\left(  R\right)  :=8\pi^{2}i\lim\limits_{\delta\rightarrow0}%
\int\limits_{\mathbb{R}^{3}}\left[  \int\limits_{-\infty}^{\left\vert
p\right\vert -\delta}+\int\limits_{\left\vert p\right\vert +\delta}^{\infty
}\right]  \frac{\sin\left(  R\left(  \left\vert p\right\vert +r\right)
\right)  }{\left\vert p\right\vert \left(  \left\vert p\right\vert -r\right)
\left(  \left\vert p\right\vert +r\right)  ^{2}}\left\langle f_{+}\left(
p\right)  ,g_{0}\left(  -r\frac{p}{\left\vert p\right\vert };\left\vert
p\right\vert \right)  \right\rangle F\left(  r\right)  rdrdp,
\]%
\[
I_{1,2}^{2}\left(  R\right)  :=-8\pi^{2}i\lim\limits_{\delta\rightarrow0}%
\int\limits_{\mathbb{R}^{3}}\left[  \int\limits_{-\infty}^{\left\vert
p\right\vert -\delta}+\int\limits_{\left\vert p\right\vert +\delta}^{\infty
}\right]  \frac{\sin\left(  R\left(  \left\vert p\right\vert -r\right)
\right)  }{\left\vert p\right\vert \left(  \left\vert p\right\vert -r\right)
^{2}}\left\langle f_{+}\left(  p\right)  ,\frac{g_{0}\left(  r\frac
{p}{\left\vert p\right\vert };\left\vert p\right\vert \right)  }{\left\vert
p\right\vert +r}\right\rangle F\left(  r\right)  rdrdp
\]
and%
\[
\left.
\begin{array}
[c]{c}%
I_{1,2}^{3}\left(  R\right)  :=-4\pi i\lim\limits_{\delta\rightarrow0}%
{\displaystyle\int\limits_{\mathbb{R}^{3}}}
\left[
{\displaystyle\int\limits_{-\infty}^{\left\vert p\right\vert -\delta}}
+%
{\displaystyle\int\limits_{\left\vert p\right\vert +\delta}^{\infty}}
\right]  \dfrac{\left(  \sqrt{\left\vert p\right\vert ^{2}+m^{2}}+\sqrt
{r^{2}+m^{2}}\right)  }{r\left\vert p\right\vert \left(  \left\vert
p\right\vert -r\right)  \left(  \left\vert p\right\vert +r\right)  }\\
\times\left\langle f_{+}\left(  p\right)  ,%
{\displaystyle\int\limits_{0}^{2\pi}}
{\displaystyle\int\limits_{0}^{\pi}}
\dfrac{\sin\left(  R\sqrt{\left\vert p\right\vert ^{2}-2r\left\vert
p\right\vert \cos\theta+r^{2}}\right)  }{\sqrt{\left\vert p\right\vert
^{2}-2r\left\vert p\right\vert \cos\theta+r^{2}}}\partial_{\theta}g_{+}\left(
r\omega(\theta,\varphi)\right)  d\theta d\varphi\right\rangle F\left(
r\right)  r^{2}drdp.
\end{array}
\right.
\]

For $I_{1,2}^{1}\left(  R\right)  $ we have

\begin{lemma}
\label{t45}Let $f,g\in\mathcal{H}_{3/2+\varepsilon}^{3/2+\varepsilon}$,
$\varepsilon>0.$ Then,
\begin{equation}
\lim\limits_{R\rightarrow\infty}I_{1,2}^{1}\left(  R\right)  =0. \label{t18}%
\end{equation}

\end{lemma}

\begin{proof}
Note that
\begin{equation}
\left.
\begin{array}
[c]{c}%
I_{1,2}^{1}\left(  R\right)  =8\pi^{2}i\lim\limits_{\delta\rightarrow0}%
{\displaystyle\int\limits_{\mathbb{R}^{3}}}
\left\langle f_{+}\left(  p\right)  ,g_{0}\left(  -p;\left\vert p\right\vert
\right)  \right\rangle i_{1,2}^{1}\left(  R,\left\vert p\right\vert
,\delta\right)  dp\\
+8\pi^{2}i\lim\limits_{\delta\rightarrow0}%
{\displaystyle\int\limits_{\mathbb{R}^{3}}}
\left[
{\displaystyle\int\limits_{-\infty}^{\left\vert p\right\vert -\delta}}
+%
{\displaystyle\int\limits_{\left\vert p\right\vert +\delta}^{\infty}}
\right]  \dfrac{\sin\left(  R\left(  \left\vert p\right\vert +r\right)
\right)  }{\left\vert p\right\vert \left(  \left\vert p\right\vert +r\right)
^{2}\left(  \left\vert p\right\vert -r\right)  }\left\langle f_{+}\left(
p\right)  ,%
{\displaystyle\int\limits_{\left\vert p\right\vert }^{r}}
\partial_{r_{1}}\left(  g_{0}\left(  -r_{1}\frac{p}{\left\vert p\right\vert
};\left\vert p\right\vert \right)  r_{1}\right)  dr_{1}\right\rangle F\left(
r\right)  drdp,
\end{array}
\right.  \label{t6}%
\end{equation}
where
\[
i_{1,2}^{1}\left(  R,\left\vert p\right\vert ,\delta\right)  :=\left[
\int\limits_{-\infty}^{\left\vert p\right\vert -\delta}+\int
\limits_{\left\vert p\right\vert +\delta}^{\infty}\right]  \frac{\sin\left(
R\left(  \left\vert p\right\vert +r\right)  \right)  }{\left(  \left\vert
p\right\vert -r\right)  \left(  \left\vert p\right\vert +r\right)  ^{2}%
}F\left(  r\right)  dr.
\]
Since for any $\varepsilon>0,$%
\[
\left.
\begin{array}
[c]{c}%
{\displaystyle\int\limits_{\mathbb{R}^{3}}}
{\displaystyle\int\limits_{0}^{\infty}}
\dfrac{1}{\left(  \left\vert p\right\vert +r\right)  ^{2}\left\vert
r-\left\vert p\right\vert \right\vert }\left\vert \left\langle \dfrac
{f_{+}\left(  p\right)  }{\left\vert p\right\vert },%
{\displaystyle\int\limits_{\left\vert p\right\vert }^{r}}
\partial_{r_{1}}\left(  g_{0}\left(  -r_{1}\frac{p}{\left\vert p\right\vert
};\left\vert p\right\vert \right)  r_{1}\right)  dr_{1}\right\rangle
\right\vert drdp\\
\leq%
{\displaystyle\int\limits_{0}^{\infty}}
{\displaystyle\int\limits_{0}^{\infty}}
\dfrac{dr}{\left(  \left\vert p\right\vert +r\right)  ^{2}\left\vert
r-\left\vert p\right\vert \right\vert ^{1/2}}%
{\displaystyle\int\limits_{\mathbb{S}^{2}}}
\left\vert \dfrac{f_{+}\left(  \left\vert p\right\vert \omega\right)
}{\left\vert p\right\vert }\right\vert \left(
{\displaystyle\int\limits_{0}^{\infty}}
\left\vert \partial_{r_{1}}\left(  g_{0}\left(  -r_{1}\omega;\left\vert
p\right\vert \right)  r_{1}\right)  \right\vert ^{2}dr_{1}\right)
^{1/2}d\omega\left\vert p\right\vert ^{2}d\left\vert p\right\vert \\
\leq\left(
{\displaystyle\int\limits_{0}^{\infty}}
{\displaystyle\int\limits_{\mathbb{S}^{2}}}
\left\vert \partial_{r_{1}}\left(  g_{0}\left(  r_{1}\omega;\left\vert
p\right\vert \right)  r_{1}\right)  \right\vert ^{2}d\omega dr_{1}\right)
^{1/2}%
{\displaystyle\int\limits_{0}^{\infty}}
{\displaystyle\int\limits_{0}^{\infty}}
\dfrac{dr}{\left(  \left\vert p\right\vert +r\right)  ^{2}\left\vert
r-\left\vert p\right\vert \right\vert ^{1/2}}\left(
{\displaystyle\int\limits_{\mathbb{S}^{2}}}
\left\vert \dfrac{f_{+}\left(  \left\vert p\right\vert \omega\right)
}{\left\vert p\right\vert }\right\vert ^{2}d\omega\right)  ^{1/2}\left\vert
p\right\vert ^{2}d\left\vert p\right\vert \\
\leq C\left(  \left\Vert f_{+}\right\Vert _{L^{\infty}\left(  \left\vert
p\right\vert \leq1\right)  }+\left\Vert f_{+}\right\Vert _{L_{\varepsilon}%
^{2}\left(  \mathbb{R}^{3}\right)  }\right)  \left(  \left\Vert g_{+}%
\right\Vert _{L^{\infty}\left(  \left\vert p\right\vert \leq1\right)
}+\left\Vert g_{+}\right\Vert _{\mathcal{H}^{1}}\right)  ,
\end{array}
\right.
\]
arguing as in (\ref{t2}) we get%
\begin{equation}
8\pi^{2}i\lim\limits_{R\rightarrow\infty}\lim\limits_{\delta\rightarrow0}%
\int\limits_{\mathbb{R}^{3}}\left[
{\displaystyle\int\limits_{-\infty}^{\left\vert p\right\vert -\delta}}
+%
{\displaystyle\int\limits_{\left\vert p\right\vert +\delta}^{\infty}}
\right]  \frac{\sin\left(  R\left(  \left\vert p\right\vert +r\right)
\right)  }{\left\vert p\right\vert \left(  \left\vert p\right\vert +r\right)
^{2}\left(  \left\vert p\right\vert -r\right)  }\left\langle f_{+}\left(
p\right)  ,\int\limits_{\left\vert p\right\vert }^{r}\partial_{r_{1}}%
g_{0}\left(  -r_{1}\frac{p}{\left\vert p\right\vert };\left\vert p\right\vert
\right)  dr_{1}\right\rangle F\left(  r\right)  rdrdp=0. \label{t5}%
\end{equation}
For $\delta<1,$ we decompose $i_{1,2}^{1}\left(  R,\left\vert p\right\vert
,\delta\right)  $ in the sum
\[
\left.
\begin{array}
[c]{c}%
i_{1,2}^{1}\left(  R,\left\vert p\right\vert ,\delta\right)  =-\dfrac
{\cos2R\left\vert p\right\vert }{\left(  2\left\vert p\right\vert \right)
^{2}}\left[
{\displaystyle\int\limits_{-1}^{-\delta}}
+%
{\displaystyle\int\limits_{\delta}^{1}}
\right]  \dfrac{\sin\left(  Rr\right)  }{r}F\left(  \left\vert p\right\vert
-r\right)  dr+\sin\left(  2R\left\vert p\right\vert \right)  \left(
i_{1,2}^{1}\left(  R,\left\vert p\right\vert ,\delta\right)  \right)  _{1}\\
+\left[
{\displaystyle\int\limits_{-1}^{-\delta}}
+%
{\displaystyle\int\limits_{\delta}^{1}}
\right]  \sin\left(  R\left(  2\left\vert p\right\vert -r\right)  \right)
\dfrac{1}{r}[\dfrac{1}{\left(  2\left\vert p\right\vert -r\right)  ^{2}%
}-\dfrac{1}{\left(  2\left\vert p\right\vert \right)  ^{2}}]F\left(
\left\vert p\right\vert -r\right)  dr+\left[
{\displaystyle\int\limits_{-\infty}^{-1}}
+%
{\displaystyle\int\limits_{1}^{\infty}}
\right]  \dfrac{\sin R\left(  2\left\vert p\right\vert -r\right)  }{r\left(
2\left\vert p\right\vert -r\right)  ^{2}}F\left(  \left\vert p\right\vert
-r\right)  dr,
\end{array}
\right.
\]
where
\[
\left(  i_{1,2}^{1}\left(  R,\left\vert p\right\vert ,\delta\right)  \right)
_{1}:=\frac{1}{\left(  2\left\vert p\right\vert \right)  ^{2}}\left[
\int\limits_{-1}^{-\delta}+\int\limits_{\delta}^{1}\right]  \frac{\cos\left(
Rr\right)  }{r}F\left(  \left\vert p\right\vert -r\right)  dr.
\]
For $\left\vert p\right\vert <\delta<1$ and any $\varepsilon>0$ we get \
\[
\left.  \left\vert \left(  i_{1,2}^{1}\left(  R,\left\vert p\right\vert
,\delta\right)  \right)  _{1}\right\vert \leq\frac{\ln\delta}{\left(
2\left\vert p\right\vert \right)  ^{2}}\leq\frac{\delta^{\varepsilon}\ln
\delta}{\left(  2\left\vert p\right\vert \right)  ^{2+\varepsilon}}.\right.
\]
If $\delta\leq\left\vert p\right\vert <1,$
\[
\left(  i_{1,2}^{1}\left(  R,\left\vert p\right\vert ,\delta\right)  \right)
_{1}=\frac{\sin\left(  2R\left\vert p\right\vert \right)  }{\left(
2\left\vert p\right\vert \right)  ^{2}}\int\limits_{-1}^{-\left\vert
p\right\vert }\frac{\cos Rr}{r}dr,
\]
and thus,
\[
\left.  \left\vert \left(  i_{1,2}^{1}\left(  R,\left\vert p\right\vert
,\delta\right)  \right)  _{1}\right\vert \leq\frac{\left\vert \ln\left\vert
p\right\vert \right\vert }{\left(  2\left\vert p\right\vert \right)  ^{2}%
}.\right.
\]
Finally, for $\left\vert p\right\vert \geq1,$
\[
\left(  i_{1,2}^{1}\left(  R,\left\vert p\right\vert ,\delta\right)  \right)
_{1}=0.
\]
From these relations, together with the estimates%
\[
\frac{1}{\left(  2\left\vert p\right\vert \right)  ^{2}}\left\vert \left[
\int\limits_{-1}^{-\delta}+\int\limits_{\delta}^{1}\right]  \frac{\sin\left(
Rr\right)  }{r}F\left(  \left\vert p\right\vert -r\right)  dr\right\vert
=\frac{1}{\left(  2\left\vert p\right\vert \right)  ^{2}}\left\vert \left[
\int\limits_{-R}^{-R\delta}+\int\limits_{R\delta}^{R}\right]  \frac{\sin r}%
{r}F\left(  \left\vert p\right\vert -\frac{r}{R}\right)  dr\right\vert \leq
C\frac{1}{\left\vert p\right\vert ^{2}},
\]
uniformly for $R$ and $\delta$ (since $\int_{-\infty}^{\infty}\frac{\sin r}%
{r}dr=2\int_{0}^{\infty}\frac{\sin r}{r}dr=\pi$ implies that $\int_{a}%
^{b}\frac{\sin r}{r}dr\leq C$, for all $a$ and $b$),%
\[
\left.
\begin{array}
[c]{c}%
{\displaystyle\int\limits_{-1}^{1}}
\left\vert \dfrac{1}{r}[\dfrac{1}{\left(  2\left\vert p\right\vert -r\right)
^{2}}-\dfrac{1}{\left(  2\left\vert p\right\vert \right)  ^{2}}]\right\vert
F\left(  \left\vert p\right\vert -r\right)  dr\leq%
{\displaystyle\int\limits_{-1}^{\left\vert p\right\vert }}
\left\vert \dfrac{1}{r}[\dfrac{4r\left\vert p\right\vert -\allowbreak r^{2}%
}{\left(  2\left\vert p\right\vert -r\right)  ^{2}\left(  2\left\vert
p\right\vert \right)  ^{2}}]\right\vert dr\\
\leq%
{\displaystyle\int\limits_{-1}^{\left\vert p\right\vert }}
[\dfrac{1}{\left(  2\left\vert p\right\vert -r\right)  ^{2}\left(  2\left\vert
p\right\vert \right)  }]+[\dfrac{1}{\left(  2\left\vert p\right\vert
-r\right)  \left(  2\left\vert p\right\vert \right)  ^{2}}]dr\leq
C\dfrac{\left(  1+\ln\left\vert p\right\vert \right)  }{\left\vert
p\right\vert ^{2}}%
\end{array}
\right.
\]
and%
\[
\left.  \left[  \int\limits_{-\infty}^{-1}+\int\limits_{1}^{\infty}\right]
\frac{1}{\left\vert r\right\vert \left(  2\left\vert p\right\vert -r\right)
^{2}}F\left(  \left\vert p\right\vert -r\right)  dr\leq C\left(  1+\frac
{1}{\left\vert p\right\vert ^{2}}\right)  ,\right.
\]
we obtain%
\[
\left.
\begin{array}
[c]{c}%
{\displaystyle\int\limits_{\mathbb{R}^{3}}}
\left\vert \left\langle f_{+}\left(  p\right)  ,g_{0}\left(  -p;\left\vert
p\right\vert \right)  \right\rangle i_{1,2}^{1}\left(  R,\left\vert
p\right\vert ,\delta\right)  \right\vert dp\leq C\left\Vert f_{+}\right\Vert
_{L^{2}}\left\Vert g_{+}\right\Vert _{L_{1}^{2}}\\
+C\left(  \left\Vert f_{+}\right\Vert _{L^{\infty}\left(  \left\vert
p\right\vert \leq1\right)  }+\left\Vert f_{+}\right\Vert _{L^{2}\left(
\mathbb{R}^{3}\right)  }\right)  \left(  \left\Vert g_{+}\right\Vert
_{L^{\infty}\left(  \left\vert p\right\vert \leq1\right)  }+\left\Vert
g_{+}\right\Vert _{L^{2}\left(  \mathbb{R}^{3}\right)  }\right)  .
\end{array}
\right.
\]
Arguing as in (\ref{t2}) we get%
\begin{equation}
8\pi^{2}i\lim\limits_{R\rightarrow\infty}\lim\limits_{\delta\rightarrow0}%
\int\limits_{\mathbb{R}^{3}}\left\langle f_{+}\left(  p\right)  ,g_{0}\left(
-p;\left\vert p\right\vert \right)  \right\rangle i_{1,2}^{1}\left(
R,\left\vert p\right\vert ,\delta\right)  dp=0. \label{t7}%
\end{equation}
Moreover, taking the limit in (\ref{t6}), as $R\rightarrow\infty,$ and using
(\ref{t7}) and (\ref{t5}) we obtain (\ref{t18}).
\end{proof}

The following result shows that the term $I_{1,2}^{2}\left(  R\right)  $ gives
the non zero part of the asymptotics of $I_{1,2}\left(  R\right)  $ as
$R\rightarrow\infty$

\begin{lemma}
\label{t57}For $f,g\in\mathcal{H}_{2}^{3/2+\varepsilon}$, $\varepsilon>0$ we
have
\begin{equation}
\lim\limits_{R\rightarrow\infty}I_{1,2}^{2}\left(  R\right)  =8\pi^{3}%
i\int\limits_{\mathbb{R}^{3}}\left\langle f_{+}\left(  p\right)  ,\frac
{g_{0}\left(  p;\left\vert p\right\vert \right)  }{4\left\vert p\right\vert
^{2}}+\frac{p\cdot\left(  \triangledown_{q}g_{0}\right)  \left(  p;\left\vert
p\right\vert \right)  }{2\left\vert p\right\vert ^{2}}\right\rangle dp,
\label{t41}%
\end{equation}
where $g_{0}\left(  q;\left\vert p\right\vert \right)  =\left(  \sqrt
{\left\vert p\right\vert ^{2}+m^{2}}+\sqrt{\left\vert q\right\vert ^{2}+m^{2}%
}\right)  g_{+}\left(  q\right)  .$
\end{lemma}

\begin{proof}
Observe that%
\begin{equation}
\left.  I_{1,2}^{2}\left(  R\right)  =I_{1,2}^{2,1}\left(  R\right)
+I_{1,2}^{2,2}\left(  R\right)  ,\right.  \label{t17}%
\end{equation}
where%
\[
I_{1,2}^{2,1}\left(  R\right)  :=-8\pi^{2}i\int\limits_{\mathbb{R}^{3}}\left[
\int\limits_{-\infty}^{-1}+\int\limits_{1}^{\infty}\right]  \frac{\sin\left(
Rr\right)  }{\left\vert p\right\vert r^{2}}\left\langle f_{+}\left(  p\right)
,\frac{g_{0}\left(  \left(  \left\vert p\right\vert -r\right)  \frac
{p}{\left\vert p\right\vert };\left\vert p\right\vert \right)  }{2\left\vert
p\right\vert -r}\right\rangle F\left(  \left\vert p\right\vert -r\right)
\left(  \left\vert p\right\vert -r\right)  drdp,
\]
and%
\[
I_{1,2}^{2,2}\left(  R\right)  :=-8\pi^{2}i\lim\limits_{\delta\rightarrow
0}\int\limits_{\mathbb{R}^{3}}\left[  \int\limits_{-1}^{-\delta}%
+\int\limits_{\delta}^{1}\right]  \frac{\sin\left(  Rr\right)  }{\left\vert
p\right\vert r^{2}}\left\langle f_{+}\left(  p\right)  ,\frac{g_{0}\left(
\left(  \left\vert p\right\vert -r\right)  \frac{p}{\left\vert p\right\vert
};\left\vert p\right\vert \right)  }{2\left\vert p\right\vert -r}\right\rangle
F\left(  \left\vert p\right\vert -r\right)  \left(  \left\vert p\right\vert
-r\right)  drdp.
\]

Since
\[
\left.
\begin{array}
[c]{c}%
{\displaystyle\int\limits_{\mathbb{R}^{3}}}
\left[
{\displaystyle\int\limits_{-\infty}^{-1}}
+%
{\displaystyle\int\limits_{1}^{\infty}}
\right]  \left\vert \dfrac{1}{\left\vert p\right\vert r^{2}}\left\langle
f_{+}\left(  p\right)  ,\dfrac{g_{0}\left(  \left(  \left\vert p\right\vert
-r\right)  \frac{p}{\left\vert p\right\vert };\left\vert p\right\vert \right)
}{2\left\vert p\right\vert -r}\right\rangle F\left(  \left\vert p\right\vert
-r\right)  \left(  \left\vert p\right\vert -r\right)  \right\vert drdp\\
\leq C%
{\displaystyle\int\limits_{0}^{\infty}}
{\displaystyle\int\limits_{0}^{\infty}}
{\displaystyle\int\limits_{\mathbb{S}^{2}}}
\left\vert f_{+}\left(  \left\vert p\right\vert \omega\right)  \left(
1+\dfrac{1}{\left\vert p\right\vert }\right)  \right\vert \left\vert
g_{+}\left(  r\omega\right)  \right\vert d\omega rdr\left\vert p\right\vert
^{2}d\left\vert p\right\vert \\
\leq C\left(
{\displaystyle\int\limits_{0}^{\infty}}
\left(
{\displaystyle\int\limits_{\mathbb{S}^{2}}}
\left\vert \dfrac{f_{+}\left(  \left\vert p\right\vert \omega\right)
}{\left\vert p\right\vert }\right\vert ^{2}d\omega\right)  ^{1/2}\left\vert
p\right\vert ^{2}d\left\vert p\right\vert \right)  \left(
{\displaystyle\int\limits_{0}^{\infty}}
\left(
{\displaystyle\int\limits_{\mathbb{S}^{2}}}
\left\vert g_{+}\left(  r\omega\right)  \right\vert ^{2}r^{2}d\omega\right)
^{1/2}dr\right) \\
\leq C\left(  \left\Vert f_{+}\right\Vert _{L^{\infty}\left(  \left\vert
p\right\vert \leq1\right)  }+\left\Vert f_{+}\right\Vert _{L_{3/2+\varepsilon
}^{2}}\right)  \left(  \left\Vert g_{+}\right\Vert _{L^{\infty}\left(
\left\vert p\right\vert \leq1\right)  }+\left\Vert g_{+}\right\Vert
_{L_{3/2+\varepsilon}^{2}}\right)  ,
\end{array}
\right.
\]
arguing as in (\ref{t2}) we obtain
\begin{equation}
\lim_{R\rightarrow\infty}I_{1,2}^{2,1}\left(  R\right)  =0. \label{t19}%
\end{equation}

As for all $\delta\leq1$%
\[
\left.
\begin{array}
[c]{c}%
\left[
{\displaystyle\int\limits_{-1}^{-\delta}}
+%
{\displaystyle\int\limits_{\delta}^{1}}
\right]
{\displaystyle\int\limits_{\mathbb{R}^{3}}}
\dfrac{1}{\left\vert p\right\vert ^{2}r^{2}}\left\vert \left\langle
f_{+}\left(  p\right)  ,g_{0}\left(  \left(  \left\vert p\right\vert
-r\right)  \frac{p}{\left\vert p\right\vert };\left\vert p\right\vert \right)
\right\rangle F\left(  \left\vert p\right\vert -r\right)  \dfrac{\left\vert
p\right\vert -r}{2\left\vert p\right\vert -r}\right\vert dpdr\\
\leq\dfrac{C}{\delta^{2}}\left(  \left\Vert f_{+}\right\Vert _{L^{\infty
}\left(  \left\vert p\right\vert \leq1\right)  }+\left\Vert f_{+}\right\Vert
_{L_{1/2+\varepsilon}^{2}}\right)  \left\Vert g_{+}\right\Vert _{L^{\infty}},
\end{array}
\right.
\]
it follows from the Fubini's theorem that%
\begin{equation}
\left.
\begin{array}
[c]{c}%
I_{1,2}^{2,2}\left(  R\right)  =-8\pi^{2}i\lim\limits_{\delta\rightarrow
0}\left[
{\displaystyle\int\limits_{-1}^{-\delta}}
+%
{\displaystyle\int\limits_{\delta}^{1}}
\right]  \sin\left(  Rr\right)
{\displaystyle\int\limits_{\mathbb{R}^{3}}}
\dfrac{1}{\left\vert p\right\vert r^{2}}\left\langle f_{+}\left(  p\right)
,\dfrac{g_{0}\left(  \left(  \left\vert p\right\vert -r\right)  \frac
{p}{\left\vert p\right\vert };\left\vert p\right\vert \right)  }{2\left\vert
p\right\vert -r}\right\rangle F\left(  \left\vert p\right\vert -r\right)
\left(  \left\vert p\right\vert -r\right)  dpdr\\
=-8\pi^{2}i\lim\limits_{\delta\rightarrow0}%
{\displaystyle\int\limits_{\delta}^{1}}
\dfrac{\sin\left(  Rr\right)  }{r^{2}}%
{\displaystyle\int\limits_{\left\vert p\right\vert \leq\frac{r}{3}}}
\left(  \left\langle f_{+}\left(  p\right)  ,g_{0}\left(  \left(  \left\vert
p\right\vert +r\right)  \frac{p}{\left\vert p\right\vert };\left\vert
p\right\vert \right)  \right\rangle \dfrac{\left\vert p\right\vert
+r}{2\left\vert p\right\vert +r}\dfrac{dp}{\left\vert p\right\vert }\right)
dr\\
-8\pi^{2}i\lim\limits_{\delta\rightarrow0}\left[
{\displaystyle\int\limits_{-1}^{-\delta}}
+%
{\displaystyle\int\limits_{\delta}^{1}}
\right]  \dfrac{\sin\left(  Rr\right)  }{r^{2}}%
{\displaystyle\int\limits_{\left\vert p\right\vert \geq\frac{\left\vert
r\right\vert }{3}}}
\left(  \left\langle f_{+}\left(  p\right)  ,\dfrac{g_{0}\left(  \left(
\left\vert p\right\vert -r\right)  \frac{p}{\left\vert p\right\vert
};\left\vert p\right\vert \right)  }{2\left\vert p\right\vert -r}\right\rangle
F\left(  \left\vert p\right\vert -r\right)  \left(  \left\vert p\right\vert
-r\right)  \dfrac{dp}{\left\vert p\right\vert }\right)  dr.
\end{array}
\right.  \label{t20}%
\end{equation}
Noting that
\[
\left.  \int\limits_{\left\vert p\right\vert \leq\frac{r}{3}}\left\vert
\left(  \left\langle f_{+}\left(  p\right)  ,g_{0}\left(  \left(  \left\vert
p\right\vert +r\right)  \frac{p}{\left\vert p\right\vert };\left\vert
p\right\vert \right)  \right\rangle \frac{\left\vert p\right\vert
+r}{2\left\vert p\right\vert +r}\right)  \right\vert \frac{dp}{\left\vert
p\right\vert }\leq Cr^{2}\left\Vert f_{+}\right\Vert _{L^{\infty}\left(
\left\vert p\right\vert \leq1\right)  }\left\Vert g_{+}\right\Vert
_{L^{\infty}\left(  \left\vert p\right\vert \leq2\right)  },\right.
\]
we get
\[
\int\limits_{0}^{1}\frac{1}{r^{2}}\left\vert \int\limits_{\left\vert
p\right\vert \leq\frac{r}{3}}\left(  \left\langle f_{+}\left(  p\right)
,g_{0}\left(  \left\vert p\right\vert +r,\frac{p}{\left\vert p\right\vert
}\right)  \right\rangle \frac{\left\vert p\right\vert +r}{2\left\vert
p\right\vert +r}\frac{dp}{\left\vert p\right\vert }\right)  \right\vert dr\leq
C,
\]
and thus, arguing as in (\ref{t2}), we see that
\begin{equation}
\lim\limits_{R\rightarrow\infty}\lim\limits_{\delta\rightarrow0}%
\int\limits_{\delta}^{1}\frac{\sin\left(  Rr\right)  }{r^{2}}\int
\limits_{\left\vert p\right\vert \leq\frac{r}{3}}\left(  \left\langle
f_{+}\left(  p\right)  ,g_{0}\left(  \left(  \left\vert p\right\vert
+r\right)  \frac{p}{\left\vert p\right\vert };\left\vert p\right\vert \right)
\right\rangle \frac{\left\vert p\right\vert +r}{2\left\vert p\right\vert
+r}\frac{dp}{\left\vert p\right\vert }\right)  dr=0. \label{t21}%
\end{equation}
Using that%
\[
\frac{1}{2\left\vert p\right\vert -r}=\frac{1}{2\left\vert p\right\vert
}+\frac{r}{4\left\vert p\right\vert ^{2}}+\int\limits_{0}^{r}\frac{1}{\left(
2\left\vert p\right\vert -t\right)  ^{3}}\left(  r-t\right)  dt,
\]
and%
\[
g_{0}\left(  \left(  \left\vert p\right\vert -r\right)  \frac{p}{\left\vert
p\right\vert };\left\vert p\right\vert \right)  =g_{0}\left(  p;\left\vert
p\right\vert \right)  -\left(  \frac{p\cdot\left(  \triangledown_{q}%
g_{0}\right)  \left(  p;\left\vert p\right\vert \right)  }{\left\vert
p\right\vert }\right)  r+\frac{1}{2}\int\limits_{0}^{r}\partial_{t}^{2}%
g_{0}\left(  \left(  \left\vert p\right\vert -t\right)  \frac{p}{\left\vert
p\right\vert };\left\vert p\right\vert \right)  \left(  r-t\right)  dt,
\]
we get%
\begin{equation}
\left.
\begin{array}
[c]{c}%
-8\pi^{2}i\lim\limits_{\delta\rightarrow0}\left[
{\displaystyle\int\limits_{-1}^{-\delta}}
+%
{\displaystyle\int\limits_{\delta}^{1}}
\right]  \dfrac{\sin\left(  Rr\right)  }{r^{2}}%
{\displaystyle\int\limits_{\left\vert p\right\vert \geq\frac{\left\vert
r\right\vert }{3}}}
\left(  \left\langle f_{+}\left(  p\right)  ,\dfrac{g_{0}\left(  \left(
\left\vert p\right\vert -r\right)  \frac{p}{\left\vert p\right\vert
};\left\vert p\right\vert \right)  }{2\left\vert p\right\vert -r}\right\rangle
F\left(  \left\vert p\right\vert -r\right)  \left(  \left\vert p\right\vert
-r\right)  \dfrac{dp}{\left\vert p\right\vert }\right)  dr\\
=J_{1}^{2,2}\left(  R\right)  +J_{2}^{2,2}\left(  R\right)  +J_{3}%
^{2,2}\left(  R\right)  ,
\end{array}
\right.  \label{t22}%
\end{equation}
where%
\[
J_{1}^{2,2}\left(  R\right)  :=-8\pi^{2}i\lim\limits_{\delta\rightarrow
0}\left[  \int\limits_{-1}^{-\delta}+\int\limits_{\delta}^{1}\right]
\frac{\sin\left(  Rr\right)  }{r^{2}}\int\limits_{\left\vert p\right\vert
\geq\frac{\left\vert r\right\vert }{3}}F\left(  \left\vert p\right\vert
-r\right)  \left\langle f_{+}\left(  p\right)  ,\frac{g_{0}\left(
p;\left\vert p\right\vert \right)  }{2\left\vert p\right\vert }\right\rangle
dpdr,
\]%
\[
\left.  J_{2}^{2,2}\left(  R\right)  :=8\pi^{2}i\lim\limits_{\delta
\rightarrow0}\left[  \int\limits_{-1}^{-\delta}+\int\limits_{\delta}%
^{1}\right]  \frac{\sin\left(  Rr\right)  }{r}\int\limits_{\left\vert
p\right\vert \geq\frac{\left\vert r\right\vert }{3}}F\left(  \left\vert
p\right\vert -r\right)  \left\langle f_{+}\left(  p\right)  ,\frac
{g_{0}\left(  p;\left\vert p\right\vert \right)  }{4\left\vert p\right\vert
^{2}}+\frac{p\cdot\left(  \triangledown_{q}g_{0}\right)  \left(  p;\left\vert
p\right\vert \right)  }{2\left\vert p\right\vert ^{2}}\right\rangle
dpdr,\right.
\]
and%
\begin{equation}
\left.
\begin{array}
[c]{c}%
J_{3}^{2,2}\left(  R\right)  :=8\pi^{2}i%
{\displaystyle\int\limits_{-1}^{1}}
\sin\left(  Rr\right)
{\displaystyle\int\limits_{\left\vert p\right\vert \geq\frac{\left\vert
r\right\vert }{3}}}
\dfrac{1}{\left\vert p\right\vert }F\left(  \left\vert p\right\vert -r\right)
\left\langle f_{+}\left(  p\right)  ,\dfrac{g_{0}\left(  p;\left\vert
p\right\vert \right)  }{4\left\vert p\right\vert ^{2}}-\dfrac{p\cdot\left(
\triangledown_{q}g_{0}\right)  \left(  p;\left\vert p\right\vert \right)
}{2\left\vert p\right\vert ^{2}}\right\rangle dpdr\\
-8\pi^{2}i%
{\displaystyle\int\limits_{-1}^{1}}
\sin\left(  Rr\right)
{\displaystyle\int\limits_{\left\vert p\right\vert \geq\frac{\left\vert
r\right\vert }{3}}}
\dfrac{1}{\left\vert p\right\vert r^{2}}\left\langle f_{+}\left(  p\right)
,w\left(  \left\vert p\right\vert ,\frac{p}{\left\vert p\right\vert
},r\right)  \right\rangle F\left(  \left\vert p\right\vert -r\right)  \left(
\left\vert p\right\vert -r\right)  dpdr,
\end{array}
\right.  \label{t14}%
\end{equation}
with%
\[
\left.
\begin{array}
[c]{c}%
w\left(  \left\vert p\right\vert ,\frac{p}{\left\vert p\right\vert },r\right)
:=\left(
{\displaystyle\int\limits_{0}^{r}}
\dfrac{1}{\left(  2\left\vert p\right\vert -t\right)  ^{3}}\left(  r-t\right)
dt\right)  g_{0}\left(  p;\left\vert p\right\vert \right) \\
-\left(  \frac{r}{4\left\vert p\right\vert ^{2}}+%
{\displaystyle\int\limits_{0}^{r}}
\dfrac{1}{\left(  2\left\vert p\right\vert -t\right)  ^{3}}\left(  r-t\right)
dt\right)  \dfrac{p\cdot\left(  \triangledown_{q}g_{0}\right)  \left(
p;\left\vert p\right\vert \right)  }{\left\vert p\right\vert }r+\dfrac
{1}{2\left(  2\left\vert p\right\vert -r\right)  }%
{\displaystyle\int\limits_{0}^{r}}
\partial_{t}^{2}g_{0}\left(  \left(  \left\vert p\right\vert -t\right)
\frac{p}{\left\vert p\right\vert };\left\vert p\right\vert \right)  \left(
r-t\right)  dt.
\end{array}
\right.
\]

Note that%
\[
\left.
\begin{array}
[c]{c}%
J_{1}^{2,2}\left(  R\right)  =-8\pi^{2}i\lim\limits_{\delta\rightarrow
0}\left[
{\displaystyle\int\limits_{-1}^{-\delta}}
+%
{\displaystyle\int\limits_{\delta}^{1}}
\right]  \dfrac{\sin\left(  Rr\right)  }{r^{2}}%
{\displaystyle\int\limits_{\left\vert p\right\vert \geq\left\vert r\right\vert
}}
\left\langle f_{+}\left(  p\right)  ,\dfrac{g_{0}\left(  p;\left\vert
p\right\vert \right)  }{2\left\vert p\right\vert }\right\rangle dpdr\\
-8\pi^{2}i\lim\limits_{\delta\rightarrow0}\left[
{\displaystyle\int\limits_{-1}^{-\delta}}
+%
{\displaystyle\int\limits_{\delta}^{1}}
\right]  \dfrac{\sin\left(  Rr\right)  }{r^{2}}%
{\displaystyle\int\limits_{\frac{\left\vert r\right\vert }{3}\leq\left\vert
p\right\vert \leq\left\vert r\right\vert }}
F\left(  \left\vert p\right\vert -r\right)  \left\langle f_{+}\left(
p\right)  ,\dfrac{g_{0}\left(  p;\left\vert p\right\vert \right)
}{2\left\vert p\right\vert }\right\rangle dpdr\\
=8\pi^{2}i%
{\displaystyle\int\limits_{0}^{1}}
\sin\left(  Rr\right)
{\displaystyle\int\limits_{\frac{r}{3}\leq\left\vert p\right\vert \leq r}}
\dfrac{1}{r^{2}}\left\langle f_{+}\left(  p\right)  ,\dfrac{g_{0}\left(
p;\left\vert p\right\vert \right)  }{2\left\vert p\right\vert }\right\rangle
dpdr.
\end{array}
\right.
\]
As for any $0<\varepsilon<1$%
\[
\left.
\begin{array}
[c]{c}%
{\displaystyle\int\limits_{0}^{1}}
{\displaystyle\int\limits_{\frac{r}{3}\leq\left\vert p\right\vert \leq r}}
\dfrac{1}{r^{2}}\left\vert \left\langle f_{+}\left(  p\right)  ,\dfrac
{g_{0}\left(  p;\left\vert p\right\vert \right)  }{2\left\vert p\right\vert
}\right\rangle \right\vert dpdr\leq%
{\displaystyle\int\limits_{0}^{1}}
\dfrac{dr}{r^{1-\varepsilon}}%
{\displaystyle\int\limits_{\left\vert p\right\vert \leq1}}
\dfrac{1}{\left\vert p\right\vert ^{2+\varepsilon}}\left\vert \left\langle
f_{+}\left(  p\right)  ,g_{0}\left(  p;\left\vert p\right\vert \right)
\right\rangle \right\vert dp\\
\leq\dfrac{C}{\varepsilon}\left\Vert f_{+}\right\Vert _{L^{\infty}\left(
\left\vert p\right\vert \leq1\right)  }\left\Vert g_{+}\right\Vert
_{L^{\infty}\left(  \left\vert p\right\vert \leq1\right)  }%
\end{array}
\right.
\]
arguing as in (\ref{t2}) we get
\begin{equation}
\lim_{R\rightarrow\infty}J_{1}^{2,2}\left(  R\right)  =0. \label{t23}%
\end{equation}

We split $J_{2}^{2,2}\left(  R\right)  $ as%
\begin{equation}
\left.
\begin{array}
[c]{c}%
J_{2}^{2,2}\left(  R\right)  =8\pi^{2}i%
{\displaystyle\int\limits_{-1}^{1}}
\sin\left(  Rr\right)
{\displaystyle\int\limits_{\mathbb{R}^{3}}}
\dfrac{1}{r}F\left(  \left\vert p\right\vert -r\right)  \left\langle
f_{+}\left(  p\right)  ,\dfrac{g_{0}\left(  p;\left\vert p\right\vert \right)
}{4\left\vert p\right\vert ^{2}}+\dfrac{p\cdot\left(  \triangledown_{q}%
g_{0}\right)  \left(  p;\left\vert p\right\vert \right)  }{2\left\vert
p\right\vert ^{2}}\right\rangle dpdr\\
-8\pi^{2}i%
{\displaystyle\int\limits_{0}^{1}}
\sin\left(  Rr\right)
{\displaystyle\int\limits_{\left\vert p\right\vert \leq\frac{r}{3}}}
\dfrac{1}{r}\left\langle f_{+}\left(  p\right)  ,\dfrac{g_{0}\left(
p;\left\vert p\right\vert \right)  }{4\left\vert p\right\vert ^{2}}%
+\dfrac{p\cdot\left(  \triangledown_{q}g_{0}\right)  \left(  p;\left\vert
p\right\vert \right)  }{2\left\vert p\right\vert ^{2}}\right\rangle dpdr.
\end{array}
\right.  \label{t11}%
\end{equation}
Noting that for any $0<\varepsilon<\frac{1}{2}$
\[
\left.
\begin{array}
[c]{c}%
{\displaystyle\int\limits_{0}^{1}}
\dfrac{1}{r}%
{\displaystyle\int\limits_{\left\vert p\right\vert \leq\frac{r}{3}}}
\left\vert \left\langle f_{+}\left(  p\right)  ,\dfrac{g_{0}\left(
p;\left\vert p\right\vert \right)  }{4\left\vert p\right\vert ^{2}}%
+\dfrac{p\cdot\left(  \triangledown_{q}g_{0}\right)  \left(  p;\left\vert
p\right\vert \right)  }{2\left\vert p\right\vert ^{2}}\right\rangle
\right\vert dpdr\\
\leq%
{\displaystyle\int\limits_{0}^{1}}
\dfrac{dr}{r^{1-\varepsilon}}%
{\displaystyle\int\limits_{\left\vert p\right\vert \leq1}}
\dfrac{\left\vert f_{+}\left(  p\right)  \right\vert }{\left\vert p\right\vert
^{1+\varepsilon}}\left(  \left\vert \dfrac{g_{0}\left(  p;\left\vert
p\right\vert \right)  }{4\left\vert p\right\vert }\right\vert +\left\vert
\dfrac{p\cdot\left(  \triangledown_{q}g_{0}\right)  \left(  p;\left\vert
p\right\vert \right)  }{\left\vert p\right\vert }\right\vert \right)  dp\\
\leq\dfrac{C}{\varepsilon}\left\Vert f_{+}\right\Vert _{L^{\infty}\left(
\left\vert p\right\vert \leq1\right)  }\left(  \left\Vert g_{+}\right\Vert
_{L^{\infty}\left(  \left\vert p\right\vert \leq1\right)  }+\left\Vert
g_{+}\right\Vert _{\mathcal{H}^{1}\left(  \left\vert p\right\vert
\leq1\right)  }\right)  ,
\end{array}
\right.
\]
and arguing as in (\ref{t2}) we obtain%
\begin{equation}
8\pi^{2}i\lim_{R\rightarrow\infty}\int\limits_{0}^{1}\sin\left(  Rr\right)
\int\limits_{\left\vert p\right\vert \leq\frac{r}{3}}\frac{1}{r}\left\langle
f_{+}\left(  p\right)  ,\frac{g_{0}\left(  p;\left\vert p\right\vert \right)
}{4\left\vert p\right\vert ^{2}}+\frac{p\cdot\left(  \triangledown_{q}%
g_{0}\right)  \left(  p;\left\vert p\right\vert \right)  }{2\left\vert
p\right\vert ^{2}}\right\rangle dpdr=0. \label{t12}%
\end{equation}
Observe now that
\[
\left.
\begin{array}
[c]{c}%
8\pi^{2}i%
{\displaystyle\int\limits_{-1}^{1}}
{\displaystyle\int\limits_{\mathbb{R}^{3}}}
\dfrac{\sin\left(  Rr\right)  }{r}F\left(  \left\vert p\right\vert -r\right)
\left\langle f_{+}\left(  p\right)  ,\dfrac{g_{0}\left(  p;\left\vert
p\right\vert \right)  }{4\left\vert p\right\vert ^{2}}+\dfrac{p\cdot\left(
\triangledown_{q}g_{0}\right)  \left(  p;\left\vert p\right\vert \right)
}{2\left\vert p\right\vert ^{2}}\right\rangle dpdr\\
=8\pi^{2}i%
{\displaystyle\int\limits_{\mathbb{R}^{3}}}
\left(
{\displaystyle\int\limits_{-1}^{1}}
\dfrac{\sin\left(  Rr\right)  }{r}F\left(  \left\vert p\right\vert -r\right)
dr\right)  \left\langle f_{+}\left(  p\right)  ,\dfrac{g_{0}\left(
p;\left\vert p\right\vert \right)  }{4\left\vert p\right\vert ^{2}}%
+\dfrac{p\cdot\left(  \triangledown_{q}g_{0}\right)  \left(  p;\left\vert
p\right\vert \right)  }{2\left\vert p\right\vert ^{2}}\right\rangle dp\\
=8\pi^{2}i%
{\displaystyle\int\limits_{\left\vert p\right\vert \geq1}}
\left(
{\displaystyle\int\limits_{-R}^{R}}
\dfrac{\sin r}{r}dr\right)  \left\langle f_{+}\left(  p\right)  ,\dfrac
{g_{0}\left(  p;\left\vert p\right\vert \right)  }{4\left\vert p\right\vert
^{2}}+\dfrac{p\cdot\left(  \triangledown_{q}g_{0}\right)  \left(  p;\left\vert
p\right\vert \right)  }{2\left\vert p\right\vert ^{2}}\right\rangle dp\\
+8\pi^{2}i%
{\displaystyle\int\limits_{\left\vert p\right\vert \leq1}}
\left(
{\displaystyle\int\limits_{-R}^{R\left\vert p\right\vert }}
\dfrac{\sin r}{r}dr\right)  \left\langle f_{+}\left(  p\right)  ,\dfrac
{g_{0}\left(  p;\left\vert p\right\vert \right)  }{4\left\vert p\right\vert
^{2}}+\dfrac{p\cdot\left(  \triangledown_{q}g_{0}\right)  \left(  p;\left\vert
p\right\vert \right)  }{2\left\vert p\right\vert ^{2}}\right\rangle dp.
\end{array}
\right.
\]
Since
\[
\left.
\begin{array}
[c]{c}%
{\displaystyle\int\limits_{\left\vert p\right\vert \geq1}}
\left\vert \left(
{\displaystyle\int\limits_{-R}^{R}}
\dfrac{\sin r}{r}dr\right)  \left\langle f_{+}\left(  p\right)  ,\dfrac
{g_{0}\left(  p;\left\vert p\right\vert \right)  }{4\left\vert p\right\vert
^{2}}+\dfrac{p\cdot\left(  \triangledown_{q}g_{0}\right)  \left(  p;\left\vert
p\right\vert \right)  }{2\left\vert p\right\vert ^{2}}\right\rangle
\right\vert dp\\
\leq C%
{\displaystyle\int\limits_{\left\vert p\right\vert \geq1}}
\left\vert \left\langle f_{+}\left(  p\right)  ,\dfrac{g_{0}\left(
p;\left\vert p\right\vert \right)  }{4\left\vert p\right\vert ^{2}}%
+\dfrac{p\cdot\left(  \triangledown_{q}g_{0}\right)  \left(  p;\left\vert
p\right\vert \right)  }{2\left\vert p\right\vert ^{2}}\right\rangle
\right\vert dp\\
\leq C\left\Vert f_{+}\right\Vert _{L^{2}}\left(  \left\Vert g_{+}\right\Vert
_{L^{2}}+\left\Vert g_{+}\right\Vert _{\mathcal{H}^{1}}\right)  ,
\end{array}
\right.
\]
and
\[
\left.
\begin{array}
[c]{c}%
{\displaystyle\int\limits_{\left\vert p\right\vert \leq1}}
\left\vert \left(
{\displaystyle\int\limits_{-R}^{R\left\vert p\right\vert }}
\dfrac{\sin r}{r}dr\right)  \left\langle f_{+}\left(  p\right)  ,\dfrac
{g_{0}\left(  p;\left\vert p\right\vert \right)  }{4\left\vert p\right\vert
^{2}}+\dfrac{p\cdot\left(  \triangledown_{q}g_{0}\right)  \left(  p;\left\vert
p\right\vert \right)  }{2\left\vert p\right\vert ^{2}}\right\rangle
\right\vert dp\\
\leq C%
{\displaystyle\int\limits_{\left\vert p\right\vert \leq1}}
\left\vert \left\langle f_{+}\left(  p\right)  ,\dfrac{g_{0}\left(
p;\left\vert p\right\vert \right)  }{4\left\vert p\right\vert ^{2}}%
+\dfrac{p\cdot\left(  \triangledown_{q}g_{0}\right)  \left(  p;\left\vert
p\right\vert \right)  }{2\left\vert p\right\vert ^{2}}\right\rangle
\right\vert dp\\
\leq C\left\Vert f_{+}\right\Vert _{L^{\infty}\left(  \left\vert p\right\vert
\leq1\right)  }\left(  \left\Vert g_{+}\right\Vert _{L^{\infty}\left(
\left\vert p\right\vert \leq1\right)  }+\left\Vert g_{+}\right\Vert
_{\mathcal{H}^{1}\left(  \left\vert p\right\vert \leq1\right)  }\right)  ,
\end{array}
\right.
\]
uniformly on $R,$ it follows from the dominated convergence theorem and the
equality $\int_{-\infty}^{\infty}\frac{\sin r}{r}dr=\pi$ that%
\begin{equation}
\left.
\begin{array}
[c]{c}%
8\pi^{2}i\lim\limits_{R\rightarrow\infty}%
{\displaystyle\int\limits_{-1}^{1}}
{\displaystyle\int\limits_{\mathbb{R}^{3}}}
\dfrac{\sin\left(  Rr\right)  }{r}F\left(  \left\vert p\right\vert -r\right)
\left\langle f_{+}\left(  p\right)  ,\dfrac{g_{0}\left(  p;\left\vert
p\right\vert \right)  }{4\left\vert p\right\vert ^{2}}+\dfrac{p\cdot\left(
\triangledown_{q}g_{0}\right)  \left(  p;\left\vert p\right\vert \right)
}{2\left\vert p\right\vert ^{2}}\right\rangle dpdr\\
=8\pi^{3}i%
{\displaystyle\int\limits_{\mathbb{R}^{3}}}
\left\langle f_{+}\left(  p\right)  ,\dfrac{g_{0}\left(  p;\left\vert
p\right\vert \right)  }{4\left\vert p\right\vert ^{2}}+\dfrac{p\cdot\left(
\triangledown_{q}g_{0}\right)  \left(  p;\left\vert p\right\vert \right)
}{2\left\vert p\right\vert ^{2}}\right\rangle dp.
\end{array}
\right.  \label{t13}%
\end{equation}
Therefore, taking the limit, as $R\rightarrow\infty,$ in (\ref{t11}), and
using (\ref{t12}), (\ref{t13}), we obtain%
\begin{equation}
\lim\limits_{R\rightarrow\infty}J_{2}^{2,2}\left(  R\right)  =8\pi^{3}%
i\int\limits_{\mathbb{R}^{3}}\left\langle f_{+}\left(  p\right)  ,\frac
{g_{0}\left(  p;\left\vert p\right\vert \right)  }{4\left\vert p\right\vert
^{2}}+\frac{p\cdot\left(  \triangledown_{q}g_{0}\right)  \left(  p;\left\vert
p\right\vert \right)  }{2\left\vert p\right\vert ^{2}}\right\rangle dp.
\label{t24}%
\end{equation}

Let us consider now $J_{3}^{2,2}\left(  R\right)  .$ Observe that for all
$0<\varepsilon<\frac{1}{2}$%
\[
\left.
\begin{array}
[c]{c}%
{\displaystyle\int\limits_{-1}^{1}}
{\displaystyle\int\limits_{\left\vert p\right\vert \geq\frac{\left\vert
r\right\vert }{3}}}
\dfrac{1}{\left\vert p\right\vert }\left\vert F\left(  \left\vert p\right\vert
-r\right)  \left\langle f_{+}\left(  p\right)  ,\dfrac{g_{0}\left(
p;\left\vert p\right\vert \right)  }{4\left\vert p\right\vert ^{2}}%
-\dfrac{p\cdot\left(  \triangledown_{q}g_{0}\right)  \left(  p;\left\vert
p\right\vert \right)  }{2\left\vert p\right\vert ^{2}}\right\rangle
\right\vert dpdr\\
\leq C%
{\displaystyle\int\limits_{0}^{1}}
\dfrac{dr}{r^{1-\varepsilon}}%
{\displaystyle\int\limits_{\mathbb{R}^{3}}}
\left\vert \left\langle f_{+}\left(  p\right)  ,\dfrac{g_{0}\left(
p;\left\vert p\right\vert \right)  }{4\left\vert p\right\vert ^{2+\varepsilon
}}-\dfrac{p\cdot\left(  \triangledown_{q}g_{0}\right)  \left(  p;\left\vert
p\right\vert \right)  }{2\left\vert p\right\vert ^{2+\varepsilon}%
}\right\rangle \right\vert dp\\
\leq C\left\Vert f_{+}\right\Vert _{L^{\infty}\left(  \left\vert p\right\vert
\leq1\right)  }\left(  \left\Vert g_{+}\right\Vert _{L^{\infty}\left(
\left\vert p\right\vert \leq1\right)  }+\left\Vert g_{+}\right\Vert
_{\mathcal{H}^{1}\left(  \left\vert p\right\vert \leq1\right)  }\right)
+C\left\Vert f_{+}\right\Vert _{L^{2}}\left(  \left\Vert g_{+}\right\Vert
_{L^{2}}+\left\Vert g_{+}\right\Vert _{\mathcal{H}^{1}}\right)  .
\end{array}
\right.
\]
Then, arguing as in (\ref{t2}) we obtain%
\begin{equation}
8\pi^{2}i\lim_{R\rightarrow\infty}\int\limits_{-1}^{1}\sin\left(  Rr\right)
\int\limits_{\left\vert p\right\vert \geq\frac{\left\vert r\right\vert }{3}%
}\frac{1}{\left\vert p\right\vert }F\left(  \left\vert p\right\vert -r\right)
\left\langle f_{+}\left(  p\right)  ,\frac{g_{0}\left(  p;\left\vert
p\right\vert \right)  }{4\left\vert p\right\vert ^{2}}-\frac{p\cdot\left(
\triangledown_{q}g_{0}\right)  \left(  p;\left\vert p\right\vert \right)
}{2\left\vert p\right\vert ^{2}}\right\rangle dpdr=0. \label{t15}%
\end{equation}

Note now that for $0<\varepsilon<1$
\[
\left.
\begin{array}
[c]{c}%
\left\vert w\left(  \left\vert p\right\vert ,\frac{p}{\left\vert p\right\vert
},r\right)  F\left(  \left\vert p\right\vert -r\right)  \right\vert \leq
C\dfrac{\left\vert r\right\vert ^{\frac{3}{2}}}{\left\vert p\right\vert
^{5/2}}\left\vert g_{0}\left(  p;\left\vert p\right\vert \right)  \right\vert
+C\dfrac{1}{\left\vert p\right\vert ^{2}}\left\vert \dfrac{p\cdot\left(
\triangledown_{q}g_{0}\right)  \left(  p;\left\vert p\right\vert \right)
}{\left\vert p\right\vert }\right\vert \\
+C\dfrac{\left\vert r\right\vert ^{\frac{3}{2}-\varepsilon}}{\left\vert
p\right\vert ^{\varepsilon}\left(  \left\vert p\right\vert -r\right)
^{2-2\varepsilon}}F\left(  \left\vert p\right\vert -r\right)  \left(
{\displaystyle\int\limits_{0}^{\left\vert p\right\vert }}
\left\vert \partial_{t}^{2}g_{0}\left(  t\frac{p}{\left\vert p\right\vert
};\left\vert p\right\vert \right)  \right\vert ^{2}t^{2}dt\right)  ^{1/2}.
\end{array}
\right.
\]
Then, for $\varepsilon<\frac{1}{2},$%
\[
\left.
\begin{array}
[c]{c}%
{\displaystyle\int\limits_{-1}^{1}}
{\displaystyle\int\limits_{\left\vert p\right\vert \geq\frac{\left\vert
r\right\vert }{3}}}
\dfrac{1}{\left\vert p\right\vert r^{2}}\left\vert \left\langle f_{+}\left(
p\right)  ,w\left(  \left\vert p\right\vert ,\frac{p}{\left\vert p\right\vert
},r\right)  \right\rangle F\left(  \left\vert p\right\vert -r\right)  \left(
\left\vert p\right\vert -r\right)  \right\vert dpdr\\
\leq C%
{\displaystyle\int\limits_{\mathbb{R}^{3}}}
\dfrac{1}{\left\vert p\right\vert ^{5/2}}\left\vert f_{+}\left(  p\right)
\right\vert \left\vert g_{0}\left(  p;\left\vert p\right\vert \right)
\right\vert dp+C%
{\displaystyle\int\limits_{-1}^{1}}
\dfrac{dr}{\left\vert r\right\vert ^{1-\varepsilon}}%
{\displaystyle\int\limits_{\mathbb{R}^{3}}}
\dfrac{\left\vert f_{+}\left(  p\right)  \right\vert }{\left\vert p\right\vert
^{1+\varepsilon}}\left\vert \dfrac{p\cdot\left(  \triangledown_{q}%
g_{0}\right)  \left(  p;\left\vert p\right\vert \right)  }{\left\vert
p\right\vert }\right\vert dp\\
+C%
{\displaystyle\int\limits_{-1}^{1}}
\dfrac{1}{\left\vert r\right\vert ^{\frac{1}{2}+\varepsilon}}%
{\displaystyle\int\limits_{\left\vert p\right\vert \geq\frac{\left\vert
r\right\vert }{3}}}
\dfrac{\left\vert f_{+}\left(  p\right)  \right\vert }{\left\vert p\right\vert
}\dfrac{F\left(  \left\vert p\right\vert -r\right)  }{\left(  \left\vert
p\right\vert -r\right)  ^{1-2\varepsilon}}\left(
{\displaystyle\int\limits_{0}^{\left\vert p\right\vert }}
\left\vert \partial_{t}^{2}g_{0}\left(  t\frac{p}{\left\vert p\right\vert
};\left\vert p\right\vert \right)  \right\vert ^{2}t^{2}dt\right)  ^{1/2}dpdr.
\end{array}
\right.
\]
Thus, using that%
\[
\left.
\begin{array}
[c]{c}%
{\displaystyle\int\limits_{-1}^{1}}
\dfrac{1}{\left\vert r\right\vert ^{\frac{1}{2}+\varepsilon}}%
{\displaystyle\int\limits_{\left\vert p\right\vert \geq\frac{\left\vert
r\right\vert }{3}}}
\dfrac{\left\vert f_{+}\left(  p\right)  \right\vert }{\left\vert p\right\vert
}\dfrac{F\left(  \left\vert p\right\vert -r\right)  }{\left(  \left\vert
p\right\vert -r\right)  ^{1-2\varepsilon}}\left(
{\displaystyle\int\limits_{0}^{\left\vert p\right\vert }}
\left\vert \partial_{t}^{2}g_{0}\left(  t\frac{p}{\left\vert p\right\vert
};\left\vert p\right\vert \right)  \right\vert ^{2}t^{2}dt\right)
^{1/2}dpdr\\
\leq C%
{\displaystyle\int\limits_{\left\vert p\right\vert \leq1}}
\left(  \dfrac{1}{\left\vert p\right\vert ^{2-2\varepsilon}}+\dfrac
{1}{\left\vert p\right\vert ^{\frac{3}{2}+\varepsilon}}\right)  \left\vert
f_{+}\left(  p\right)  \right\vert \left(
{\displaystyle\int\limits_{0}^{1}}
\left\vert \partial_{t}^{2}g_{0}\left(  t\frac{p}{\left\vert p\right\vert
};\left\vert p\right\vert \right)  \right\vert ^{2}t^{2}dt\right)  ^{1/2}dp\\
+C%
{\displaystyle\int\limits_{\left\vert p\right\vert \geq1}}
\dfrac{\left\vert f_{+}\left(  p\right)  \right\vert }{\left\vert p\right\vert
}\left(  \left(
{\displaystyle\int\limits_{0}^{1}}
\left\vert \partial_{t}^{2}g_{0}\left(  t\frac{p}{\left\vert p\right\vert
};\left\vert p\right\vert \right)  \right\vert ^{2}t^{2}dt\right)
^{1/2}+\left(
{\displaystyle\int\limits_{1}^{\infty}}
\left\vert \partial_{t}^{2}g_{0}\left(  t\frac{p}{\left\vert p\right\vert
};\left\vert p\right\vert \right)  \right\vert ^{2}dt\right)  ^{1/2}\right)
dp\\
\leq C\left(  \left\Vert f_{+}\right\Vert _{L^{\infty}\left(  \left\vert
p\right\vert \leq1\right)  }+\left\Vert f_{+}\right\Vert _{L_{1/2+\varepsilon
}^{2}}\right)  \left\Vert g_{+}\right\Vert _{\mathcal{H}^{2}},
\end{array}
\right.
\]
we obtain%
\[
\left.
\begin{array}
[c]{c}%
{\displaystyle\int\limits_{-1}^{1}}
{\displaystyle\int\limits_{\left\vert p\right\vert \geq\frac{\left\vert
r\right\vert }{3}}}
\dfrac{1}{\left\vert p\right\vert r^{2}}\left\vert \left\langle f_{+}\left(
p\right)  ,w\left(  \left\vert p\right\vert ,\frac{p}{\left\vert p\right\vert
},r\right)  \right\rangle F\left(  \left\vert p\right\vert -r\right)  \left(
\left\vert p\right\vert -r\right)  \right\vert dpdr\\
\leq C\left(  \left\Vert f_{+}\right\Vert _{L^{\infty}\left(  \left\vert
p\right\vert \leq1\right)  }\left\Vert g_{+}\right\Vert _{L^{\infty}\left(
\left\vert p\right\vert \leq1\right)  }+\left\Vert f_{+}\right\Vert _{L^{2}%
}\left\Vert g_{+}\right\Vert _{L^{2}}\right) \\
+C\left(  \left\Vert f_{+}\right\Vert _{L^{\infty}\left(  \left\vert
p\right\vert \leq1\right)  }+\left\Vert f_{+}\right\Vert _{L^{2}}\right)
\left\Vert g_{+}\right\Vert _{\mathcal{H}^{1}}+C\left(  \left\Vert
f_{+}\right\Vert _{L^{\infty}\left(  \left\vert p\right\vert \leq1\right)
}+\left\Vert f_{+}\right\Vert _{L_{1/2+\varepsilon}^{2}}\right)  \left\Vert
g_{+}\right\Vert _{\mathcal{H}^{2}}\\
\leq C\left(  \left\Vert f_{+}\right\Vert _{L^{\infty}\left(  \left\vert
p\right\vert \leq1\right)  }+\left\Vert f_{+}\right\Vert _{L_{1/2+\varepsilon
}^{2}}\right)  \left(  \left\Vert g_{+}\right\Vert _{L^{\infty}\left(
\left\vert p\right\vert \leq1\right)  }+\left\Vert g_{+}\right\Vert
_{\mathcal{H}^{2}}\right)  .
\end{array}
\right.
\]
Hence, arguing as in (\ref{t2}) we see that%
\begin{equation}
\left.  -8\pi^{2}i\lim_{R\rightarrow\infty}\int\limits_{-1}^{1}\sin\left(
Rr\right)  \int\limits_{\left\vert p\right\vert \geq\frac{\left\vert
r\right\vert }{3}}\frac{1}{\left\vert p\right\vert r^{2}}\left\langle
f_{+}\left(  p\right)  ,w\left(  \left\vert p\right\vert ,\frac{p}{\left\vert
p\right\vert },r\right)  \right\rangle F\left(  \left\vert p\right\vert
-r\right)  \left(  \left\vert p\right\vert -r\right)  dpdr=0.\right.
\label{t16}%
\end{equation}
Taking the limit, as $R\rightarrow\infty,$ in (\ref{t14}), and taking in
account (\ref{t15}) and (\ref{t16}) we arrive to
\begin{equation}
\lim_{R\rightarrow\infty}J_{3}^{2,2}\left(  R\right)  =0. \label{t25}%
\end{equation}
Moreover, passing to the limit, as $R\rightarrow\infty,$ in (\ref{t22}), and
using (\ref{t23}), (\ref{t24}) and (\ref{t25}) we get%
\[
\left.
\begin{array}
[c]{c}%
-8\pi^{2}i\lim\limits_{R\rightarrow\infty}\lim\limits_{\delta\rightarrow
0}\left[
{\displaystyle\int\limits_{-1}^{-\delta}}
+%
{\displaystyle\int\limits_{\delta}^{1}}
\right]  \dfrac{\sin\left(  Rr\right)  }{r^{2}}%
{\displaystyle\int\limits_{\left\vert p\right\vert \geq\frac{\left\vert
r\right\vert }{3}}}
\left(  \left\langle f_{+}\left(  p\right)  ,\dfrac{g_{0}\left(  \left(
\left\vert p\right\vert -r\right)  \frac{p}{\left\vert p\right\vert
};\left\vert p\right\vert \right)  }{2\left\vert p\right\vert -r}\right\rangle
F\left(  \left\vert p\right\vert -r\right)  \left(  \left\vert p\right\vert
-r\right)  \dfrac{dp}{\left\vert p\right\vert }\right)  dr\\
=8\pi^{3}i%
{\displaystyle\int\limits_{\mathbb{R}^{3}}}
\left\langle f_{+}\left(  p\right)  ,\dfrac{g_{0}\left(  p;\left\vert
p\right\vert \right)  }{4\left\vert p\right\vert ^{2}}+\dfrac{p\cdot\left(
\triangledown_{q}g_{0}\right)  \left(  p;\left\vert p\right\vert \right)
}{2\left\vert p\right\vert ^{2}}\right\rangle dp.
\end{array}
\right.
\]
Using the last relation together with (\ref{t21}) in (\ref{t20}) we obtain
\begin{equation}
I_{1,2}^{2,2}\left(  R\right)  =8\pi^{3}i\int\limits_{\mathbb{R}^{3}%
}\left\langle f_{+}\left(  p\right)  ,\frac{g_{0}\left(  p;\left\vert
p\right\vert \right)  }{4\left\vert p\right\vert ^{2}}+\frac{p\cdot\left(
\triangledown_{q}g_{0}\right)  \left(  p;\left\vert p\right\vert \right)
}{2\left\vert p\right\vert ^{2}}\right\rangle dp+o\left(  1\right)  ,
\label{t26}%
\end{equation}
as $R\rightarrow\infty.$ Moreover, using equalities (\ref{t19}) and
(\ref{t26}) in (\ref{t17}) we arrive to (\ref{t41}).
\end{proof}

Now we show that $I_{1,2}^{3}\left(  R\right)  $ is $o\left(  1\right)  $ as
$R\rightarrow\infty.$ We have

\begin{lemma}
\label{t49}Suppose that $f,g\in\mathcal{H}_{2}^{3/2+\varepsilon}$,
$\varepsilon>0$. Then,
\[
\lim_{R\rightarrow\infty}I_{1,2}^{3}\left(  R\right)  =0.
\]

\end{lemma}

\begin{proof}
Recall that%
\[
\left.
\begin{array}
[c]{c}%
I_{1,2}^{3}\left(  R\right)  =-4\pi i\lim\limits_{\delta\rightarrow0}%
{\displaystyle\int\limits_{\mathbb{R}^{3}}}
\left[
{\displaystyle\int\limits_{-\infty}^{\left\vert p\right\vert -\delta}}
+%
{\displaystyle\int\limits_{\left\vert p\right\vert +\delta}^{\infty}}
\right] \\
\times\frac{\left(  \sqrt{\left\vert p\right\vert ^{2}+m^{2}}+\sqrt
{r^{2}+m^{2}}\right)  }{r\left\vert p\right\vert \left(  \left\vert
p\right\vert -r\right)  \left(  \left\vert p\right\vert +r\right)
}\left\langle f_{+}\left(  p\right)  ,%
{\displaystyle\int\limits_{0}^{2\pi}}
{\displaystyle\int\limits_{0}^{\pi}}
\frac{\sin\left(  R\sqrt{\left\vert p\right\vert ^{2}-2r\left\vert
p\right\vert \cos\theta+r^{2}}\right)  }{\sqrt{\left\vert p\right\vert
^{2}-2r\left\vert p\right\vert \cos\theta+r^{2}}}\partial_{\theta}g_{+}\left(
r\omega(\theta,\varphi)\right)  d\theta d\varphi\right\rangle F\left(
r\right)  r^{2}drdp.
\end{array}
\right.
\]
Noting that for all $0<\delta\leq1$%
\[
\left[  \int\limits_{\left\vert p\right\vert -1}^{\left\vert p\right\vert
-\delta}+\int\limits_{\left\vert p\right\vert +\delta}^{\left\vert
p\right\vert +1}\right]  \frac{\sqrt{\left\vert p\right\vert ^{2}+m^{2}}%
}{\left(  \left\vert p\right\vert -r\right)  \left\vert p\right\vert
}\left\langle f_{+}\left(  p\right)  ,\int\limits_{0}^{2\pi}\int
\limits_{0}^{\pi}\frac{\sin\left(  R\left\vert p\right\vert \sqrt
{2-2\cos\theta}\right)  }{\left\vert p\right\vert \sqrt{2-2\cos\theta}%
}\partial_{\theta}g_{+}\left(  \left\vert p\right\vert \omega(\theta
,\varphi)\right)  d\theta d\varphi\right\rangle dr=0,
\]
we decompose $I_{1,2}^{3}\left(  R\right)  $ as follows%
\[
\left.
\begin{array}
[c]{c}%
I_{1,2}^{3}\left(  R\right)  =I_{1,2}^{3,1}\left(  R;0,\pi/4\right)
+I_{1,2}^{3,1}\left(  R;\pi/4,3\pi/4\right)  +I_{1,2}^{3,1}\left(
R;3\pi/4,\pi\right)  +I_{1,2}^{3,2}\left(  R;0,\pi/4\right) \\
+I_{1,2}^{3,2}\left(  R;\pi/4,3\pi/4\right)  +I_{1,2}^{3,2}\left(
R;3\pi/4,\pi\right)  +I_{1,2}^{3,3}\left(  R\right)  +I_{1,2}^{3,4}\left(
R\right)  ,
\end{array}
\right.
\]
where%
\[
\left.
\begin{array}
[c]{c}%
I_{1,2}^{3,1}\left(  R;a,b\right)  :=-4\pi i\lim\limits_{\delta\rightarrow0}%
{\displaystyle\int\limits_{\mathbb{R}^{3}}}
\left[
{\displaystyle\int\limits_{\left\vert p\right\vert -1}^{\left\vert
p\right\vert -\delta}}
+%
{\displaystyle\int\limits_{\left\vert p\right\vert +\delta}^{\left\vert
p\right\vert +1}}
\right] \\
\times\dfrac{1}{\left(  \left\vert p\right\vert -r\right)  }\left(
\frac{\left(  \sqrt{\left\vert p\right\vert ^{2}+m^{2}}+\sqrt{r^{2}+m^{2}%
}\right)  }{\left\vert p\right\vert \left(  \left\vert p\right\vert +r\right)
}F\left(  r\right)  r-\frac{\sqrt{\left\vert p\right\vert ^{2}+m^{2}}%
}{\left\vert p\right\vert }\right)  \left\langle f_{+}\left(  p\right)
,i_{1,2}^{3,1}\left(  \left\vert p\right\vert ,r,\theta,\varphi;R;a,b\right)
d\theta d\varphi\right\rangle drdp,
\end{array}
\right.
\]
with%
\begin{equation}
i_{1,2}^{3,1}\left(  \left\vert p\right\vert ,r;R;a,b\right)  :=%
{\displaystyle\int\limits_{0}^{2\pi}}
{\displaystyle\int\limits_{a}^{b}}
\left(  \frac{\sin\left(  R\sqrt{\left\vert p\right\vert ^{2}-2r\left\vert
p\right\vert \cos\theta+r^{2}}\right)  }{\sqrt{\left\vert p\right\vert
^{2}-2r\left\vert p\right\vert \cos\theta+r^{2}}}\partial_{\theta}g_{+}\left(
r\omega(\theta,\varphi)\right)  \right)  d\theta d\varphi, \label{t125}%
\end{equation}%
\[
\left.  I_{1,2}^{3,2}\left(  R;a,b\right)  :=-4\pi i\lim\limits_{\delta
\rightarrow0}\int\limits_{\mathbb{R}^{3}}\left[  \int\limits_{\left\vert
p\right\vert -1}^{\left\vert p\right\vert -\delta}+\int\limits_{\left\vert
p\right\vert +\delta}^{\left\vert p\right\vert +1}\right]  \left(  \frac
{\sqrt{\left\vert p\right\vert ^{2}+m^{2}}}{\left(  \left\vert p\right\vert
-r\right)  \left\vert p\right\vert }\left\langle f_{+}\left(  p\right)
,i_{1,2}^{3,2}\left(  \left\vert p\right\vert ,r,\theta,\varphi;R;a,b\right)
\right\rangle \right)  drdp,\right.
\]
with%
\[
i_{1,2}^{3,2}\left(  \left\vert p\right\vert ,r;R;a,b\right)  :=%
{\displaystyle\int\limits_{0}^{2\pi}}
{\displaystyle\int\limits_{a}^{b}}
\left(  \frac{\sin\left(  R\sqrt{\left\vert p\right\vert ^{2}-2r\left\vert
p\right\vert \cos\theta+r^{2}}\right)  }{\sqrt{\left\vert p\right\vert
^{2}-2r\left\vert p\right\vert \cos\theta+r^{2}}}-\frac{\sin\left(
R\left\vert p\right\vert \sqrt{2-2\cos\theta}\right)  }{\left\vert
p\right\vert \sqrt{2-2\cos\theta}}\right)  \partial_{\theta}g_{+}\left(
r\omega(\theta,\varphi)\right)  d\theta d\varphi,
\]%
\[
\left.  I_{1,2}^{3,3}\left(  R\right)  :=-4\pi i\lim\limits_{\delta
\rightarrow0}%
{\displaystyle\int\limits_{\mathbb{R}^{3}}}
\left[
{\displaystyle\int\limits_{\left\vert p\right\vert -1}^{\left\vert
p\right\vert -\delta}}
+%
{\displaystyle\int\limits_{\left\vert p\right\vert +\delta}^{\left\vert
p\right\vert +1}}
\right]  \left(  \dfrac{\sqrt{\left\vert p\right\vert ^{2}+m^{2}}}{\left(
\left\vert p\right\vert -r\right)  \left\vert p\right\vert }\left\langle
f_{+}\left(  p\right)  ,i_{1,2}^{3,3}\left(  \left\vert p\right\vert
,r,\theta,\varphi;R\right)  \right\rangle \right)  drdp,\right.
\]
with%
\[
i_{1,2}^{3,3}\left(  \left\vert p\right\vert ,r;R\right)  :=%
{\displaystyle\int\limits_{0}^{2\pi}}
{\displaystyle\int\limits_{0}^{\pi b}}
\left(  \frac{\sin\left(  R\left\vert p\right\vert \sqrt{2-2\cos\theta
}\right)  }{\left\vert p\right\vert \sqrt{2-2\cos\theta}}\left(
\partial_{\theta}g_{+}\left(  r\omega(\theta,\varphi)\right)  -\partial
_{\theta}g_{+}\left(  \left\vert p\right\vert \omega(\theta,\varphi)\right)
\right)  \right)  d\theta d\varphi,
\]
and%
\[
\left.
\begin{array}
[c]{c}%
I_{1,2}^{3,4}\left(  R\right)  :=-4\pi i%
{\displaystyle\int\limits_{\mathbb{R}^{3}}}
\left[
{\displaystyle\int\limits_{-\infty}^{\left\vert p\right\vert -1}}
+%
{\displaystyle\int\limits_{\left\vert p\right\vert +1}^{\infty}}
\right]  \frac{\left(  \sqrt{\left\vert p\right\vert ^{2}+m^{2}}+\sqrt
{r^{2}+m^{2}}\right)  }{r\left\vert p\right\vert \left(  \left\vert
p\right\vert -r\right)  \left(  \left\vert p\right\vert +r\right)  }\\
\times\left\langle f_{+}\left(  p\right)  ,%
{\displaystyle\int\limits_{0}^{2\pi}}
{\displaystyle\int\limits_{0}^{\pi}}
\frac{\sin\left(  R\sqrt{\left\vert p\right\vert ^{2}-2r\left\vert
p\right\vert \cos\theta+r^{2}}\right)  }{\sqrt{\left\vert p\right\vert
^{2}-2r\left\vert p\right\vert \cos\theta+r^{2}}}\partial_{\theta}g_{+}\left(
\omega(\theta,\varphi)\right)  d\theta d\varphi\right\rangle F\left(
r\right)  r^{2}drdp.
\end{array}
\right.
\]

Let us consider first $I_{1,2}^{3,1}\left(  R;0,\pi/4\right)  .$ Using that
$\sin x=\frac{e^{ix}-e^{-ix}}{2i}$ we have%
\[
\frac{\sin\left(  R\sqrt{\left\vert p\right\vert ^{2}-2r\left\vert
p\right\vert \cos\theta+r^{2}}\right)  }{\sqrt{\left\vert p\right\vert
^{2}-2r\left\vert p\right\vert \cos\theta+r^{2}}}\partial_{\theta}g_{+}\left(
r\omega(\theta,\varphi)\right)  =\frac{e^{iR\sqrt{\left\vert p\right\vert
^{2}-2r\left\vert p\right\vert \cos\theta+r^{2}}}-e^{-iR\sqrt{\left\vert
p\right\vert ^{2}-2r\left\vert p\right\vert \cos\theta+r^{2}}}}{2i\sqrt
{\left\vert p\right\vert ^{2}-2r\left\vert p\right\vert \cos\theta+r^{2}}%
}\partial_{\theta}g_{+}\left(  r\omega(\theta,\varphi)\right)  .
\]
Noting that
\[
\frac{e^{\pm iR\sqrt{\left\vert p\right\vert ^{2}-2r\left\vert p\right\vert
\cos\theta+r^{2}}}}{\sqrt{\left\vert p\right\vert ^{2}-2r\left\vert
p\right\vert \cos\theta+r^{2}}}=\frac{\partial_{\theta}\left(  \left(
\sin\theta\right)  e^{\pm iR\sqrt{\left\vert p\right\vert ^{2}-2r\left\vert
p\right\vert \cos\theta+r^{2}}}\right)  }{\left(  \cos\theta\right)
\sqrt{\left\vert p\right\vert ^{2}-2r\left\vert p\right\vert \cos\theta+r^{2}%
}\pm iRr\left\vert p\right\vert \sin^{2}\theta}%
\]
and integrating by parts we get
\[
\left.
\begin{array}
[c]{c}%
{\displaystyle\int\limits_{0}^{\pi/4}}
\frac{e^{\pm iR\sqrt{\left\vert p\right\vert ^{2}-2r\left\vert p\right\vert
\cos\theta+r^{2}}}}{\sqrt{\left\vert p\right\vert ^{2}-2r\left\vert
p\right\vert \cos\theta+r^{2}}}\partial_{\theta}g_{+}\left(  r\omega
(\theta,\varphi)\right)  d\theta=%
{\displaystyle\int\limits_{0}^{\pi/4}}
\frac{\partial_{\theta}\left(  \left(  \sin\theta\right)  e^{\pm
iR\sqrt{\left\vert p\right\vert ^{2}-2r\left\vert p\right\vert \cos
\theta+r^{2}}}\right)  }{\cos\theta\sqrt{\left\vert p\right\vert
^{2}-2r\left\vert p\right\vert \cos\theta+r^{2}}\pm iRr\left\vert p\right\vert
\sin^{2}\theta}\partial_{\theta}g_{+}\left(  r\omega(\theta,\varphi)\right)
d\theta\\
=\sqrt{2}\frac{e^{\pm iR\sqrt{\left\vert p\right\vert ^{2}-\sqrt{2}r\left\vert
p\right\vert +r^{2}}}}{\sqrt{2}\sqrt{\left\vert p\right\vert ^{2}-\sqrt
{2}r\left\vert p\right\vert +r^{2}}\pm iRr\left\vert p\right\vert }\left.
\partial_{\theta}g_{+}\left(  r\omega(\theta,\varphi)\right)  \right\vert
_{\theta=\pi/4}\\
-%
{\displaystyle\int\limits_{0}^{\pi/4}}
e^{\pm iR\sqrt{\left\vert p\right\vert ^{2}-2r\left\vert p\right\vert
\cos\theta+r^{2}}}\sin\theta\left(  \partial_{\theta}\frac{\partial_{\theta
}g_{+}\left(  r\omega(\theta,\varphi)\right)  }{\cos\theta\sqrt{\left\vert
p\right\vert ^{2}-2r\left\vert p\right\vert \cos\theta+r^{2}}\pm iRr\left\vert
p\right\vert \sin^{2}\theta}\right)  d\theta,
\end{array}
\right.
\]
and hence,%
\begin{equation}
\left.
\begin{array}
[c]{c}%
i_{1,2}^{3,1}\left(  \left\vert p\right\vert ,r;R;0,\pi/4\right)
=j_{1,2}^{3,1}\left(  \left\vert p\right\vert ,r;R\right)  +j_{1,2}%
^{3,1}\left(  \left\vert p\right\vert ,r;-R\right)  +l_{1,2}^{3,1}\left(
\left\vert p\right\vert ,r;R\right)  +l_{1,2}^{3,1}\left(  \left\vert
p\right\vert ,r;-R\right)  ,
\end{array}
\right.  \label{t107}%
\end{equation}
with%
\[
j_{1,2}^{3,1}\left(  \left\vert p\right\vert ,r;\pm R\right)  :=\pm
\dfrac{\sqrt{2}}{2i}\left(  \dfrac{e^{\pm iR\sqrt{\left\vert p\right\vert
^{2}-\sqrt{2}r\left\vert p\right\vert +r^{2}}}}{\sqrt{2}\sqrt{\left\vert
p\right\vert ^{2}-\sqrt{2}r\left\vert p\right\vert +r^{2}}\pm iRr\left\vert
p\right\vert }\right)
{\displaystyle\int\limits_{0}^{2\pi}}
\left.  \partial_{\theta}g_{+}\left(  r\omega(\theta,\varphi)\right)
\right\vert _{\theta=\pi/4}d\varphi
\]
and%
\begin{equation}
l_{1,2}^{3,1}\left(  \left\vert p\right\vert ,r;\pm R\right)  :=\mp\dfrac
{1}{2i}%
{\displaystyle\int\limits_{0}^{2\pi}}
{\displaystyle\int\limits_{0}^{\pi/4}}
e^{\pm iR\sqrt{\left\vert p\right\vert ^{2}-2r\left\vert p\right\vert
\cos\theta+r^{2}}}\sin\theta\left(  \partial_{\theta}\frac{\partial_{\theta
}g_{+}\left(  r\omega(\theta,\varphi)\right)  }{\cos\theta\sqrt{\left\vert
p\right\vert ^{2}-2r\left\vert p\right\vert \cos\theta+r^{2}}\pm iRr\left\vert
p\right\vert \sin^{2}\theta}\right)  d\theta d\varphi. \label{t130}%
\end{equation}
Since%
\[
\left.
\begin{array}
[c]{c}%
\left\vert
{\displaystyle\int\limits_{0}^{2\pi}}
\partial_{\theta}g_{+}\left(  r\omega(\theta,\varphi)\right)  d\varphi
\right\vert =\left\vert
{\displaystyle\int\limits_{0}^{2\pi}}
\partial_{\theta}g_{+}\left(  r\omega(\theta,\varphi)\right)  -\left.
\partial_{\theta}g_{+}\left(  r\omega(\theta,\varphi)\right)  \right\vert
_{\theta=0}d\varphi\right\vert \\
\leq C%
{\displaystyle\int\limits_{0}^{2\pi}}
\left\vert
{\displaystyle\int\limits_{0}^{\theta}}
\partial_{\theta}^{2}g_{+}\left(  r\omega(\theta,\varphi)\right)
d\theta\right\vert d\varphi\leq C%
{\displaystyle\int\limits_{0}^{2\pi}}
{\displaystyle\int\limits_{0}^{\pi}}
\left\vert \partial_{\theta}^{2}g_{+}\left(  r\omega(\theta,\varphi)\right)
\right\vert d\theta d\varphi
\end{array}
\right.
\]
and%
\[
\left\vert \dfrac{e^{\pm iR\sqrt{\left\vert p\right\vert ^{2}-\sqrt
{2}r\left\vert p\right\vert +r^{2}}}}{\sqrt{2}\sqrt{\left\vert p\right\vert
^{2}-\sqrt{2}r\left\vert p\right\vert +r^{2}}\pm iRr\left\vert p\right\vert
}\right\vert \leq\dfrac{1}{\left\vert \left\vert p\right\vert -r\right\vert
^{1-\beta}\left(  Rr\left\vert p\right\vert \right)  ^{\beta}},\text{ \ }%
\beta\leq1,
\]
we get%
\begin{equation}
\left\vert j_{1,2}^{3,1}\left(  \left\vert p\right\vert ,r;\pm R\right)
\right\vert \leq C\dfrac{1}{\left\vert \left\vert p\right\vert -r\right\vert
^{1-\beta}\left(  Rr\left\vert p\right\vert \right)  ^{\beta}}%
{\displaystyle\int\limits_{0}^{2\pi}}
{\displaystyle\int\limits_{0}^{\pi/4}}
\left\vert \partial_{\theta}^{2}g_{+}\left(  r\omega(\theta,\varphi)\right)
\right\vert d\theta d\varphi. \label{t108}%
\end{equation}
Note that the following estimates are true
\begin{equation}
\cos^{2}\theta\left(  \left\vert p\right\vert ^{2}-2r\left\vert p\right\vert
\cos\theta+r^{2}\right)  +\left(  Rr\left\vert p\right\vert \right)  ^{2}%
\sin^{4}\theta\geq C\left(  \left\vert p\right\vert -r\right)  ^{2-2\beta
}\left(  Rr\left\vert p\right\vert \right)  ^{2\beta}\sin^{4\beta}\theta,
\label{t113}%
\end{equation}
for $\beta\leq1$ and $\theta\in\lbrack0,\pi/4]$,
\begin{equation}
\left.  \dfrac{\left\vert r\right\vert \left\vert p\right\vert \sin\theta
\cos\theta}{\sqrt{\left\vert p\right\vert ^{2}-2r\left\vert p\right\vert
\cos\theta+r^{2}}}\leq C\sqrt{r\left\vert p\right\vert }\right.  \label{t114}%
\end{equation}
and%
\begin{equation}
\left.  \left\vert
{\displaystyle\int\limits_{0}^{2\pi}}
\partial_{\theta}g_{+}\left(  r\omega(\theta,\varphi)\right)  d\varphi
\right\vert =\left\vert
{\displaystyle\int\limits_{0}^{2\pi}}
\left(  \partial_{\theta}g_{+}\left(  r\omega(\theta,\varphi)\right)  -\left.
\partial_{\theta}g_{+}\left(  r\omega(\theta,\varphi)\right)  \right\vert
_{\theta=0}\right)  d\varphi\right\vert \leq C\left\vert \theta\right\vert
^{1/2}A\left(  g_{+};r\right)  ,\right.  \label{t115}%
\end{equation}
where%
\[
A\left(  g_{+};r\right)  :=\left(
{\displaystyle\int\limits_{0}^{2\pi}}
{\displaystyle\int\limits_{0}^{\pi}}
\left\vert \partial_{\theta}^{2}g_{+}\left(  r\omega(\theta,\varphi)\right)
\right\vert ^{2}d\theta d\varphi\right)  ^{1/2}.
\]
Using (\ref{t113}) with $\beta=3/4-\varepsilon,$ $\varepsilon>0,$ we get%
\begin{equation}
\left.
\begin{array}
[c]{c}%
\left\vert
{\displaystyle\int\limits_{0}^{2\pi}}
{\displaystyle\int\limits_{0}^{\pi/4}}
\sin\theta\left(  \dfrac{\partial_{\theta}^{2}g_{+}\left(  r\omega
(\theta,\varphi)\right)  }{\cos\theta\sqrt{\left\vert p\right\vert
^{2}-2r\left\vert p\right\vert \cos\theta+r^{2}}\pm iRr\left\vert p\right\vert
\sin^{2}\theta}\right)  d\theta d\varphi\right\vert \\
\leq\dfrac{C}{\left\vert \left\vert p\right\vert -r\right\vert ^{1-\beta
}\left(  Rr\left\vert p\right\vert \right)  ^{\beta}}%
{\displaystyle\int\limits_{0}^{2\pi}}
{\displaystyle\int\limits_{0}^{\pi/4}}
\left(  \dfrac{\left\vert \partial_{\theta}^{2}g_{+}\left(  r\omega
(\theta,\varphi)\right)  \right\vert }{\left(  \sin\theta\right)  ^{2\beta-1}%
}\right)  d\theta d\varphi\leq\dfrac{C}{\left\vert \left\vert p\right\vert
-r\right\vert ^{1/4+\varepsilon}\left(  Rr\left\vert p\right\vert \right)
^{3/4-\varepsilon}}A\left(  g_{+};r\right)  .
\end{array}
\right.  \label{t127}%
\end{equation}
It follows from (\ref{t113}), (\ref{t114}) and the estimate%
\[
\left.  \frac{\sqrt{\left\vert p\right\vert ^{2}-2r\left\vert p\right\vert
\cos\theta+r^{2}}\sin\theta}{\cos^{2}\theta\left(  \left\vert p\right\vert
^{2}-2r\left\vert p\right\vert \cos\theta+r^{2}\right)  +\left(  Rr\left\vert
p\right\vert \right)  ^{2}\sin^{4}\theta}\leq\frac{\sin\theta}{\left(
\cos^{2}\theta\left(  \left\vert p\right\vert ^{2}-2r\left\vert p\right\vert
\cos\theta+r^{2}\right)  +\left(  Rr\left\vert p\right\vert \right)  ^{2}%
\sin^{4}\theta\right)  ^{1/2}}\right.
\]
that%
\[
\left.
\begin{array}
[c]{c}%
\left\vert \partial_{\theta}\dfrac{1}{\cos\theta\sqrt{\left\vert p\right\vert
^{2}-2r\left\vert p\right\vert \cos\theta+r^{2}}+iRr\left\vert p\right\vert
\sin^{2}\theta}\right\vert \leq C\dfrac{\sqrt{\left\vert p\right\vert
^{2}-2r\left\vert p\right\vert \cos\theta+r^{2}}\sin\theta+\sqrt{r\left\vert
p\right\vert }+Rr\left\vert p\right\vert \sin\theta}{\cos^{2}\theta\left(
\left\vert p\right\vert ^{2}-2r\left\vert p\right\vert \cos\theta
+r^{2}\right)  +\left(  Rr\left\vert p\right\vert \right)  ^{2}\sin^{4}\theta
}\\
\leq C\dfrac{\sin\theta}{\left(  \left\vert p\right\vert -r\right)
^{1-\beta_{1}}\left(  Rr\left\vert p\right\vert \right)  ^{\beta_{1}}\left(
\sin\theta\right)  ^{2\beta_{1}}}+C\dfrac{\sqrt{r\left\vert p\right\vert }%
}{\left(  \left\vert p\right\vert -r\right)  ^{2-2\beta_{2}}\left(
Rr\left\vert p\right\vert \right)  ^{2\beta_{2}}\left(  \sin\theta\right)
^{4\beta_{2}}}\\
+C\dfrac{Rr\left\vert p\right\vert \sin\theta}{\left(  \left\vert p\right\vert
-r\right)  ^{2-2\beta_{3}}\left(  Rr\left\vert p\right\vert \right)
^{2\beta_{3}}\left(  \sin\theta\right)  ^{4\beta_{3}}},
\end{array}
\right.
\]
for $\beta_{j}\leq1,$ $j=1,2,3.$ Thus, using (\ref{t115}) we have
\begin{equation}
\left.
\begin{array}
[c]{c}%
\left\vert
{\displaystyle\int\limits_{0}^{2\pi}}
{\displaystyle\int\limits_{0}^{\pi/4}}
\sin\theta\left(  \partial_{\theta}\dfrac{1}{\cos\theta\sqrt{\left\vert
p\right\vert ^{2}-2r\left\vert p\right\vert \cos\theta+r^{2}}+iRr\left\vert
p\right\vert \sin^{2}\theta}\right)  \partial_{\theta}g_{+}\left(
r\omega(\theta,\varphi)\right)  d\theta d\varphi\right\vert \\
\leq C\left(  \left(  \dfrac{1}{\left\vert \left\vert p\right\vert
-r\right\vert ^{1-\beta_{1}}\left(  Rr\left\vert p\right\vert \right)
^{\beta_{1}}}%
{\displaystyle\int\limits_{0}^{\pi/4}}
\dfrac{d\theta}{\left(  \sin\theta\right)  ^{2\beta_{1}-5/2}}+\dfrac
{\sqrt{r\left\vert p\right\vert }}{\left\vert \left\vert p\right\vert
-r\right\vert ^{2-2\beta_{2}}\left(  Rr\left\vert p\right\vert \right)
^{2\beta_{2}}}%
{\displaystyle\int\limits_{0}^{\pi/4}}
\dfrac{d\theta}{\left(  \sin\theta\right)  ^{4\beta_{2}-3/2}}\right)  A\left(
g_{+};r\right)  \right) \\
+C\left(  \dfrac{1}{\left\vert \left\vert p\right\vert -r\right\vert
^{2-2\beta_{3}}\left(  Rr\left\vert p\right\vert \right)  ^{2\beta_{3}-1}}%
{\displaystyle\int\limits_{0}^{\pi/4}}
\dfrac{d\theta}{\left(  \sin\theta\right)  ^{4\beta_{3}-5/2}}\right)  A\left(
g_{+};r\right)  .
\end{array}
\right.  \label{t118}%
\end{equation}
Moreover, taking $\beta_{1}=3/4-\varepsilon,$ $\beta_{2}=5/8-\left(
1/2\right)  \varepsilon$ and $\beta_{3}=7/8-\left(  1/2\right)  \varepsilon,$
for $\varepsilon>0,\ $in (\ref{t118}), and using the resulting estimate,
together with (\ref{t127}) in (\ref{t130}), we get%
\begin{equation}
\left.  \left\vert l_{1,2}^{3,1}\left(  \left\vert p\right\vert ,r;\pm
R\right)  \right\vert \leq C\left(  \dfrac{1}{\left\vert \left\vert
p\right\vert -r\right\vert ^{1/4+\varepsilon}\left(  Rr\left\vert p\right\vert
\right)  ^{3/4-\varepsilon}}+\dfrac{\sqrt{r\left\vert p\right\vert }%
}{\left\vert \left\vert p\right\vert -r\right\vert ^{3/4+\varepsilon}\left(
Rr\left\vert p\right\vert \right)  ^{5/4-\varepsilon}}\right)  A\left(
g_{+};r\right)  .\right.  \label{t109}%
\end{equation}
Using (\ref{t108}), with $\beta=3/4-\varepsilon,$ and (\ref{t109}) in
(\ref{t107}) we get
\begin{equation}
\left.
\begin{array}
[c]{c}%
\left\vert i_{1,2}^{3,1}\left(  \left\vert p\right\vert ,r;R;0,\pi/4\right)
\right\vert \leq C\left(  \tfrac{1}{\left\vert \left\vert p\right\vert
-r\right\vert ^{1/4+\varepsilon}\left(  Rr\left\vert p\right\vert \right)
^{3/4-\varepsilon}}+\tfrac{\sqrt{r\left\vert p\right\vert }}{\left\vert
\left\vert p\right\vert -r\right\vert ^{3/4+\varepsilon}\left(  Rr\left\vert
p\right\vert \right)  ^{5/4-\varepsilon}}\right)  A\left(  g_{+};r\right) \\
\leq C\left(  \tfrac{1}{\left\vert \left\vert p\right\vert -r\right\vert
^{1/4+\varepsilon}\left(  Rr\left\vert p\right\vert \right)  ^{3/4-\varepsilon
}}+\tfrac{\sqrt{r\left\vert p\right\vert }}{\left\vert \left\vert p\right\vert
-r\right\vert ^{3/4+\varepsilon}\left(  Rr\left\vert p\right\vert \right)
^{5/4-\varepsilon}}\right)  \left(  B\left(  g_{+};r\right)  \right)  ^{1/2},
\end{array}
\right.  \label{t110}%
\end{equation}
with%
\begin{equation}
B\left(  g_{+};r\right)  :=%
{\displaystyle\int\limits_{0}^{2\pi}}
{\displaystyle\int\limits_{0}^{\pi}}
\left(  \left\vert \partial_{\theta}g_{+}\left(  r\omega(\theta,\varphi
)\right)  \right\vert ^{2}+\left\vert \partial_{\theta}^{2}g_{+}\left(
r\omega(\theta,\varphi)\right)  \right\vert ^{2}\right)  d\theta d\varphi.
\label{t131}%
\end{equation}
On the other hand, by (\ref{t125}),%
\begin{equation}
\left\vert i_{1,2}^{3,1}\left(  \left\vert p\right\vert ,r;R;0,\pi/4\right)
\right\vert \leq CR\left(
{\displaystyle\int\limits_{0}^{2\pi}}
{\displaystyle\int\limits_{0}^{\pi}}
\left\vert \partial_{\theta}g_{+}\left(  r\omega(\theta,\varphi)\right)
\right\vert ^{2}d\theta d\varphi\right)  ^{1/2}\leq CR\left(  B\left(
g_{+};r\right)  \right)  ^{1/2}. \label{t116}%
\end{equation}
Thus, observing that%
\[
\left.
\begin{array}
[c]{c}%
\left\vert \dfrac{\left(  \sqrt{\left\vert p\right\vert ^{2}+m^{2}}%
+\sqrt{r^{2}+m^{2}}\right)  }{\left\vert p\right\vert \left(  \left\vert
p\right\vert +r\right)  }F\left(  r\right)  r-\dfrac{\sqrt{\left\vert
p\right\vert ^{2}+m^{2}}}{\left\vert p\right\vert }\right\vert \leq\\
\leq\left\vert \left(  \dfrac{\left(  \sqrt{\left\vert p\right\vert ^{2}%
+m^{2}}+\sqrt{r^{2}+m^{2}}\right)  }{\left\vert p\right\vert \left(
\left\vert p\right\vert +r\right)  }-\dfrac{\sqrt{\left\vert p\right\vert
^{2}+m^{2}}}{\left\vert p\right\vert ^{2}}\right)  r\right\vert +\dfrac
{\sqrt{\left\vert p\right\vert ^{2}+m^{2}}}{\left\vert p\right\vert ^{2}%
}\left\vert \left(  F\left(  r\right)  r-F\left(  \left\vert p\right\vert
\right)  \left\vert p\right\vert \right)  \right\vert \leq C\left(
\left\langle p\right\rangle ^{-1}+\frac{1}{\left\vert p\right\vert ^{2}%
}\right)  \left\vert \left\vert p\right\vert -r\right\vert ,
\end{array}
\right.
\]
we obtain by using (\ref{t110}) and (\ref{t116})%
\begin{equation}
\left.
\begin{array}
[c]{c}%
{\displaystyle\int\limits_{\mathbb{R}^{3}}}
{\displaystyle\int\limits_{\left\vert p\right\vert -1}^{\left\vert
p\right\vert +1}}
\dfrac{1}{\left\vert \left\vert p\right\vert -r\right\vert }\left\vert
\tfrac{\left(  \sqrt{\left\vert p\right\vert ^{2}+m^{2}}+\sqrt{r^{2}+m^{2}%
}\right)  }{\left\vert p\right\vert \left(  \left\vert p\right\vert +r\right)
}F\left(  r\right)  r-\tfrac{\sqrt{\left\vert p\right\vert ^{2}+m^{2}}%
}{\left\vert p\right\vert }\right\vert \left\vert f_{+}\left(  p\right)
\right\vert \left\vert i_{1,2}^{3,1}\left(  \left\vert p\right\vert
,r;R;0,\pi/4\right)  \right\vert drdp\leq\\
\leq C%
{\displaystyle\int\limits_{\mathbb{R}^{3}}}
\left(  \left\langle p\right\rangle ^{-1}+\frac{1}{\left\vert p\right\vert
^{2}}\right)  \left\vert f_{+}\left(  p\right)  \right\vert
{\displaystyle\int\limits_{\left\vert p\right\vert -1}^{\left\vert
p\right\vert +1}}
\left\vert i_{1,2}^{3,1}\left(  \left\vert p\right\vert ,r;R;0,\pi/4\right)
\right\vert ^{\alpha}\left\vert i_{1,2}^{3,1}\left(  \left\vert p\right\vert
,r;R;0,\pi/4\right)  \right\vert ^{1-\alpha}drdp\\
\leq CR^{1-\alpha}%
{\displaystyle\int\limits_{\mathbb{R}^{3}}}
\left(  \left\langle p\right\rangle ^{-1}+\frac{1}{\left\vert p\right\vert
^{2}}\right)  \left(
{\displaystyle\int\limits_{\left\vert p\right\vert -1}^{\left\vert
p\right\vert +1}}
B\left(  g_{+};r\right)  dr\right)  ^{1/2}\\
\times\left\vert f_{+}\left(  p\right)  \right\vert \left(
{\displaystyle\int\limits_{\left\vert p\right\vert -1}^{\left\vert
p\right\vert +1}}
\left(  \tfrac{1}{\left\vert \left\vert p\right\vert -r\right\vert
^{1/4+\varepsilon}\left(  Rr\left\vert p\right\vert \right)  ^{3/4-\varepsilon
}}\right)  ^{2\alpha}+\left(  \tfrac{\sqrt{r\left\vert p\right\vert }%
}{\left\vert \left\vert p\right\vert -r\right\vert ^{3/4+\varepsilon}\left(
Rr\left\vert p\right\vert \right)  ^{5/4-\varepsilon}}\right)  ^{2\alpha
}dr\right)  ^{1/2},
\end{array}
\right.  \label{t117}%
\end{equation}
for $\alpha\leq1.$ Note that
\[
\left.
\begin{array}
[c]{c}%
{\displaystyle\int\limits_{\left\vert p\right\vert -1}^{\left\vert
p\right\vert +1}}
\left(  \tfrac{1}{\left\vert \left\vert p\right\vert -r\right\vert
^{1/4+\varepsilon}\left(  Rr\left\vert p\right\vert \right)  ^{3/4-\varepsilon
}}\right)  ^{2\alpha}+\left(  \tfrac{\sqrt{r\left\vert p\right\vert }%
}{\left\vert \left\vert p\right\vert -r\right\vert ^{3/4+\varepsilon}\left(
Rr\left\vert p\right\vert \right)  ^{5/4-\varepsilon}}\right)  ^{2\alpha}dr\\
\leq R^{-2\alpha\left(  5/4-\varepsilon\right)  }\left\vert p\right\vert
^{-2\alpha\left(  3/4-\varepsilon\right)  }%
{\displaystyle\int\limits_{\left\vert p\right\vert -1}^{\left\vert
p\right\vert +1}}
\left(  \tfrac{1}{\left\vert \left\vert p\right\vert -r\right\vert
^{3/4+\varepsilon}\left\vert r\right\vert ^{3/4-\varepsilon}}\right)
^{2\alpha}dr+\left(  R\left\vert p\right\vert \right)  ^{-2\alpha\left(
3/4-\varepsilon\right)  }%
{\displaystyle\int\limits_{\left\vert p\right\vert -1}^{\left\vert
p\right\vert +1}}
\left(  \tfrac{1}{\left\vert \left\vert p\right\vert -r\right\vert
^{1/4+\varepsilon}\left\vert r\right\vert ^{3/4-\varepsilon}}\right)
^{2\alpha}dr\\
\leq CR^{-2\alpha\left(  3/4-\varepsilon\right)  }\left(  1+\tfrac
{1}{\left\vert p\right\vert ^{4\alpha\left(  3/4+\varepsilon\right)  }%
}\right)  .
\end{array}
\right.
\]
Therefore, taking $\alpha=4/7+\varepsilon$ in (\ref{t117}) and using the last
inequality and (\ref{t131})$,$ we obtain%
\[
\left.
\begin{array}
[c]{c}%
{\displaystyle\int\limits_{\mathbb{R}^{3}}}
{\displaystyle\int\limits_{\left\vert p\right\vert -1}^{\left\vert
p\right\vert +1}}
\dfrac{1}{\left\vert \left\vert p\right\vert -r\right\vert }\left\vert
\tfrac{\left(  \sqrt{\left\vert p\right\vert ^{2}+m^{2}}+\sqrt{r^{2}+m^{2}%
}\right)  }{\left\vert p\right\vert \left(  \left\vert p\right\vert +r\right)
}F\left(  r\right)  r-\tfrac{\sqrt{\left\vert p\right\vert ^{2}+m^{2}}%
}{\left\vert p\right\vert }\right\vert \left\vert f_{+}\left(  p\right)
\right\vert \left\vert i_{1,2}^{3,1}\left(  \left\vert p\right\vert
,r;R;0,\pi/4\right)  \right\vert drdp\\
\leq\dfrac{C}{R^{\varepsilon_{1}}}\left(
{\displaystyle\int\limits_{\mathbb{R}^{3}}}
\left(  1+\frac{1}{\left\vert p\right\vert ^{2}}\right)  \left\vert
f_{+}\left(  p\right)  \right\vert \left(  1+\tfrac{1}{\left\vert p\right\vert
^{1-\varepsilon_{1}}}\right)  dp\right)  \left\Vert g_{+}\right\Vert
_{\mathcal{H}^{2}}\leq\dfrac{C}{R^{\varepsilon_{1}}}\left(  \left(  \left\Vert
f_{+}\right\Vert _{L^{\infty}}+\left\Vert f_{+}\right\Vert _{L^{2}}\right)
\left\Vert g_{+}\right\Vert _{\mathcal{H}^{2}}\right)  ,
\end{array}
\right.
\]
for $\varepsilon_{1}>0$ small enough, and furthermore, we conclude that%
\[
\lim_{R\rightarrow\infty}I_{1,2}^{3,1}\left(  R;0,\pi/4\right)  =0.
\]
Proceeding similarly, we obtain%
\[
\lim_{R\rightarrow\infty}I_{1,2}^{3,1}\left(  R;\pi/4,3\pi/4\right)  =0,
\]
\
\[
\lim_{R\rightarrow\infty}I_{1,2}^{3,1}\left(  R;3\pi/4,\pi\right)  =0,
\]%
\[
\lim_{R\rightarrow\infty}I_{1,2}^{3,2}\left(  R;0,\pi/4\right)  =0,
\]%
\[
\lim_{R\rightarrow\infty}I_{1,2}^{3,2}\left(  R;\pi/4,3\pi/4\right)  =0
\]
and\
\[
\lim_{R\rightarrow\infty}I_{1,2}^{3,2}\left(  R;3\pi/4,\pi\right)  =0.
\]
Actually, the proof for the parts $I_{1,2}^{3,1}\left(  R;\pi/4,3\pi/4\right)
$ and $I_{1,2}^{3,2}\left(  R;\pi/4,3\pi/4\right)  $ is more simple since
$\sin\theta\neq0$ on $[\pi/4,3\pi/4].$

Let us now consider\ $I_{1,2}^{3,2}\left(  R\right)  .$ Note that for
$\alpha\leq1$%
\[
\left.
\begin{array}
[c]{c}%
\left\vert
{\displaystyle\int\limits_{0}^{2\pi}}
\left(  \partial_{\theta}g_{+}\left(  r\omega(\theta,\varphi)\right)
-\partial_{\theta}g_{+}\left(  \left\vert p\right\vert \omega(\theta
,\varphi)\right)  \right)  d\varphi\right\vert \\
=\left\vert
{\displaystyle\int\limits_{0}^{2\pi}}
\left(  \partial_{\theta}g_{+}\left(  r\omega(\theta,\varphi)\right)
-\partial_{\theta}g_{+}\left(  \left\vert p\right\vert \omega(\theta
,\varphi)\right)  \right)  d\varphi\right\vert ^{\alpha}\left\vert
{\displaystyle\int\limits_{0}^{2\pi}}
\left(  \partial_{\theta}g_{+}\left(  r\omega(\theta,\varphi)\right)
-\partial_{\theta}g_{+}\left(  \left\vert p\right\vert \omega(\theta
,\varphi)\right)  \right)  d\varphi\right\vert ^{1-\alpha}\\
\leq C\left\vert \left\vert p\right\vert -r\right\vert ^{\alpha/2}\left\vert
\theta\right\vert ^{(1/2)\left(  1-\alpha\right)  }\left(
{\displaystyle\int\limits_{0}^{2\pi}}
{\displaystyle\int\limits_{\left\vert p\right\vert }^{r}}
\left\vert \partial_{r}\partial_{\theta}g_{+}\left(  r\omega(\theta
,\varphi)\right)  \right\vert ^{2}drd\varphi\right)  ^{\alpha/2}\left(
\left(  A\left(  g_{+};r\right)  \right)  ^{\frac{1-\alpha}{2}}+A\left(
g_{+};\left\vert p\right\vert \right)  ^{\frac{1-\alpha}{2}}\right)  .
\end{array}
\right.
\]
Then, by H\"{o}lder inequality we get for $\alpha<1/2$%
\[
\left.
\begin{array}
[c]{c}%
\left\vert i_{1,2}^{3,2}\left(  \left\vert p\right\vert ,r;R\right)
\right\vert \leq%
{\displaystyle\int\limits_{0}^{\pi}}
\left(  \frac{1}{\sqrt{2-2\cos\theta}}\left\vert
{\displaystyle\int\limits_{0}^{2\pi}}
\left(  \partial_{\theta}g_{+}\left(  r\omega(\theta,\varphi)\right)
-\partial_{\theta}g_{+}\left(  \left\vert p\right\vert \omega(\theta
,\varphi)\right)  \right)  d\varphi\right\vert \right)  d\theta\\
\leq C\left\vert \left\vert p\right\vert -r\right\vert ^{\alpha/2}\left(
\left(  A\left(  g_{+};r\right)  \right)  ^{\frac{1-\alpha}{2}}+A\left(
g_{+};\left\vert p\right\vert \right)  ^{\frac{1-\alpha}{2}}\right)
{\displaystyle\int\limits_{0}^{\pi/4}}
\left(  \left\vert \theta\right\vert ^{-1/2-(1/2)\alpha}\left(
{\displaystyle\int\limits_{0}^{2\pi}}
{\displaystyle\int\limits_{\left\vert p\right\vert }^{r}}
\left\vert \partial_{r}\partial_{\theta}g_{+}\left(  r\omega(\theta
,\varphi)\right)  \right\vert ^{2}drd\varphi\right)  ^{\alpha/2}\right)
d\theta\\
\leq C\left\vert \left\vert p\right\vert -r\right\vert ^{\alpha/2}\left(
\left(  A\left(  g_{+};r\right)  \right)  ^{\frac{1-\alpha}{2}}+A\left(
g_{+};\left\vert p\right\vert \right)  ^{\frac{1-\alpha}{2}}\right)  \left(
\left\Vert g_{+}\right\Vert _{\mathcal{H}^{2}}\right)  ^{\alpha/2}.
\end{array}
\right.
\]
Thus, we have%
\[
\left.
\begin{array}
[c]{c}%
{\displaystyle\int\limits_{\mathbb{R}^{3}}}
{\displaystyle\int\limits_{\left\vert p\right\vert -1}^{\left\vert
p\right\vert +1}}
\dfrac{\sqrt{\left\vert p\right\vert ^{2}+m^{2}}}{\left\vert \left\vert
p\right\vert -r\right\vert \left\vert p\right\vert }\left\vert f_{+}\left(
p\right)  \right\vert \left\vert i_{1,2}^{3,2}\left(  \left\vert p\right\vert
,r,\theta,\varphi;R\right)  \right\vert drdp\\
\leq C\left(  \left\Vert g_{+}\right\Vert _{\mathcal{H}^{2}}\right)
^{\alpha/2}%
{\displaystyle\int\limits_{-\infty}^{\infty}}
\left(  A\left(  g_{+};r\right)  \right)  ^{\frac{1-\alpha}{2}}%
{\displaystyle\int\limits_{\mathbb{R}^{3}}}
\dfrac{\sqrt{\left\vert p\right\vert ^{2}+m^{2}}}{\left\vert \left\vert
p\right\vert -r\right\vert ^{1-\alpha/2}\left\vert p\right\vert }\left\vert
f_{+}\left(  p\right)  \right\vert dpdr\\
+C\left(  \left\Vert g_{+}\right\Vert _{\mathcal{H}^{2}}\right)  ^{\alpha/2}%
{\displaystyle\int\limits_{\mathbb{R}^{3}}}
\dfrac{\sqrt{\left\vert p\right\vert ^{2}+m^{2}}}{\left\vert p\right\vert
}\left\vert f_{+}\left(  p\right)  \right\vert A\left(  g_{+};\left\vert
p\right\vert \right)  ^{\frac{1-\alpha}{2}}\left(
{\displaystyle\int\limits_{\left\vert p\right\vert -1}^{\left\vert
p\right\vert +1}}
\frac{1}{\left\vert \left\vert p\right\vert -r\right\vert ^{1-\alpha/2}%
}dr\right)  dp\\
\leq C\left\Vert f_{+}\right\Vert _{L_{1}^{\infty}}\left\Vert g_{+}\right\Vert
_{\mathcal{H}_{3/2+\varepsilon}^{2}}.
\end{array}
\right.
\]
Therefore, arguing as in (\ref{t2}) we get%
\[
\lim_{R\rightarrow\infty}I_{1,2}^{3,2}\left(  R\right)  =0.
\]

Finally, the proof for the term $I_{1,2}^{3,4}\left(  R\right)  $ is similar
to that of $I_{1,2}^{3,1}$ or $I_{1,2}^{3,2}.$ It results to be easier since
there is no irregular term $\left\vert \left\vert p\right\vert -r\right\vert
^{-1}.$
\end{proof}

We already obtained the asymptotics of $I_{1}\left(  R\right)  $ and
$I_{4}\left(  R\right)  $ as $R\rightarrow\infty.$ Let us now study the
behavior of $I_{2}\left(  R\right)  $ and $I_{3}\left(  R\right)  $ for big
$R.$ We prove the following

\begin{lemma}
\label{t48}Let $f,g\in\mathcal{H}_{2}^{3/2+\varepsilon}$, $\varepsilon>0.$
Then,
\[
\lim_{R\rightarrow\infty}I_{2}\left(  R\right)  =\lim_{R\rightarrow\infty
}I_{3}\left(  R\right)  =0.
\]

\end{lemma}

\begin{proof}
We consider the term $I_{2}\left(  R\right)  .$ The proof for $I_{3}\left(
R\right)  $ is analogous. Let $\varphi\in C_{0}^{\infty}\left(  \mathbb{R}%
\right)  $, such that $\varphi\left(  0\right)  =1.$ Then, it follows from the
proof of Lemma \ref{t42} (see relation (\ref{time11})) and the dominated
convergence theorem that \
\[
\left.
\begin{array}
[c]{c}%
I_{2}\left(  R\right)  =\lim\limits_{\varepsilon\rightarrow\infty}%
{\displaystyle\int\limits_{\mathbb{R}^{3}}}
{\displaystyle\int\limits_{\mathbb{R}^{3}}}
\left\langle \left(
{\displaystyle\int\limits_{0}^{\infty}}
e^{-i\left(  \sqrt{\left\vert p\right\vert ^{2}+m^{2}}+\sqrt{\left\vert
q\right\vert ^{2}+m^{2}}\right)  t}\varphi\left(  \varepsilon t\right)
dt\right)  f_{+}\left(  p\right)  ,\tilde{\zeta}_{R}\left(  p-q\right)
g_{-}\left(  q\right)  \right\rangle dqdp\\
=-i%
{\displaystyle\int\limits_{\mathbb{R}^{3}}}
{\displaystyle\int\limits_{\mathbb{R}^{3}}}
\left\langle f_{+}\left(  p\right)  ,\tilde{\zeta}_{R}\left(  p-q\right)
\frac{g_{-}\left(  q\right)  }{\sqrt{\left\vert p\right\vert ^{2}+m^{2}}%
+\sqrt{\left\vert q\right\vert ^{2}+m^{2}}}\right\rangle dqdp\\
-i\lim\limits_{\varepsilon\rightarrow\infty}\varepsilon%
{\displaystyle\int\limits_{0}^{\infty}}
{\displaystyle\int\limits_{\mathbb{R}^{3}}}
{\displaystyle\int\limits_{\mathbb{R}^{3}}}
\varphi^{\prime}\left(  \varepsilon t\right)  \left\langle \frac{e^{-i\left(
\sqrt{\left\vert p\right\vert ^{2}+m^{2}}+\sqrt{\left\vert q\right\vert
^{2}+m^{2}}\right)  t}}{\sqrt{\left\vert p\right\vert ^{2}+m^{2}}%
+\sqrt{\left\vert q\right\vert ^{2}+m^{2}}}f_{+}\left(  p\right)
,\tilde{\zeta}_{R}\left(  p-q\right)  g_{-}\left(  q\right)  \right\rangle
dqdpdt.
\end{array}
\right.
\]
Integrating by parts, as in expression (\ref{time10}), in the second integral
of the R.H.S. of the last relation, we show that the limit $-i\lim
\limits_{\varepsilon\rightarrow\infty}\int\limits_{0}^{\infty}\int
\limits_{\mathbb{R}^{3}}\int\limits_{\mathbb{R}^{3}}\varphi^{\prime}\left(
\varepsilon t\right)  \left\langle \frac{e^{-i\left(  \sqrt{\left\vert
p\right\vert ^{2}+m^{2}}+\sqrt{\left\vert q\right\vert ^{2}+m^{2}}\right)  t}%
}{\sqrt{\left\vert p\right\vert ^{2}+m^{2}}+\sqrt{\left\vert q\right\vert
^{2}+m^{2}}}f_{+}\left(  p\right)  ,\tilde{\zeta}_{R}\left(  p-q\right)
g_{-}\left(  q\right)  \right\rangle dqdpdt$ exists, and hence,%
\[
I_{2}\left(  R\right)  =-i\int\limits_{\mathbb{R}^{3}}\int\limits_{\mathbb{R}%
^{3}}\left\langle f_{+}\left(  p\right)  ,\tilde{\zeta}_{R}\left(  p-q\right)
\frac{g_{-}\left(  q\right)  }{\sqrt{\left\vert p\right\vert ^{2}+m^{2}}%
+\sqrt{\left\vert q\right\vert ^{2}+m^{2}}}\right\rangle dqdp.
\]
Passing to the spherical coordinate system, where the $z-$axis is directed
along the vector $p,$ we obtain
\[
\left.  I_{1,2}\left(  R\right)  =-i\int\limits_{\mathbb{R}^{3}}%
\int\limits_{0}^{2\pi}\int\limits_{0}^{\pi}\int\limits_{0}^{\infty
}\left\langle f_{+}\left(  p\right)  ,\tilde{\zeta}_{R}\left(  p-r\omega
(\theta,\varphi)\right)  g_{1}\left(  r\omega(\theta,\varphi);\left\vert
p\right\vert \right)  \right\rangle r^{2}\sin\theta drd\theta d\varphi
dp,\right.
\]
where $g_{1}\left(  r\omega(\theta,\varphi);\left\vert p\right\vert \right)
:=\frac{g_{-}\left(  r\omega(\theta,\varphi)\right)  }{\sqrt{\left\vert
p\right\vert ^{2}+m^{2}}+\sqrt{r^{2}+m^{2}}}$ and $\omega(\theta
,\varphi)=\left(  \cos\varphi\sin\theta,\sin\varphi\sin\theta,\cos
\theta\right)  .$ Using (\ref{t28}) with $g_{1}$ instead of $g_{0}$ we have%
\[
\left.
\begin{array}
[c]{c}%
{\displaystyle\int\limits_{0}^{\pi}}
\tilde{\zeta}_{R}\left(  p-r\omega(\theta,\varphi)\right)  g_{1}\left(
r\omega(\theta,\varphi);\left\vert p\right\vert \right)  \sin\theta d\theta\\
=-\frac{\sin R\left(  \left\vert p\right\vert +r\right)  }{r\left\vert
p\right\vert \left(  \left\vert p\right\vert +r\right)  }g_{1}\left(
-r\frac{p}{\left\vert p\right\vert };\left\vert p\right\vert \right)
+\frac{\sin R\left(  \left\vert p\right\vert -r\right)  }{r\left\vert
p\right\vert \left(  \left\vert p\right\vert -r\right)  }g_{1}\left(
r\frac{p}{\left\vert p\right\vert };\left\vert p\right\vert \right)  +\frac
{1}{r\left\vert p\right\vert }%
{\displaystyle\int\limits_{0}^{\pi}}
\frac{\sin R\sqrt{\left\vert p\right\vert ^{2}-2r\left\vert p\right\vert
\cos\theta+r^{2}}}{\sqrt{\left\vert p\right\vert ^{2}-2r\left\vert
p\right\vert \cos\theta+r^{2}}}\partial_{\theta}g_{1}\left(  r\omega
(\theta,\varphi);\left\vert p\right\vert \right)  d\theta,
\end{array}
\right.
\]
and hence,%
\begin{equation}
I_{2}\left(  R\right)  :=I_{2}^{1}\left(  R\right)  +I_{2}^{2}\left(
R\right)  +I_{2}^{3}\left(  R\right)  , \label{t50}%
\end{equation}
with%
\[
\left.  I_{2}^{1}\left(  R\right)  :=8\pi^{2}i\int\limits_{\mathbb{R}^{3}}%
\int\limits_{0}^{\infty}\left\langle f_{+}\left(  p\right)  ,\frac{\sin
R\left(  \left\vert p\right\vert +r\right)  }{\left\vert p\right\vert \left(
\left\vert p\right\vert +r\right)  }g_{1}\left(  -r\frac{p}{\left\vert
p\right\vert };\left\vert p\right\vert \right)  \right\rangle rdrdp,\right.
\]%
\[
\left.  I_{2}^{2}\left(  R\right)  :=-8\pi^{2}i\int\limits_{\mathbb{R}^{3}%
}\int\limits_{0}^{\infty}\left\langle f_{+}\left(  p\right)  ,\frac{\sin
R\left(  \left\vert p\right\vert -r\right)  }{\left\vert p\right\vert \left(
\left\vert p\right\vert -r\right)  }g_{1}\left(  r\frac{p}{\left\vert
p\right\vert };\left\vert p\right\vert \right)  \right\rangle rdrdp,\right.
\]
and%
\[
\left.  I_{2}^{3}\left(  R\right)  :=-4\pi i\int\limits_{\mathbb{R}^{3}}%
\int\limits_{0}^{\infty}\frac{\left\langle f_{+}\left(  p\right)  ,i_{2}%
^{3}\left(  r,p;R\right)  \right\rangle }{\left\vert p\right\vert \left(
\sqrt{\left\vert p\right\vert ^{2}+m^{2}}+\sqrt{r^{2}+m^{2}}\right)
}rdrdp,\right.
\]
where%
\[
i_{2}^{3}\left(  r,p;R\right)  :=\int\limits_{0}^{2\pi}\int\limits_{0}^{\pi
}\frac{\sin R\sqrt{\left\vert p\right\vert ^{2}-2r\left\vert p\right\vert
\cos\theta+r^{2}}}{\sqrt{\left\vert p\right\vert ^{2}-2r\left\vert
p\right\vert \cos\theta+r^{2}}}\partial_{\theta}g_{-}\left(  r\omega
(\theta,\varphi)\right)  d\theta d\varphi.
\]
Since
\[
\left.
\begin{array}
[c]{c}%
\int\limits_{\mathbb{R}^{3}}\int\limits_{0}^{\infty}\left\vert f_{+}\left(
p\right)  \right\vert \frac{\left\vert g_{1}\left(  -r\frac{p}{\left\vert
p\right\vert };\left\vert p\right\vert \right)  \right\vert }{\left\vert
p\right\vert \left(  \left\vert p\right\vert +r\right)  }rdrdp\leq
\int\limits_{0}^{\infty}\left(  \int\limits_{\mathbb{S}^{2}}\left\vert
f_{+}\left(  p\right)  \right\vert ^{2}d\omega\right)  ^{1/2}d\left\vert
p\right\vert \int\limits_{0}^{\infty}\left(  \int\limits_{\mathbb{S}^{2}}%
\frac{\left\vert g_{-}\left(  r\omega\right)  \right\vert ^{2}}{r^{2}+m^{2}%
}d\omega\right)  ^{1/2}rdr\\
\leq C\left(  \left\Vert f_{+}\right\Vert _{L^{\infty}\left(  \left\vert
p\right\vert \leq1\right)  }+\left\Vert f_{+}\right\Vert _{L^{2}\left(
\mathbb{R}^{3}\right)  }\right)  \left(  \left\Vert g_{-}\right\Vert
_{L^{\infty}\left(  \left\vert p\right\vert \leq1\right)  }+\left\Vert
g_{-}\right\Vert _{L^{2}\left(  \mathbb{R}^{3}\right)  }\right)  ,
\end{array}
\right.
\]
arguing as in (\ref{t2}) we get%
\begin{equation}
\lim_{R\rightarrow\infty}I_{2}^{1}\left(  R\right)  =0. \label{t51}%
\end{equation}
Similarly to Lemma \ref{t45} in the case of $I_{1,2}^{1}\left(  R\right)  $ we
prove that
\begin{equation}
\lim_{R\rightarrow\infty}I_{2}^{2}\left(  R\right)  =0. \label{t52}%
\end{equation}
Observing that $\int\limits_{0}^{2\pi}\left.  \partial_{\theta}g_{-}\left(
r\omega(\theta,\varphi)\right)  \right\vert _{\theta=0}d\varphi=0$ we have
\[
i_{2}^{3}\left(  r,p;R\right)  =\int\limits_{0}^{2\pi}\int\limits_{0}^{\pi
}\frac{\sin R\sqrt{\left\vert p\right\vert ^{2}-2r\left\vert p\right\vert
\cos\theta+r^{2}}}{\sqrt{\left\vert p\right\vert ^{2}-2r\left\vert
p\right\vert \cos\theta+r^{2}}}\left(  \partial_{\theta}g_{-}\left(
r\omega(\theta,\varphi)\right)  -\left.  \partial_{\theta}g_{-}\left(
r\omega(\theta,\varphi)\right)  \right\vert _{\theta=0}\right)  d\theta
d\varphi.
\]
Then,
\[
\left.  \left\vert i_{2}^{3}\left(  r,p;R\right)  \right\vert \leq\frac
{C}{\sqrt{r\left\vert p\right\vert }}\int\limits_{0}^{\pi}\frac{\theta^{1/2}%
}{\sqrt{1-\cos\theta}}d\theta\left(  \int\limits_{0}^{2\pi}\int\limits_{0}%
^{\pi}\left\vert \partial_{\theta}^{2}g_{-}\left(  r\omega(\theta
,\varphi)\right)  \right\vert ^{2}d\theta d\varphi\right)  ^{1/2}\leq\frac
{C}{\sqrt{r\left\vert p\right\vert }}\left(  \int\limits_{0}^{2\pi}%
\int\limits_{0}^{\pi}\left\vert \partial_{\theta}^{2}g_{-}\left(
r\omega(\theta,\varphi)\right)  \right\vert ^{2}d\theta d\varphi\right)
^{1/2},\right.
\]
and moreover,%
\[
\left.
\begin{array}
[c]{c}%
{\displaystyle\int\limits_{\mathbb{R}^{3}}}
{\displaystyle\int\limits_{0}^{\infty}}
\frac{\left\vert f_{+}\left(  p\right)  \right\vert \left\vert i_{2}%
^{3}\left(  r,p;R\right)  \right\vert }{\left\vert p\right\vert \left(
\sqrt{\left\vert p\right\vert ^{2}+m^{2}}+\sqrt{r^{2}+m^{2}}\right)
}rdrdp\leq C%
{\displaystyle\int\limits_{\mathbb{R}^{3}}}
\frac{\left\vert f_{+}\left(  p\right)  \right\vert }{\left\vert p\right\vert
^{3/2}}dp\left(
{\displaystyle\int\limits_{0}^{\infty}}
\frac{1}{\sqrt{r^{2}+m^{2}}}\left(
{\displaystyle\int\limits_{0}^{2\pi}}
{\displaystyle\int\limits_{0}^{\pi}}
\left\vert \partial_{\theta}\left(  \omega(\theta,\varphi)\cdot\nabla
g_{-}\left(  r\omega(\theta,\varphi)\right)  \right)  \right\vert ^{2}%
r^{2}d\theta d\varphi\right)  ^{1/2}dr\right) \\
\leq C\left(  \left\Vert f_{+}\right\Vert _{L^{\infty}\left(  \left\vert
p\right\vert \leq1\right)  }+\left\Vert f_{+}\right\Vert _{L_{1}^{2}\left(
\mathbb{R}^{3}\right)  }\right)  \left\Vert g_{+}\right\Vert _{\mathcal{H}%
_{1}^{2}}.
\end{array}
\right.
\]
Hence, arguing as in (\ref{t2}) we see that%
\begin{equation}
\lim_{R\rightarrow\infty}I_{2}^{3}\left(  R\right)  =0. \label{t53}%
\end{equation}
Taking the limit as $R\rightarrow\infty$ in (\ref{t50}) and using (\ref{t51}),
(\ref{t52}) and (\ref{t53}) we conclude that
\[
\lim_{R\rightarrow\infty}I_{2}\left(  R\right)  =0.
\]

\end{proof}

\section{Time delay.}

\subsection{Proof of Theorem \ref{T1}\label{proofT1}.}

We begin by presenting a result that allows us to express the time delay
$\delta\mathcal{T}\left(  f\right)  $ in terms of the scattering operator
$\mathbf{S}$ (see, for example, Proposition 7.22, page 365 of \cite{7}).

\begin{proposition}
\label{t89}Let $f$ be such that, for every fixed $R<\infty,$ each of the
functions $t\rightarrow\left\Vert \zeta\left(  \frac{\left\vert x\right\vert
}{R}\right)  e^{-iH_{0}t}f\right\Vert _{L^{2}\left(  \mathbb{R}^{3}\right)  }$
and $t\rightarrow\left\Vert \zeta\left(  \frac{\left\vert x\right\vert }%
{R}\right)  e^{-iH_{0}t}\mathbf{S}f\right\Vert _{L^{2}\left(  \mathbb{R}%
^{3}\right)  }$ belong to $L^{2}\left(  [0,\infty)\right)  .$ Assume that the
wave operators exist and are complete. Moreover, suppose that the function
$t\rightarrow\left\Vert \left(  W_{-}-e^{itH}e^{-itH_{0}}\right)  f\right\Vert
_{L^{2}\left(  \mathbb{R}^{3}\right)  }$ belongs to $L^{1}\left(
(-\infty,0]\right)  $ and that $t\rightarrow\left\Vert \left(  W_{+}%
-e^{itH}e^{-itH_{0}}\right)  \mathbf{S}f\right\Vert _{L^{2}\left(
\mathbb{R}^{3}\right)  }$ belongs to $L^{1}\left(  [0,\infty)\right)  .$
Then,
\begin{equation}
\delta\mathcal{T}\left(  f\right)  =\lim_{R\rightarrow\infty}\int
\limits_{0}^{\infty}\left\{  \left\Vert \zeta\left(  \frac{\left\vert
x\right\vert }{R}\right)  e^{-iH_{0}t}\mathbf{S}f\right\Vert ^{2}-\left\Vert
\zeta\left(  \frac{\left\vert x\right\vert }{R}\right)  e^{-iH_{0}%
t}f\right\Vert ^{2}\right\}  dt. \label{t103}%
\end{equation}

\end{proposition}

In particular, Proposition \ref{t89} relates the time delay with the
expectation values $I\left(  R\right)  ,$ defined by (\ref{t96}). To prove
Theorem \ref{T1} first we need to show that under the assumptions of Theorem
\ref{T1} relation (\ref{t103}) holds. We only have to verify that the
functions $t\rightarrow\left\Vert \zeta\left(  \frac{\left\vert x\right\vert
}{R}\right)  e^{-iH_{0}t}f\right\Vert _{L^{2}\left(  \mathbb{R}^{3}\right)  }$
and $t\rightarrow\left\Vert \zeta\left(  \frac{\left\vert x\right\vert }%
{R}\right)  e^{-iH_{0}t}\mathbf{S}f\right\Vert _{L^{2}\left(  \mathbb{R}%
^{3}\right)  }$ belong to $L^{2}\left(  [0,\infty)\right)  .$ Let us prove
that this is true if $f\in\mathcal{H}_{2}^{3/2+\varepsilon}\left(
\mathbb{R}^{3};\mathbb{C}^{4}\right)  $ and $\mathbf{S}f\in\mathcal{H}%
_{2}^{3/2+\varepsilon}\left(  \mathbb{R}^{3};\mathbb{C}^{4}\right)  .$ It
suffices to show these inclusions for $t\geq1.$ Observe that%
\begin{equation}
\left.
\begin{array}
[c]{c}%
\left\Vert \zeta\left(  \frac{\left\vert x\right\vert }{R}\right)
\mathbf{P}_{\pm}e^{-iH_{0}t}g\right\Vert _{L^{2}\left(  \mathbb{R}^{3}\right)
}=\left\Vert \mathcal{F}\left(  \zeta\left(  \frac{\left\vert x\right\vert
}{R}\right)  \mathbf{P}_{\pm}e^{-iH_{0}t}g\right)  \right\Vert _{L^{2}\left(
\mathbb{R}^{3}\right)  }\\
=\left(  2\pi\right)  ^{-3/2}\left\Vert R^{3}%
{\displaystyle\int\limits_{\mathbb{R}^{3}}}
\left(  \mathcal{F}\zeta\right)  \left(  R\left(  p-q\right)  \right)  e^{\mp
i\sqrt{\left\vert q\right\vert ^{2}+m^{2}}t}\left(  \mathcal{F}\mathbf{P}%
_{\pm}g\right)  \left(  q\right)  dq\right\Vert _{L^{2}\left(  \mathbb{R}%
^{3}\right)  }.
\end{array}
\right.  \label{t86}%
\end{equation}
For $g\in\mathcal{S}$, integrating by parts, we have
\[
\left.
\begin{array}
[c]{c}%
{\displaystyle\int\limits_{\mathbb{R}^{3}}}
\left(  \mathcal{F}\zeta\right)  \left(  R\left(  p-q\right)  \right)  e^{\mp
i\sqrt{\left\vert q\right\vert ^{2}+m^{2}}t}\left(  \mathcal{F}\mathbf{P}%
_{\pm}g\right)  \left(  q\right)  dq\\
=\mp\dfrac{i}{t}%
{\displaystyle\int\limits_{0}^{\infty}}
{\displaystyle\int\limits_{\mathbb{S}^{2}}}
e^{\mp i\sqrt{\left\vert q\right\vert ^{2}+m^{2}}t}\frac{q\cdot\nabla
}{\left\vert q\right\vert }\left(  \left(  \mathcal{F}\zeta\right)  \left(
R\left(  p-q\right)  \right)  \left\vert q\right\vert \sqrt{q^{2}+m^{2}%
}\left(  \mathcal{F}\mathbf{P}_{\pm}g\right)  \left(  q\right)  \right)
d\left\vert q\right\vert d\omega,
\end{array}
\right.
\]
with $q=\left\vert q\right\vert \omega,$ and then, by Young's inequality,
\begin{equation}
\left\Vert
{\displaystyle\int\limits_{\mathbb{R}^{3}}}
\left(  \mathcal{F}\zeta\right)  \left(  R\left(  p-q\right)  \right)  e^{\mp
i\sqrt{\left\vert q\right\vert ^{2}+m^{2}}t}\left(  \mathcal{F}\mathbf{P}%
_{\pm}g\right)  \left(  q\right)  dq\right\Vert _{L^{2}\left(  \mathbb{R}%
^{3}\right)  }\leq C\frac{1}{t}\left\Vert \mathcal{F}g\right\Vert
_{\mathcal{H}_{3/2+\varepsilon}^{2}}.\label{t87}%
\end{equation}
Moreover, by continuity, (\ref{t87}) holds for any $g\in\mathcal{H}%
_{2}^{3/2+\varepsilon}.$ Hence, it follows from (\ref{t86}) that
\begin{equation}
\left\Vert \zeta\left(  \frac{\left\vert x\right\vert }{R}\right)
e^{-iH_{0}t}g\right\Vert _{L^{2}\left(  \mathbb{R}^{3}\right)  }\leq C\frac
{1}{t}\left\Vert \mathcal{F}g\right\Vert _{\mathcal{H}_{3/2+\varepsilon}^{2}%
}.\label{t92}%
\end{equation}
Therefore, using (\ref{t92}) with $g=\mathbf{S}f$ and $g=f,$ we conclude that
the functions $t\rightarrow\left\Vert \zeta\left(  \frac{\left\vert
x\right\vert }{R}\right)  e^{-iH_{0}t}\mathbf{S}f\right\Vert _{L^{2}\left(
\mathbb{R}^{3}\right)  }$ and $t\rightarrow\left\Vert \zeta\left(
\frac{\left\vert x\right\vert }{R}\right)  e^{-iH_{0}t}f\right\Vert
_{L^{2}\left(  \mathbb{R}^{3}\right)  }$ belong to $L^{2}\left(
\mathbb{R}\right)  .$\ Then, the assumptions of Proposition \ref{t89} are
satisfied and (\ref{t103}) is valid.

Let us prove now the first assertion of Theorem \ref{T1}. From the unitarity
of the scattering operator $\mathbf{S}$, since $\mathcal{F}\mathbf{S}%
\mathcal{F}^{-1}$ commutes with $\frac{\sqrt{\left\vert p\right\vert
^{2}+m^{2}}}{\left\vert p\right\vert }$ and $P_{+}\left(  p\right)
+P_{-}\left(  p\right)  =I$, we get
\[
\left.
{\displaystyle\sum\limits_{\sigma=\pm}}
\left(  \frac{\sqrt{\left\vert p\right\vert ^{2}+m^{2}}}{\left\vert
p\right\vert }f_{\sigma}\left(  p\right)  ,f_{\sigma}\left(  p\right)
\right)  _{L^{2}\left(  \mathbb{R}^{3};\mathbb{C}^{4}\right)  }=%
{\displaystyle\sum\limits_{\sigma=\pm}}
\left(  \frac{\sqrt{\left\vert p\right\vert ^{2}+m^{2}}}{\left\vert
p\right\vert }\left(  \mathbf{S}f\right)  _{\sigma}\left(  p\right)  ,\left(
\mathbf{S}f\right)  _{\sigma}\left(  p\right)  \right)  _{L^{2}\left(
\mathbb{R}^{3};\mathbb{C}^{4}\right)  },\right.
\]
where $f_{\pm}\left(  p\right)  =P_{\pm}\left(  p\right)  \hat{f}\left(
p\right)  $ and $\left(  \mathbf{S}f\right)  _{\pm}\left(  p\right)  =P_{\pm
}\left(  p\right)  \widehat{\mathbf{S}f}\left(  p\right)  .$ Then, applying
Theorem \ref{t58} to the R.H.S. of (\ref{t103}) we have%
\begin{equation}
\left.
\begin{array}
[c]{c}%
\delta\mathcal{T}\left(  f\right)  =i%
{\displaystyle\int\limits_{\mathbb{R}^{3}}}
\left\langle \left(  \mathbf{S}f\right)  _{+}\left(  p\right)  ,\frac
{\sqrt{\left\vert p\right\vert ^{2}+m^{2}}}{2\left\vert p\right\vert ^{2}%
}\left(  \mathbf{S}f\right)  _{+}\left(  p\right)  +\frac{\sqrt{\left\vert
p\right\vert ^{2}+m^{2}}}{2\left\vert p\right\vert ^{2}}p\cdot\triangledown
\left(  \mathbf{S}f\right)  _{+}\left(  p\right)  +\frac{1}{2\left\vert
p\right\vert ^{2}}p\cdot\triangledown\left(  \sqrt{\left\vert p\right\vert
^{2}+m^{2}}\left(  \mathbf{S}f\right)  _{+}\left(  p\right)  \right)
\right\rangle dp\\
-i%
{\displaystyle\int\limits_{\mathbb{R}^{3}}}
\left\langle \left(  \mathbf{S}f\right)  _{-}\left(  p\right)  ,\frac
{\sqrt{\left\vert p\right\vert ^{2}+m^{2}}}{2\left\vert p\right\vert ^{2}%
}\left(  \mathbf{S}f\right)  _{-}\left(  p\right)  +\frac{\sqrt{\left\vert
p\right\vert ^{2}+m^{2}}}{2\left\vert p\right\vert ^{2}}p\cdot\triangledown
\left(  \mathbf{S}f\right)  _{-}\left(  p\right)  +\frac{1}{2\left\vert
p\right\vert ^{2}}p\cdot\triangledown\left(  \sqrt{\left\vert p\right\vert
^{2}+m^{2}}\left(  \mathbf{S}f\right)  _{-}\left(  p\right)  \right)
\right\rangle dp\\
-i%
{\displaystyle\int\limits_{\mathbb{R}^{3}}}
\left\langle f_{+}\left(  p\right)  ,\frac{\sqrt{\left\vert p\right\vert
^{2}+m^{2}}}{2\left\vert p\right\vert ^{2}}f_{+}\left(  p\right)  +\frac
{\sqrt{\left\vert p\right\vert ^{2}+m^{2}}}{2\left\vert p\right\vert ^{2}%
}p\cdot\triangledown\left(  f_{+}\left(  p\right)  \right)  +\frac
{1}{2\left\vert p\right\vert ^{2}}p\cdot\triangledown\left(  \sqrt{\left\vert
p\right\vert ^{2}+m^{2}}f_{+}\left(  p\right)  \right)  \right\rangle dp\\
+i%
{\displaystyle\int\limits_{\mathbb{R}^{3}}}
\left\langle f_{-}\left(  p\right)  ,\frac{\sqrt{\left\vert p\right\vert
^{2}+m^{2}}}{2\left\vert p\right\vert ^{2}}f_{-}\left(  p\right)  +\frac
{\sqrt{\left\vert p\right\vert ^{2}+m^{2}}}{2\left\vert p\right\vert ^{2}%
}p\cdot\triangledown\left(  f_{-}\left(  p\right)  \right)  +\frac
{1}{2\left\vert p\right\vert ^{2}}p\cdot\triangledown\left(  \sqrt{\left\vert
p\right\vert ^{2}+m^{2}}f_{-}\left(  p\right)  \right)  \right\rangle dp.
\end{array}
\right.  \label{t93}%
\end{equation}
Noting that
\[
\frac{\sqrt{\left\vert p\right\vert ^{2}+m^{2}}}{2\left\vert p\right\vert
^{2}}p\cdot\triangledown f=\frac{1}{2\left\vert p\right\vert ^{2}}%
p\cdot\triangledown\left(  \sqrt{\left\vert p\right\vert ^{2}+m^{2}}f\left(
p\right)  \right)  -\frac{1}{2\sqrt{\left\vert p\right\vert ^{2}+m^{2}}}f,
\]
we get%
\begin{equation}
\left.
\begin{array}
[c]{c}%
\dfrac{\sqrt{\left\vert p\right\vert ^{2}+m^{2}}}{2\left\vert p\right\vert
^{2}}f+\dfrac{\sqrt{\left\vert p\right\vert ^{2}+m^{2}}}{2\left\vert
p\right\vert ^{2}}p\cdot\triangledown f+\dfrac{1}{2\left\vert p\right\vert
^{2}}\left(  p\cdot\triangledown\right)  \left(  \sqrt{\left\vert p\right\vert
^{2}+m^{2}}f\left(  p\right)  \right) \\
=\frac{1}{2}\left(  \dfrac{2}{\left\vert p\right\vert ^{2}}p\cdot
\triangledown\left(  \sqrt{\left\vert p\right\vert ^{2}+m^{2}}f\left(
p\right)  \right)  +\dfrac{\sqrt{\left\vert p\right\vert ^{2}+m^{2}}%
}{\left\vert p\right\vert ^{2}}f\right)  -\dfrac{1}{2\sqrt{\left\vert
p\right\vert ^{2}+m^{2}}}f\\
=\dfrac{\sqrt{\left\vert p\right\vert ^{2}+m^{2}}}{2}\left(  2\dfrac
{p\cdot\triangledown}{\left\vert p\right\vert ^{2}}f+\dfrac{1}{\left\vert
p\right\vert ^{2}}f\right)  +\dfrac{1}{2\sqrt{\left\vert p\right\vert
^{2}+m^{2}}}f,
\end{array}
\right.  \label{t188}%
\end{equation}
Also, recalling (\ref{t164}) we have%
\begin{equation}
\mathcal{F}\mathbf{A}_{0}=\frac{i}{2}\left(  2\frac{p\cdot\nabla}{\left\vert
p\right\vert ^{2}}+\frac{1}{\left\vert p\right\vert ^{2}}\right)
\mathcal{F}\text{,} \label{t189}%
\end{equation}
Therefore, using (\ref{t188}) and (\ref{t189}) in (\ref{t93}) we get
(\ref{t102}).

We now prove the second affirmation\ of Theorem \ref{T1}. Let us find the
decomposition of the operator $\mathbf{T}$ in the spectral representation of
the operator $H_{0}.$ Passing to the spherical coordinate system in the
integrals in the R.H.S. of (\ref{t93}) and making the change $E=\sqrt
{r^{2}+m^{2}}$ in the first and third terms and $E=-\sqrt{r^{2}+m^{2}}$ in the
other two terms we get
\begin{equation}
\left.
\begin{array}
[c]{c}%
\delta\mathcal{T}\left(  f\right)  =\\
=i%
{\displaystyle\int\limits_{\left(  -\infty,-m\right)  \cup\left(
m,\infty\right)  }}
{\displaystyle\int\limits_{\mathbb{S}^{2}}}
\left\langle S\left(  E\right)  \Gamma_{0}\left(  E\right)  f,\frac
{E}{2\left(  E^{2}-m^{2}\right)  }S\left(  E\right)  \Gamma_{0}\left(
E\right)  f+\frac{1}{2E}S\left(  E\right)  \Gamma_{0}\left(  E\right)
f+\partial_{E}\left(  S\left(  E\right)  \Gamma_{0}\left(  E\right)  f\right)
\right\rangle d\omega dE\\
-i%
{\displaystyle\int\limits_{\left(  -\infty,-m\right)  \cup\left(
m,\infty\right)  }}
{\displaystyle\int\limits_{\mathbb{S}^{2}}}
\left\langle \Gamma_{0}\left(  E\right)  f,\frac{E}{2\left(  E^{2}%
-m^{2}\right)  }\Gamma_{0}\left(  E\right)  f+\frac{1}{2E}\Gamma_{0}\left(
E\right)  f+\partial_{E}\Gamma_{0}\left(  E\right)  f\right\rangle d\omega dE.
\end{array}
\right.  \label{t105}%
\end{equation}
Suppose that the scattering matrix $S\left(  E\right)  $ is continuously
differentiable with respect to $E$ on some open set $O\subset(-\infty
,-m)\cup(m,+\infty)\setminus\sigma_{p}\left(  H\right)  $ and $f\in\Phi\left(
O\right)  $, where $\Phi\left(  O\right)  $ is defined by (\ref{t190}). Then,
from the unitarity of the scattering matrix $S\left(  E\right)  $ it follows
that%
\begin{equation}
\left.  \delta\mathcal{T}\left(  f\right)  =\int\limits_{\left(
-\infty,-m\right)  \cup\left(  m,\infty\right)  }\int\limits_{\mathbb{S}^{2}%
}\left\langle \Gamma_{0}\left(  E\right)  f,T\left(  E\right)  \Gamma
_{0}\left(  E\right)  f\right\rangle d\omega dE,\right.  \label{t104}%
\end{equation}
where%
\[
T\left(  E\right)  =-iS\left(  E\right)  ^{\ast}\frac{d}{dE}S\left(  E\right)
.
\]
The operators $T\left(  E\right)  $ are bounded in $L^{2}\left(
\mathbb{S}^{2}\right)  .$ Relation (\ref{t104}) shows that the family
$\{T\left(  E\right)  \}_{E\in\left(  -\infty,-m\right)  \cup\left(
m,\infty\right)  }$ defines a decomposition $\mathbf{T}$ in the spectral
representation of $H_{0}$ for any $f\in\Phi\left(  O\right)  .$ That is,
\[
\left(  \mathcal{F}_{0}\mathbf{T}f\right)  \left(  E,\omega\right)  =T\left(
E\right)  \Gamma_{0}\left(  E\right)  f,\text{ \ }f\in\Phi\left(  O\right)  .
\]
Therefore, for any $f\in\Phi\left(  O\right)  ,$ the operator $\mathbf{T}$ is
the Eisenbud-Wigner time delay operator.

\begin{rem}\rm
We observe that the condition of strong differentiability of the scattering
matrix may be relaxed (see\ page 485 of \cite{9}). Actually, one can obtain
(\ref{t104})\ from (\ref{t105}) by only requesting strong continuity of
$S\left(  E\right)  ,$ but in this case the operator $\frac{d}{dE}S\left(
E\right)  $ may be unbounded in $L^{2}\left(  \mathbb{S}^{2}\right)  .$
\end{rem}

The proof of Theorem \ref{T1} is complete.

\subsection{Proof of Theorem \ref{T2}\label{proofT2}.}

As in the case of the Schr\"{o}dinger operator (\cite{40}), the proof of
Theorem \ref{T2} consists in showing that the assumptions of Theorem \ref{T2}
imply the ones of Theorem \ref{T1}. That is, we need to prove that for
$\mathbf{V}$ satisfying Condition \ref{C1} and for $f\in\mathcal{D}_{\tau},$
$\tau>2,$ the function $t\rightarrow\left\Vert \left(  W_{-}-e^{itH}%
e^{-itH_{0}}\right)  f\right\Vert _{L^{2}\left(  \mathbb{R}^{3}\right)  }$
belongs to $L^{1}\left(  (-\infty,0]\right)  $, $t\rightarrow\left\Vert
\left(  W_{+}-e^{itH}e^{-itH_{0}}\right)  \mathbf{S}f\right\Vert
_{L^{2}\left(  \mathbb{R}^{3}\right)  }$ belongs to $L^{1}\left(
[0,\infty)\right)  $ and $f,\mathbf{S}f\in\mathcal{H}_{2}^{3/2+\varepsilon
}\left(  \mathbb{R}^{3};\mathbb{C}^{4}\right)  ,$ and, moreover, the
scattering matrix $S\left(  E\right)  $ is strongly differentiable on some
open set $O\subset(-\infty,-m)\cup(m,+\infty)\setminus\sigma_{p}\left(
H\right)  ,$ containing the support of $\psi_{f},$ given by the definition of
$\mathcal{D}_{\tau}.$ First of all, from the definition of the set
$\mathcal{D}_{\tau}$ it follows that $f\in\mathcal{H}_{2}^{3/2+\varepsilon
}\left(  \mathbb{R}^{3};\mathbb{C}^{4}\right)  .$ Next we note that the proofs
of Lemmas 2 to 9 of \cite{40} remain true if we consider the Dirac operator
instead of the Schr\"{o}dinger operator. We only make two remarks. Firstly,
instead of relations (11) and (12) in Lemma 4 of \cite{40}, one has%
\[
\lbrack H,x_{k}]=-i\alpha_{k}%
\]
and
\[
\lbrack H,\left\vert x\right\vert ^{m}]=m\left\vert x\right\vert ^{m-1}%
\sum_{j=1}^{3}\left(  -i\alpha_{j}\right)  ,
\]
respectively. Then, the rest of the proof of Lemma 4, in the case of the Dirac
operator, is similar, taking $\varepsilon=0$ in all of the formulas. Secondly,
in Lemmas 5 to 8, in our case, we need that $\varphi\in C_{0}^{\infty}\left(
\mathbb{R}\right)  $ be equal to $0$ in some neighborhood of $0.$ By using
Lemmas 1-9 of \cite{40}, we see that Corollary of Proposition 2 of \cite{40}
also holds in the case of the Dirac equation. Thus, in the way analogous to
Proposition 3 of \cite{40} we attain the inclusions $\left\Vert \left(
W_{-}-e^{itH}e^{-itH_{0}}\right)  f\right\Vert _{L^{2}\left(  \mathbb{R}%
^{3}\right)  }\in L^{1}\left(  (-\infty,0]\right)  $ and $\left\Vert \left(
W_{+}-e^{itH}e^{-itH_{0}}\right)  \mathbf{S}f\right\Vert _{L^{2}\left(
\mathbb{R}^{3}\right)  }\in L^{1}\left(  [0,\infty)\right)  $. Observing that
in the case of the Dirac operator $[x^{2},H_{0}]=2i\alpha\cdot x,$%
\[
\left.  \lbrack x^{2},U_{t}^{\ast}]=-i\left(
{\displaystyle\int\limits_{0}^{-t}}
U_{t+\tau}^{\ast}[x^{2},H_{0}]U_{\tau}d\tau\right)  =2i\left(
{\displaystyle\int\limits_{0}^{-t}}
U_{t+\tau}^{\ast}\left(  -i\alpha\cdot x\right)  U_{\tau}d\tau\right)
=t^{2}U_{t}^{\ast}-2tU_{t}^{\ast}\left(  \alpha\cdot x\right)  \right.
\]
and
\[
\lbrack x^{2},U_{t}^{0}]=i\left(
{\displaystyle\int\limits_{0}^{-t}}
\left(  U_{\tau}^{0}\right)  ^{\ast}[x^{2},H_{0}]U_{\tau+t}^{0}d\tau\right)
=-2i\left(
{\displaystyle\int\limits_{0}^{-t}}
\left(  U_{\tau}^{0}\right)  ^{\ast}\left(  -i\alpha\cdot x\right)  U_{\tau
+t}^{0}d\tau\right)  =-t^{2}U_{t}^{0}+2t\left(  \alpha\cdot x\right)
U_{t}^{0},
\]
where $U_{t}^{0}:=e^{-itH_{0}}$ and $U_{t}:=e^{-itH},$ and proceeding
similarly to the proof of Proposition 4 of \cite{40} we obtain $\mathbf{S}%
f\in\mathcal{D}_{2}\subset\mathcal{H}_{2}^{3/2+\varepsilon}\left(
\mathbb{R}^{3};\mathbb{C}^{4}\right)  .$ Finally, the strong differentiability
of the scattering matrix $S\left(  E\right)  $ on $(-\infty,-m)\cup
(m,+\infty)\setminus\sigma_{p}\left(  H\right)  $ can be obtained from the
stationary formula (\ref{basicnotions14}) in the way analogous to the case of
the Schr\"{o}dinger operator (see Theorem 3.5 of \cite{15}) by using
(\ref{basicnotions27}), (\ref{basicnotions13}) and the result about the
differentiability of the resolvent for the Schr\"{o}dinger operator given in
Lemma 3.4 of \cite{15}. Theorem \ref{T2} is proved.

\subsection{Proof of Theorem \ref{T4}.\label{proofT4}}

Recall that for an operator $A$ of trace class $\operatorname*{Det}\left(
I+A\right)  $ is the determinant of $I+A$ (\cite{39}, \cite{36}). Suppose that
the potential $\mathbf{V}$ satisfies Condition \ref{C2}$.$ Then, it follows
from Theorem 4.5 of \cite{yafaev2005} that the operator $S\left(  E\right)
-I$ is of trace class, the SSF $\xi\left(  E;H,H_{0}\right)  $ exists and it
is related to the scattering matrix $S\left(  E\right)  $ by the Birman-Krein
formula%
\begin{equation}
\operatorname*{Det}S\left(  E\right)  =e^{-2\pi i\xi\left(  E;H,H_{0}\right)
}.\label{t213}%
\end{equation}
Observe now that Condition \ref{C2} implies Condition \ref{basicnotions26}.
Moreover, under Condition \ref{C2} there are no eigenvalues embedded in the
absolutely continuous spectrum of $H$ (see Remark \ref{Rem1}). Hence,
(\ref{basicnotions27}) and (\ref{basicnotions13}) hold for any $E\in\left(
-\infty,-m\right)  \cup\left(  m,\infty\right)  $. Then, using
(\ref{basicnotions14}), (\ref{basicnotions27}), (\ref{basicnotions13}) and the
result about the differentiability of the resolvent for the Schr\"{o}dinger
operator given in Lemma 3.4 of \cite{15}, similarly to the case of the
Schr\"{o}dinger operator (see Proposition 1.9, page 300 of \cite{36}), we
prove that under Condition \ref{C2} $S\left(  E\right)  $ is continuously differentiable in
the trace class. Therefore, as the scattering matrix is unitary, the operator
$S\left(  E\right)  ^{\ast}\frac{d}{dE}S\left(  E\right)  $ is of trace class
and the following relation is satisfied
\begin{equation}
\left(  \operatorname*{Det}S\left(  E\right)  \right)  ^{-1}\frac{d}%
{dE}\operatorname*{Det}S\left(  E\right)  =\operatorname*{Tr}\left(  S\left(
E\right)  ^{\ast}\frac{d}{dE}S\left(  E\right)  \right)  .\label{t184}%
\end{equation}
For the convenience of the reader we include the simple proof of (\ref{t184}).
Let $\{f_{n}\}$ be an orthonormal basis of $L^{2}\left(  \mathbb{S}%
^{2};\mathbb{C}^{4}\right)  .$ We consider the square matrix $\left\{  \left(
S\left(  E\right)  f_{n},f_{m}\right)  \right\}  ,$ where $n,m\leq N.$ Here
$\left(  \cdot,\cdot\right)  $ denotes the scalar product of $L^{2}\left(
\mathbb{S}^{2};\mathbb{C}^{4}\right)  .$ By the definition of the determinant%
\begin{equation}
\operatorname*{Det}S\left(  E\right)  =\lim_{N\rightarrow\infty}\det\left(
\left\{  \left(  S\left(  E\right)  f_{n},f_{m}\right)  \right\}  \right)
.\label{t185}%
\end{equation}
Since $S\left(  E\right)  -I$ as of trace class, the limit in the R.H.S. of
the last expression exists. Moreover, using Jacobi's formula we get%
\[
\left.  \dfrac{d}{dE}\det\left(  \left\{  \left(  S\left(  E\right)
f_{n},f_{m}\right)  \right\}  \right)  =\det\left(  \left\{  \left(  S\left(
E\right)  f_{n},f_{m}\right)  \right\}  \right)  \operatorname*{Tr}\left(
\left\{  \left(  S\left(  E\right)  f_{n},f_{m}\right)  \right\}
^{-1}\left\{  \left(  \dfrac{d}{dE}S\left(  E\right)  f_{n},f_{m}\right)
\right\}  \right)  .\right.
\]
Taking the basis $\{f_{n}\}$ that consists of eigenvalues of $S\left(
E\right)  $ we get%
\[
\left.  \frac{d}{dE}\det\left(  \left\{  \left(  S\left(  E\right)
f_{n},f_{m}\right)  \right\}  \right)  =\det\left(  \left\{  \left(  S\left(
E\right)  f_{n},f_{m}\right)  \right\}  \right)  \operatorname*{Tr}\left(
\left\{  \lambda_{n,m}\right\} ^{-1}\left\{  \left(  \dfrac{d}{dE}S\left(
E\right)  f_{n},f_{m}\right)  \right\}  \right)  ,\right.
\]
where $\lambda_{n,n}$ is the $n$-th eigenvalue of $S\left(  E\right)  $ and
$\lambda_{n,m}=0,$ for $n\neq m.$ Thus,
\[
\left.
\begin{array}
[c]{c}%
\dfrac{d}{dE}\det  \left(\left\{  \left(  S\left( E\right)  f_{n}%
,f_{m}\right)  \right\}  \right)  =\det\left(  \left\{  \left(  S\left(
E\right)  f_{n},f_{m}\right)  \right\}  \right)
{\displaystyle\sum\limits_{n=1}^{N}}
\left( \left\{  \lambda_{n,n}\right\} ^{-1}\, \dfrac{d}{dE}S\left(  E\right)  f_{n},f_{n}\right)  =\\
=\det\left(  \left\{  \left(  S\left(  E\right)  f_{n},f_{m}\right) \right\}
\right)
{\displaystyle\sum\limits_{n=1}^{N}}
\left(\left(S\left(  E\right)\right)^*   \dfrac{d}{dE}S\left(  E\right)  f_{n}, f_{n}\right),
\end{array}
\right.
\]
where we used that as $S\left(  E\right)  $ is unitary $\lambda_{n,n}%
^{-1}=\lambda_{n,n}^{\ast}$. Taking the limit, as $N\rightarrow\infty,$ we
arrive to
\begin{equation}
\left.  \lim_{N\rightarrow\infty}\dfrac{d}{dE}\det\left(  \left\{  \left(
S\left(  E\right)  f_{n},f_{m}\right)  \right\}  \right)  =\left(
\operatorname*{Det}S\left(  E\right)  \right)  \operatorname*{Tr}\left(
S\left(  E\right)  ^{\ast}\frac{d}{dE}S\left(  E\right)  \right)  .\right.
\label{t187}%
\end{equation}
By (\ref{t187})
\begin{equation}
\dfrac{d}{dE}\operatorname*{Det}\left(  S\left(  E\right)  \right)  =\left(
\operatorname*{Det}S\left(  E\right)  \right)  \operatorname*{Tr}\left(
S\left(  E\right)  ^{\ast}\frac{d}{dE}S\left(  E\right)  \right)
,\label{derivada}%
\end{equation}
with the derivative in the left-hand side in distribution sense, but as
$S\left(  E\right)  ^{\ast}\frac{d}{dE}S\left(  E\right)  $ is continuous in
the trace class, (\ref{derivada}) holds with $\dfrac{d}{dE}\operatorname*{Det}%
\left(  S\left(  E\right)  \right)  $ in classical sense, what proves \eqref{t184}. 
 On
the other hand, similarly to Theorem 1.20, page 351 of \cite{36} we show that
$\xi\left(  E;H,H_{0}\right)  $ is continuous for $E\in\left(  -\infty
,-m\right)  \cup\left(  m,\infty\right)  $.
 Hence, differentiating
(\ref{t213}) and using (\ref{t184}) we see that $\xi\left(  E;H,H_{0}\right)
$ is differentiable on $\left(  -\infty,-m\right)  \cup\left(  m,\infty
\right)  $ and, furthermore, we attain (\ref{t186}). Theorem \ref{T4} is proved.

\end{document}